\documentclass[12pt,reqno]{amsart}

\usepackage[margin=1.0in]{geometry}
\usepackage{graphicx}
\usepackage{amsmath}
\usepackage{amsfonts}
\usepackage{amssymb}
\usepackage{mathrsfs}
\usepackage[normalem]{ulem}
\usepackage{caption}
\usepackage{subcaption}
\usepackage{amsthm}
\usepackage{textcomp}
\usepackage{hyperref}
\usepackage{slashed}
\usepackage{mathtools}

\newtheorem*{proposition*}{Proposition}
\newtheorem*{theorem*}{Theorem}
\newtheorem*{conjecture*}{Conjecture}

\newtheorem*{claim*}{Claim}
\newtheorem*{lemma*}{Lemma}
\newtheorem*{corollary*}{Corollary}

\newtheorem{theorem}{Theorem}[section]
\newtheorem{thm}{Theorem}

\newtheorem{proposition}[theorem]{Proposition}

\newtheorem{lemma}[theorem]{Lemma}
\newtheorem{corollary}[theorem]{Corollary}

\newtheorem*{definition*}{Definition}
\newtheorem{definition}{Definition}[section]
\newtheorem*{assumption*}{Assumption}

\newtheorem*{remark*}{Remark}
\newtheorem{remark}{Remark}[section]

\newcommand{\R}{\mathbb{R}}
\newcommand{\s}{\mathbb{S}}

\newcommand{\Z}{\mathbb{Z}}
\newcommand{\N}{\mathbb{N}}

\DeclareMathOperator{\tr}{\textnormal{tr}}

\newcommand{\snabla}{\slashed{\nabla}}
\newcommand{\sg}{\slashed{g}}

\numberwithin{equation}{section}
\allowdisplaybreaks

\begin{document}
\title{Linear waves in the interior of extremal black holes II}
\author{Dejan Gajic}
\address{University of Cambridge, Department of Applied Mathematics and Theoretical Physics, Wilberforce Road, Cambridge CB3 0WA, United Kingdom}
\email{D.Gajic@damtp.cam.ac.uk}
\date{}
\begin{abstract}
We consider solutions to the linear wave equation in the interior region of extremal Kerr black holes. We show that axisymmetric solutions can be extended continuously beyond the Cauchy horizon and moreover that, if we assume suitably fast polynomial decay in time along the event horizon, their local energy is finite. We also extend these results to non-axisymmetric solutions on slowly rotating extremal Kerr--Newman black holes. These results are the analogues of results obtained in [D. Gajic, \emph{Linear waves in the interior of extremal black holes I}, arXiv:1509.06568] for extremal Reissner--Nordstr\"om and stand in stark contrast to previously established results for the subextremal case, where the local energy was shown to generically blow up at the Cauchy horizon.
\end{abstract}
\maketitle
\tableofcontents
\section{Introduction}
In the precursor \cite{gajic2015} of this paper, we established the following results for the linear wave equation,
\begin{equation}
\label{eq:waveqkerr}
\square_g\phi=0,
\end{equation}
in the black hole interior of extremal Reissner--Nordstr\"om spacetimes, which describe maximally charged $(e^2=M^2)$, stationary, spherically symmetric black holes (see \cite{Hawk1973} for an overview of the geometry of Reissner--Nordstr\"om spacetimes):
\begin{itemize}
\item[(A)] Uniform boundedness and extendibility of $\phi$ in $C^0$ across the Cauchy horizon (Theorem 1 of \cite{gajic2015}).
\item [(B)] Extendibility of $\phi$ in $H^1_{\textnormal{loc}}$ across the Cauchy horizon (Theorem 2 and 3 of \cite{gajic2015}).
\item [(C)] Extendibility of $\phi$ in $C^{0,\alpha}$, with $0<\alpha<1$, across the Cauchy horizon (Theorem 5 of \cite{gajic2015}).
\item [(D)] Extendibility of spherically symmetric $\phi$ in $C^1$ across the Cauchy horizon (Theorem 4 of \cite{gajic2015}).
\item [(E)] Extendibility of spherically symmetric $\phi$ in $C^2$ across the Cauchy horizon (Theorem 6 of \cite{gajic2015}).
\end{itemize}

For result (A), we considered Cauchy initial data for $\phi$ on an asymptotically flat spacelike hypersurface intersecting the event horizon, which decay suitably fast towards spacelike infinity. For results (B), (C) and (D) we imposed stronger decay estimates in affine time on $\phi$ and its tangential derivatives along the event horizon than those that had previously been established in \cite{Aretakis2011a,Aretakis2011} for $\phi$ arising from Cauchy data. The required decay estimates have since been obtained in \cite{Angelopoulos2015} for suitable Cauchy data. For result (E) we assumed more precise asymptotics of $\phi$ along the event horizon, which are motivated by the numerical results in \cite{Lucietti2013} and have not yet been shown to hold for $\phi$ arising from generic, suitably decaying Cauchy data in a mathematically rigorously setting.

In this paper, \textbf{we shall prove the analogues of \textnormal{(A)}, \textnormal{(B)} and \textnormal{(C)} for \underline{axisymmetric} solutions $\phi$ to \textnormal{(\ref{eq:waveqkerr})} in the black hole interior of extremal Kerr--Newman spacetimes}; see Theorem \ref{thm:linftyboundv1}--\ref{thm:c11boundphiv1} below. The Kerr--Newman spacetimes are a three-parameter family, characterised by a mass $M$, a rotation parameter $a$ and a charge $e$ \cite{Newman1965}. Extremal Kerr--Newman spacetimes constitute a two-parameter subfamily of spacetimes, satisfying the constraint $M^2=a^2+e^2$; they can be viewed as a continuous family that connects the extremal Reissner--Nordstr\"om solutions ($a^2=0$) to the extremal Kerr solutions ($a^2=M^2$). For an overview of the geometry of Kerr--Newman spacetimes, see \cite{Carter1968}.

In \cite{are6}, polynomial decay in affine time of axisymmetric $\phi$ and its tangential derivatives was shown to hold along the event horizon of extremal Kerr ($a^2=M^2$) for suitably decaying Cauchy initial data. To obtain the analogue of (A) for axisymmetric $\phi$ in the extremal Kerr interior, we will assume the decay rates that follow from \cite{are6}. For the analogues (A) in extremal Kerr--Newman spacetimes for which $a^2<M^2$ and moreover, for the analogues of (B) and (C) in \emph{any} extremal Kerr--Newman spacetime, we assume polynomial decay in time of $\phi$ along the event horizon that is \emph{conjectured}, but not yet \emph{proved}, to hold.

Note that the methods involved in proving results (D) and (E) rely fundamentally on the spherical symmetry of $\phi$ and the background spacetime. For this reason, they do not carry over to extremal Kerr--Newman.

In addition, \textbf{we will show that we can drop the axisymmetry assumption on $\phi$ and prove the analogues of \textnormal{(A)}, \textnormal{(B)} and \textnormal{(C)} in extremal Kerr--Newman spacetimes that are sufficiently close to extremal Reissner--Nordstr\"om}, i.e.\ with a sufficiently small rotation parameter $a$; see Theorem \ref{thm:linftyboundv2}--\ref{thm:c11boundphiv2} below. We refer to this subfamily of extremal Kerr--Newman as \emph{slowly rotating} extremal Kerr--Newman. We \emph{assume}, again, the decay for $\phi$ along the event horizon that is expected to hold for suitably decaying Cauchy initial data in this setting. This assumption is now also necessary for the analogue of (A), as the required polynomial decay has not yet been proved to hold for $\phi$ (without axisymmetry) along the event horizon of slowly rotating extremal Kerr--Newman.

The analogue of (A) has recently been obtained for the wave equation in \emph{sub}extremal Reissner--Nordstr\"om ($e^2<M^2$) \cite{Franzen2014} and subextremal Kerr ($a^2<M^2$) \cite{franz2} by Franzen (see also the results of Hintz \cite{Hintz2015} for the very slowly rotating case, where $a^2\ll M^2$), whereas the analogue of (B) has been shown to \emph{fail} in subextremal Reissner--Nordstr\"om for generic Cauchy data \cite{Luk2015} by Luk--Oh. See also related results concerning instabilities in subextremal Kerr \cite{Dafermos2015e,Luk2015a}.

The results of this paper are related to Christodoulou's formulation of the strong cosmic censorship conjecture \cite{Christodoulou2009}. For further background and motivation, see the introduction of \cite{gajic2015}.

\subsection{Linear waves in the exterior region of extremal Kerr}
We will review in this section several results for the wave equation (\ref{eq:waveqkerr}) in the exterior region of extremal Kerr.

Aretakis considered in \cite{are6} axisymmetric solutions $\phi$ to (\ref{eq:waveqkerr}) in the exterior region of extremal Kerr, arising from Cauchy data on a spacelike hypersurface $\Sigma$ intersecting the event horizon $\mathcal{H}^+$; see Figure \ref{fig:fullspacetime}. He established polynomial decay in time for $\phi$ everywhere in the exterior, including along $\mathcal{H}^+$. 

In \cite{are4}, he moreover proved the existence of conserved quantities, the \emph{Aretakis constants}, along $\mathcal{H}^+$ for solutions $\phi$ that need not be axisymmetric. If non-vanishing, these constants constitute an obstruction to the decay of either $\phi$ itself or its transversal derivative. Since axisymmetric solutions $\phi$ have been shown to decay along $\mathcal{H}^+$, this means that, generically, their transversal derivatives \emph{cannot} decay. Furthermore, higher-order transversal derivatives will generically blow up in infinite time along $\mathcal{H}^+$. These non-decay and blow-up results have been dubbed ``the Aretakis instability'' in the literature \cite{Lucietti2013}.

Aretakis extended the above results in \cite{are5} to show non-decay and blow-up of higher-order derivatives of $\phi$ even in the case of data with vanishing Aretakis constants. There is still no proof of pointwise and energy boundedness or decay for \uline{non}-axisymmetric $\phi$ in the exterior region of extremal Kerr (cf.\ a complete picture of the boundedness and decay properties of the linear wave equation in the full \emph{sub}extremal range of Kerr(--Newman) spacetimes has recently been obtained in \cite{Dafermos2014c,Civin2014}).

Lucietti--Reall generalised the Aretakis constants to higher-spin equations in extremal Kerr in \cite{lure1}. In particular, they showed that conserved quantities also form an obstruction to the decay of solutions to the Teukolsky equation, which governs the evolution of perturbations of certain components of the curvature tensor in the context of the linearised Einstein equations.

\subsection{Linear waves in the interior region of extremal Kerr--Newman}
\label{sec:roughversionsresults}
In this section, we will give an overview of the main theorems proved in this paper; we will state more detailed versions of the theorems in Section \ref{sec:mainresults}. In Section \ref{sec:geometry} we will give the precise definitions of the spacetime regions of interest in extremal Kerr--Newman that are mentioned in the paragraphs below and we will present the construction of double-null coordinates that cover these regions.

In this paper, we will restrict to a spacetime rectangle ${\mathcal{D}_{u_0,v_0}}$, which is a subset of $\mathcal{M}\cup \mathcal{CH}^+$, where $\mathcal{M}$ denotes the extremal Kerr--Newman manifold and $\mathcal{CH}^+$ is the inner horizon of extremal Kerr--Newman. We take ${\mathcal{D}_{u_0,v_0}}$ to be the intersection of the causal future of the event horizon segment $\mathcal{H}^+\cap \{v\geq v_0\}$ and the causal past of the inner horizon segment \\$\mathcal{CH}^+\cap\{u\leq u_0\}$, with respect to the manifold-with-boundary $\mathcal{M}\cup \mathcal{CH}^+$, where $v_0$ and $u_0$ are chosen suitably, such that restriction of ${\mathcal{D}_{u_0,v_0}}$ to the interior region is entirely contained within the domain of the $(u,v)$ Eddington--Finkelstein-type double-null coordinates; see Figure \ref{fig:fullspacetime}. Note that we have defined ${\mathcal{D}_{u_0,v_0}}$ to include a segment of $\mathcal{CH}^+$.

We can employ an ingoing null coordinate $U(u)$ in $\mathcal{M}\cap {\mathcal{D}_{u_0,v_0}}$, which can be extended across $\mathcal{H}^+$, and an outgoing null coordinate $\widetilde{V}(v)$, which can be extended beyond $\mathcal{CH}^+$, to express ${\mathcal{D}_{u_0,v_0}}$ as the following set:
\begin{equation*}
{\mathcal{D}_{u_0,v_0}}=\{0\leq U \leq U(u_0),\:\widetilde{V}(v_0)\leq \widetilde{V}\leq 0,\:(U,V)\neq (0,0)\},
\end{equation*}
where $U=0$ at at $\mathcal{H}^+$ and $\widetilde{V}=0$ at $\mathcal{CH}^+$.
\begin{figure}[h!]
\begin{center}
\includegraphics[width=2.3in]{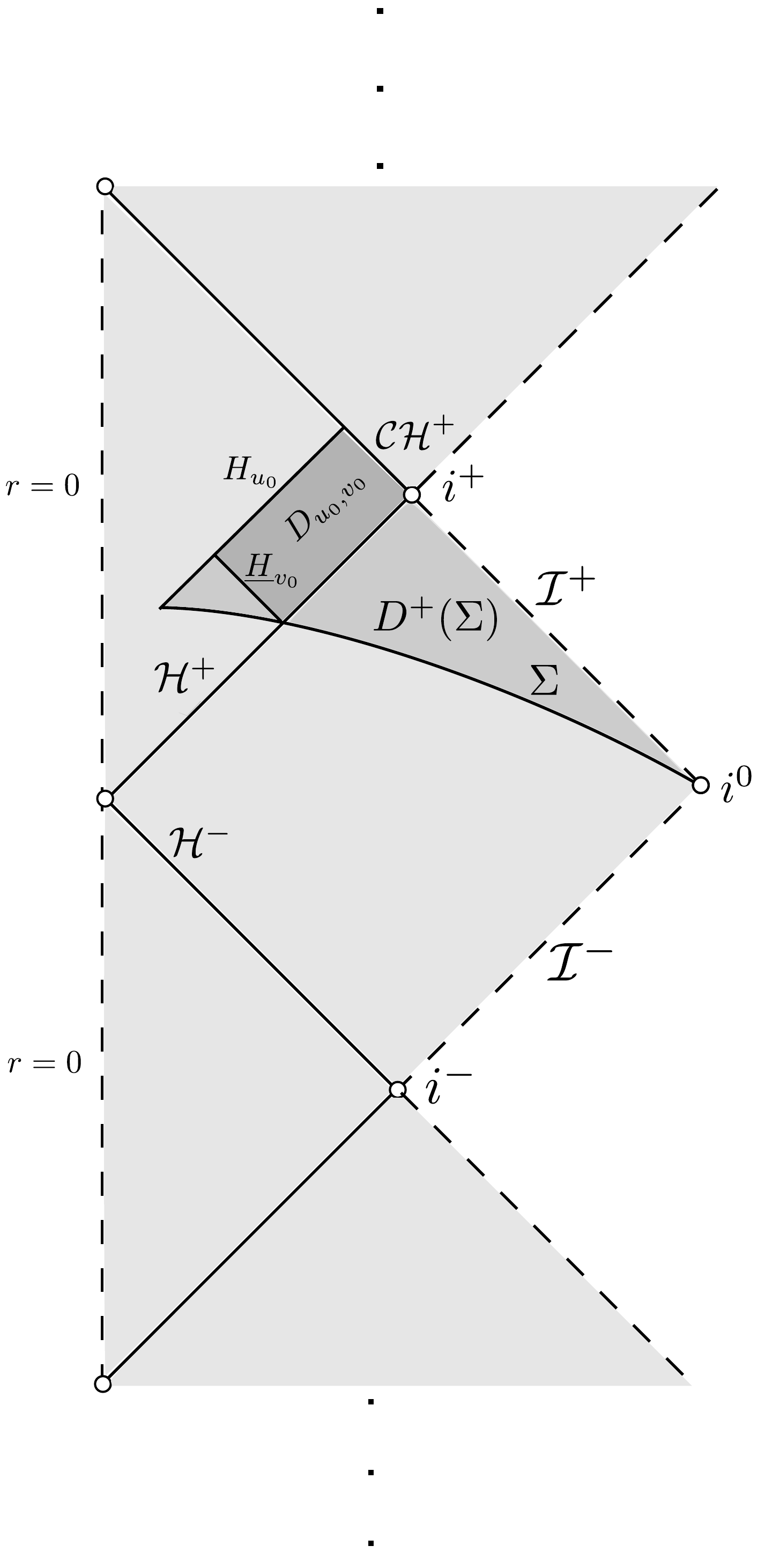}
\caption{\label{fig:fullspacetime} The Penrose diagram of the maximally analytic extension of extremal Kerr--Newman}
\end{center}
\end{figure}

We equip $\mathcal{H}^+$ and $\underline{H}_{v_0}$, the ingoing null hypersurface in ${\mathcal{D}_{u_0,v_0}}$ which is a subset of $\{v=v_0\}$, with characteristic initial data for the wave equation (\ref{eq:waveqkerr}). 

We can also consider solutions $\phi$ arising from Cauchy initial data for (\ref{eq:waveqkerr}) on an asymptotically flat spacelike hypersurface $\Sigma$ in extremal Kerr--Newman. It is natural to choose a hypersurface $\Sigma$ that has a non-trivial intersection with the black hole interior; see the discussion in \cite{Aretakis2010}. As a consequence of the geometry of the interior of extremal Kerr--Newman, $\Sigma$ must be incomplete; see Figure \ref{fig:fullspacetime}. We restrict to the future domain of dependence of $\Sigma$, which we denote by $D^+(\Sigma)$. The inner horizon $\mathcal{CH}^+$ contains part of the boundary of $D^+(\Sigma)$, so we will sometimes refer to $\mathcal{CH}^+$ as the \emph{Cauchy horizon}.

By choosing $u_0$ and $v_0$ appropriately, the rectangle ${\mathcal{D}_{u_0,v_0}}$ is a subset of $D^+(\Sigma)\cup \mathcal{CH}^+$. The characteristic data on $\mathcal{H}^+\cup \underline{H}_{v_0}$ can therefore be taken to be \emph{compatible} with the decay of $\phi$ and its tangential derivatives along $\mathcal{H}^+\cup \underline{H}_{v_0}$, that is expected to hold generically for $\phi$ arising from suitable Cauchy initial data on $\Sigma$.

In Section \ref{sec:subsubaxisymm} below we only consider axisymmetric solutions to (\ref{eq:waveqkerr}) on extremal Kerr-Newman. In Section~ \ref{sec:subsubslowlyrot} we instead consider slowly rotating extremal Kerr--Newman spacetimes, i.e.\ extremal Kerr--Newman spacetimes with a rotation parameter $0\leq a^2< a_c$, where $a_c=a_c(M)>0$ is chosen suitably small. In this case, we \emph{do not} impose axisymmetry on the solutions to (\ref{eq:waveqkerr}). See Section~ \ref{sec:mainsteps} for the definition of $a_c$.

\subsubsection{Axisymmetric solutions}
\label{sec:subsubaxisymm}
We first formulate an analogue of Theorem 1 of \cite{gajic2015} for axisymmetric solutions $\phi$ to (\ref{eq:waveqkerr}) in extremal Kerr--Newman, where $\phi$ arises from characteristic initial data along $\mathcal{H}^+\cup \uline{H}_{v_0}$.

\begin{thm}[$L^{\infty}$-boundedness and $C^0$-extendibility for axisymmetric solutions]
\label{thm:linftyboundv1}
Let $\phi$ be an axisymmetric solution to (\ref{eq:waveqkerr}) in extremal Kerr--Newman arising from suitably regular characteristic initial data on $\uline{H}_{v_0}\cup \mathcal{H}^+$, such that
for $\epsilon>0$ arbitrarily small,
\begin{align*}
\sum_{|k|\leq 2}\int_{S^2_{-\infty,v}}|\snabla^k\phi|^2<&\infty,\\
\sum_{0\leq j_1+j_2\leq 4}\int_{\mathcal{H}^+\cap\{v\geq v_0\}}v^{1+\epsilon}|\snabla^{j_1}\partial_v^{j_2+1}\phi|^2+|\snabla^{j_1+1} \partial_v^{j_2}\phi|^2<&\infty.
\end{align*}
where $\snabla$ denotes derivatives tangential to 2-spheres $S^2_{-\infty, v}$ that foliate $\mathcal{H}^+$. Then there exists a constant $C=C(M,a,\epsilon)>0$ and a natural norm $D_0>0$ on initial data for $\phi$, such that
\begin{equation*}
|\phi|\leq CD_0,
\end{equation*}
everywhere in $\mathcal{M}\cap \mathcal{D}_{u_0,v_0}$. Moreover, $\phi$ admits a $C^0$ extension beyond $\mathcal{CH}^+$.
\end{thm}

Theorem \ref{thm:pointwiseboundkn} and \ref{thm:C0extension} together form a more precise version of Theorem \ref{thm:linftyboundv1}.

In view of the decay results along $\mathcal{H}^+$ in \cite{are6} for $\phi$ in extremal Kerr arising from Cauchy initial data on a spacelike hypersurface $\Sigma$, we can reformulate Theorem \ref{thm:linftyboundv1} if we restrict to the subfamily of extremal Kerr spacetimes, where we consider suitably regular Cauchy data along $\Sigma$ in accordance with the results of \cite{are6}. 

\begin{thm}[$L^{\infty}$-boundedness and $C^0$-extendibility for axisymmetric solutions in extremal Kerr]
\label{thm:linftyboundv0}
Let $\phi$ be an axisymmetric solution to (\ref{eq:waveqkerr}) in extremal Kerr arising from suitably regular and decaying data on $\Sigma$. Then there exists a constant $C=C(M,\Sigma)>0$ and a natural norm $D_0>0$ on initial data for $\phi$, such that
\begin{equation*}
|\phi|\leq CD_0,
\end{equation*}
everywhere in $D^+(\Sigma)$. Moreover, $\phi$ admits a $C^0$ extension beyond $\mathcal{CH}^+$.
\end{thm}

Furthermore, we obtain the analogue of Theorem 2 of \cite{gajic2015}.
\begin{thm}[$H^1_{\textnormal{loc}}$-extendibility for axisymmetric solutions]
\label{thm:H1boundv1}
Let $\phi$ be an axisymmetric solution to (\ref{eq:waveqkerr}) in extremal Kerr--Newman arising from suitably regular characteristic initial data on $\uline{H}_{v_0}\cup \mathcal{H}^+$, such that
\begin{equation}
\label{eq:initialdecH+}
\int_{\mathcal{H}^+\cap\{v\geq v_0\}} v^2(\partial_v\phi)^2+|\snabla\phi|^2<\infty.
\end{equation}
Then $\phi$ admits a $H^1_{\textnormal{loc}}$ extension beyond $\mathcal{CH}^+$.
\end{thm}

In \cite{gajic2015} we reformulated Theorem 2 of \cite{gajic2015} by imposing Cauchy data on a spacelike hypersurface instead of characteristic data on the event horizon to obtain Theorem 3 of \cite{gajic2015}. We made use of the improved decay results along the event horizon of extremal Reissner--Nordstr\"om that were proved in \cite{Angelopoulos2015}. However, as the decay estimates for $\phi$ along $\mathcal{H}^+$ that are necessary for (\ref{eq:initialdecH+}) to hold have not yet been obtained for suitable data on $\Sigma$ in any extremal Kerr--Newman spacetime with $a\neq 0$, we do not reformulate Theorem \ref{thm:H1boundv1} above by imposing Cauchy data on $\Sigma$. Theorem \ref{thm:H1boundv1} follows from Theorem \ref{thm:eestaxisymm} after applying the estimates (\ref{est:fundthmcalcCS}) and (\ref{est:energytoH1}).

We can further conclude that $\phi$ can be extended beyond $\mathcal{CH}^+$ in the H\"older space $C^{0,\alpha}$ with $\alpha<1$. This result is the analogue of Theorem 5 of \cite{gajic2015}.

\begin{thm}[$C^{0,\alpha}$-extendibility of axisymmetric solutions]
\label{thm:c11boundphiv1}
Let $\alpha<1$. Let $\phi$ be an axisymmetric solution to (\ref{eq:waveqkerr}) in extremal Kerr--Newman arising from suitably regular and decaying characteristic initial data on $\uline{H}_{v_0}\cup \mathcal{H}^+$. Then $\phi$ admits a $C^{0,\alpha}$ extension beyond $\mathcal{CH}^+$.
\end{thm}

The precise necessary initial decay requirements along $\mathcal{H}^+$ appear in Theorem \ref{thm:alphaholdercont}.

\subsubsection{Slowly rotating extremal Kerr--Newman}
\label{sec:subsubslowlyrot}
We now restrict to the slowly rotating subfamily of extremal Kerr--Newman spacetimes, satisfying $0 \leq |a|<a_c$, where $a_c$ is the parameter described above. In particular, this subfamily excludes extremal Kerr. We will state analogues of the results from Section~ \ref{sec:subsubaxisymm} in slowly rotating extremal Kerr--Newman \underline{without} the restriction to axisymmetric solutions of (\ref{eq:waveqkerr}).

In slowly rotating extremal Kerr--Newman we can obtain $L^{\infty}$-boundedness and $C^0$-ex\-ten\-di\-bi\-li\-ty without an axisymmetry assumption on $\phi$.

\begin{thm}[$L^{\infty}$-boundedness and $C^0$-extendibility in slowly rotating extremal Kerr--New\-man]
\label{thm:linftyboundv2}
Let $\phi$ be a solution to (\ref{eq:waveqkerr}) in extremal Kerr--Newman arising from suitably regular and decaying characteristic data on $\uline{H}_{v_0}\cup \mathcal{H}^+$. Then there exists a constant $$C=~C(M,a,\Sigma)>0$$ and a natural norm $D_0>0$ that on initial data for $\phi$, such that
\begin{equation*}
|\phi|\leq CD_0,
\end{equation*}
everywhere in ${\mathcal{D}_{u_0,v_0}}$. Moreover, $\phi$ admits a $C^0$ extension beyond $\mathcal{CH}^+$.
\end{thm}
See Theorem \ref{thm:pointwiseboundkn} and Theorem \ref{thm:C0extension} in Section~ \ref{sec:mainresults} for precise requirements for the initial data on $\uline{H}_{v_0}\cup \mathcal{H}^+$. As there are presently no decay results available for non-axisymmetric solutions in the exterior region of slowly rotating extremal Kerr--Newman, we do not reformulate Theorem \ref{thm:linftyboundv2} by imposing Cauchy data on $\Sigma$ instead of characteristic data on $\uline{H}_{v_0}\cup \mathcal{H}^+$.

We also obtain an analogue of Theorem \ref{thm:H1boundv1} for $\phi$, without the assumption of axisymmetry, in slowly rotating extremal Kerr--Newman.
\begin{thm}[$H^1_{\textnormal{loc}}$-extendibility in slowly rotating Kerr--Newman]
\label{thm:H1boundv2}
Let $\phi$ be a solution to (\ref{eq:waveqkerr}) in extremal Kerr--Newman arising from suitably regular and decaying characteristic data on $\uline{H}_{v_0}\cup \mathcal{H}^+$. Then $\phi$ admits a $H^1_{\textnormal{loc}}$ extension beyond $\mathcal{CH}^+$.
\end{thm}
Here, we require decay of more derivatives in the initial data along the event horizon, compared to Theorem \ref{thm:H1boundv1}. See Theorem \ref{thm:pointwiseboundkn} and Theorem \ref{thm:improvedestoutgoingenergy} in Section~ \ref{sec:mainresults} for the precise decay rates. Theorem \ref{thm:H1boundv2} follows from Theorem \ref{thm:eestaxisymm}, \ref{thm:improvedestoutgoingenergy} and \ref{thm:pointwiseboundkn}, together with the estimate (\ref{est:energytoH1}).

Finally, we obtain an analogue of Theorem \ref{thm:c11boundphiv1} without the assumption of axisymmetry for $\phi$ in slowly rotating extremal Kerr--Newman.
\begin{thm}[$C^{0,\alpha}$-extendibility of $\phi$ in slowly-rotating extremal Kerr--Newman]
\label{thm:c11boundphiv2}
Let $\alpha<1$. Let $\phi$ be a solution to (\ref{eq:waveqkerr}) in extremal Kerr--Newman arising from suitably regular and decaying characteristic data on $\uline{H}_{v_0}\cup \mathcal{H}^+$. Then $\phi$ admits a $C^{0,\alpha}$ extension beyond $\mathcal{CH}^+$.
\end{thm}
The precise necessary initial decay requirements along $\mathcal{H}^+$ appear in Theorem \ref{thm:alphaholdercont}.

\subsection{Main ideas in the proofs of Theorem \ref{thm:linftyboundv1}--\ref{thm:c11boundphiv2}}
\label{sec:mainsteps}
In this section, we will outline the main steps in the proofs of Theorem \ref{thm:linftyboundv1}--\ref{thm:c11boundphiv2}. We will restrict to the region ${\mathcal{D}_{u_0,v_0}}$ in extremal Kerr--Newman with $0\leq |a|\leq M$ by default, unless specifically mentioned otherwise, and consider appropriate characteristic initial data for $\phi$ on $\underline{H}_{v_0}\cup \mathcal{H}^+$. We will highlight new difficulties that arise in extremal Kerr--Newman when $a\neq 0$, compared to extremal Reissner--Nordstr\"om (where $a=0$), which was treated in \cite{gajic2015}.

\subsubsection{Part 0: Constructing a double-null foliation}Before doing any estimates involving the wave equation, we first construct a suitable double-full foliation of the interior region of extremal Kerr--Newman. As Kerr--Newman spacetimes with $a\neq 0$ are not spherically symmetric, in contrast with Reissner--Nordstr\"om spacetimes, the existence of global double-null coordinates in the interior region is not immediate. In \cite{Gajic2015b}, a suitable global double-null foliation of extremal Kerr--Newman is constructed, which covers both the exterior and interior regions, following ideas of \cite{Pretorius1998}. We will use the results of \cite{Gajic2015b} here.

\subsubsection{Part 1: Vector field multipliers and energy estimates (Theorem \ref{thm:H1boundv1})} We obtain uniform bounds on weighted $L^2$ norms of $\phi$ along null hypersurfaces by means of energy estimates. Energy estimates are derived by using the vector field method; see for example \cite{Klainerman2010} for a general overview and the discussion in Section~ \ref{sec:vfmethod} for further particulars. Energy estimates for axisymmetric $\phi$ are obtained very similarly to the energy estimates in extremal Reissner--Nordstr\"om in \cite{gajic2015}; we use the following vector field multiplier:
\begin{equation*}
N_{p,q}=u^p\partial_u+v^q\partial_v,
\end{equation*}
with $p=q=2$, where $u$ and $v$ are double-null coordinates obtained in Part 0 that are akin to the Eddington--Finkelstein double-null coordinates in extremal Reissner--Nordstr\"om. See Section~ \ref{sec:estimatesmetricomp} for an overview of the construction and main properties of the Eddington--Finkelstein-type double-null coordinates $u$ and $v$ in extremal Kerr--Newman, and see Section~ \ref{sec:vfmethod} for more details regarding $N_{p,q}$.

As in extremal Reissner--Nordstr\"om, the energy estimates rely crucially on the following polynomial decay rate of the $g_{uv}$ component of the metric in Eddington--Finkelstein-type double-null coordinates:
\begin{equation*}
g_{uv}\sim (v+|u|)^{-2};
\end{equation*}
see Section~ \ref{sec:estimatesmetricconn} for the corresponding estimates. The above bounds play an important role in the proof of Theorem \ref{thm:H1boundv1}; see Section~ \ref{sec:energyestimatesaxisymm}.

If we drop the axisymmetry assumption on $\phi$, we have to take into account additional error terms in the energy estimates; most notably, extra error terms arise that involve the non-vanishing \emph{torsion} of the double-null foliation, denoted by $\zeta$. The torsion can be expressed as a commutator,
\begin{equation*}
\zeta=\frac{1}{4}\Omega^{-2}[L,\underline{L}],
\end{equation*}
where $L$ and $\underline{L}$ are vector fields that are tangent to null generators of the outgoing and ingoing null hypersurfaces, respectively. See Section~ \ref{sec:estimatesmetricconn}. 

In the $a=0$ case, $L$ and $\underline{L}$ are coordinate vector fields, so they commute, and $\zeta$ vanishes everywhere. If $a\neq 0$, $\zeta$ does not vanish. It turns out, however, that axisymmetric $\phi$ still satisfy $\zeta(\phi)=0$ if $a\neq 0$, so the error terms involving $\zeta$ do not form an obstruction for axisymmetric $\phi$.

In the case of non-axisymmetric $\phi$, we can estimate the error terms involving $\zeta$ by invoking the \emph{Hawking vector field}, which we denote by $H$. In Eddington--Finkelstein-type $(u,v)$ coordinates, we can express
\begin{equation*}
H=\frac{1}{2}(\partial_u+\partial_v).
\end{equation*}

The vector field $H$ is Killing and null along $\mathcal{H}^+$ and $\mathcal{CH}^+$. Moreover, it extends as a timelike vector field near the horizons, if $|a|<a_c<M$; in fact, it is precisely the requirement of a timelike $H$ near the horizons that determines $a_c$. See Section~ \ref{sec:hawkvf} for more details. The timelike character of $H$ implies that the error terms appearing in energy estimates with respect to the weighted vector field
\begin{equation*}
Y_p=|u|^pH,
\end{equation*}
with $p\geq 0$, have a good sign. We can still employ $N_{p,q}$ in suitable neighbourhoods of the $\mathcal{H}^+$ and $\mathcal{CH}^+$ that have a finite spacetime volume, but we use $Y_p$ in their complement in ${\mathcal{D}_{u_0,v_0}}$. This method gives rise to an $\epsilon$-loss in the exponents of the $u$ and $v$ weights that appear in the energies, compared to the case of axisymmetric $\phi$, which prevents us from directly inferring Theorem \ref{thm:H1boundv2}. See Section~ \ref{sec:energyestimatesslowlyrot}.

\subsubsection{Part 2: Commutation vector fields and pointwise estimates  (Theorem \ref{thm:linftyboundv1} and \ref{thm:linftyboundv2})}
We subsequently use the uniformly bounded weighted $L^2$ norms from Part 1 to obtain a uniform bound for the $L^{\infty}$ norm of $\phi$ everywhere in the interior and to prove continuous extendibility across $\mathcal{CH}^+$. For this purpose, we apply standard Sobolev inequalities on the spheres $S^2_{u,v}$ corresponding to the double-null foliation, i.e.\ we can estimate
\begin{equation*}
||\phi||_{L^{\infty}(S^2_{u,v})}\leq C\sum_{|k|\leq 2}||\snabla^k\phi||_{L^2(S^2_{u,v})},
\end{equation*}
where $\snabla$ denotes the covariant derivative restricted to $S^2_{u,v}$. Moreover, we apply the fundamental theorem of calculus along the null generators of ingoing null hypersurfaces, together with a Cauchy--Schwarz inequality, to arrive at the following estimate:
\begin{equation}
\label{est:fundthmcalcCS}
\int_{S^2_{u,v}}\phi^2\,d\mu_{\slashed{g}}\leq \int_{S^2_{-\infty,v}}\phi^2\,d\mu_{\slashed{g}}+\int_{-\infty}^u|u'|^{-p}\,du'\int_{-\infty}^u\int_{S^2_{u',v}}|u'|^p(\underline{L}\phi)^2\,d\mu_{\slashed{g}}\,du',
\end{equation}
with $p>1$. The second term on the right-hand side of the inequality can be controlled by a weighted energy along an ingoing null hypersurface. 

In order to estimate $||\snabla^k\phi||_{L^2(S^2_{u,v})}$ with $k\geq 1$, we also need to consider appropriately weighted energies for angular derivatives of $\phi$. Replacing $\phi$ by $\snabla^k \phi$ (or $\partial^k_{\vartheta^A}\phi$, where $\vartheta^A$, with $A=1,2$, are coordinates on the spheres $S^2_{u,v}$) in the estimates from Part 1 results in error terms that cannot be controlled using the methods mentioned in Part 1.

Obtaining estimates for angular derivatives of $\phi$ in $L^2(S^2_{u,v})$ turns out not to be a problem in extremal Reissner--Nordstr\"om, as the spacetime is spherically symmetric, which means that the angular momentum operators, Killing vector fields $O_i$ corresponding to the isometries of spherical symmetry, where $i=1,2,3$, control all derivatives tangential to the spheres of the double-null foliation; see also Section~ 2.1 of \cite{gajic2015} for explicit expressions of $O_i$ with respect to spherical polar coordinates. Since the vector fields $O_i$ are Killing, they commute with the operator $\square_g$, so the functions $O_i(\phi)$ are also solutions to (\ref{eq:waveqkerr}). Any energy estimate for $\phi$ therefore automatically holds for $O_i(\phi)$. 

In extremal Kerr--Newman with $a\neq 0$, however, the only angular momentum operator that remains a Killing vector field is $\Phi$, the generator of rotations about the axis of symmetry. Fortuitously, there exists a \emph{second-order} operator $Q$, the \emph{Carter operator}, which also commutes with $\square_g$. This operator is closely related to the conserved \emph{Carter constant}; see \cite{Carter1977}. See also Andersson--Blue \cite{andblue1}, for example, for more details on the Carter constant and operator, and for applications of the commutation property of $Q$.

The operator $Q$, together with the vector fields $\Phi$ and $T$, the Killing vector field corresponding to time-translation symmetry, control the derivatives of $\phi$ that are tangent to the spheres of the Boyer--Lindquist foliation of Kerr-Newman. To obtain control over derivatives tangent to the spheres $S^2_{u,v}$, (which do \uline{not} coincide with the Boyer--Lindquist spheres if $a\neq 0$), we need to additionally commute $\square_g$ with $L$ and $\underline{L}$. 

In contrast with the error terms arising from commuting $\square_g$ with $\snabla$, or $\partial_{\vartheta^A}$, the error terms corresponding to a commutation with $L$ and $\underline{L}$ \emph{can} be controlled via the methods of Part 1 by using profusely the Killing property of $Q$, $\Phi$ and $T$; see the estimates in Section~ \ref{sec:commutatorestimates}. As a result, we are able to prove Theorem \ref{thm:linftyboundv1} and \ref{thm:linftyboundv2} if we restrict to extremal Kerr ($|a|=M$); see Section~ \ref{sec:uniformboundednessphi} and \ref{sec:extendibilityphiC0}.

\subsubsection{Part 3: Decay estimates (Theorem \ref{thm:c11boundphiv1}, \ref{thm:H1boundv2} and \ref{thm:c11boundphiv2})}
In the final step, we consider the difference function $\psi=\phi-\phi|_{\mathcal{H}^+}$, such that $\psi$ vanishes along $\mathcal{H}^+$. The function $\psi$ has the advantage that it can be shown to decay uniformly in $u$. By treating the wave equation as a transport equation for $L\phi$ along ingoing null generators, we can use the $u$-decay of $\psi$ to obtain $v$-decay of $||L\phi||_{L^2(S^2_{u,v})}$ with the rate $v^{-2+\epsilon}$ , for any $\epsilon>0$. If $\phi$ is axisymmetric, we can in fact improve this decay rate to $v^{-2}\log(v)$. By commuting further with $L$ and $\underline{L}$ and applying standard Sobolev inequalities on $S^2_{u,v}$, this allows one to obtain \emph{pointwise} decay for $|L\phi|$ with the rates $v^{-2+\epsilon}$ and $v^{-2}\log(v)$, respectively.

The outgoing derivative corresponding to double-null coordinates that cover the region \emph{beyond} $\mathcal{CH}^+$ in the maximal analytic extension of extremal Kerr--Newman, denoted by $\partial_{\widetilde{V}}$, is related to $L$ as follows:
\begin{equation*}
|\partial_{\widetilde{V}}\phi|\sim v^2 |L\phi|.
\end{equation*}
Since we cannot remove the $\epsilon$ in the decay rate $v^{-2+\epsilon}$ of $|L\phi|$, we are unable to infer boundedness of $\partial_{\widetilde{V}}\phi$ at $\mathcal{CH}^+$ or $C^1$-extendibility of $\phi$ at $\mathcal{CH}^+$. We \emph{can} nevertheless infer that $\phi$ is extendible as a $C^{0,\alpha}$ function beyond $\mathcal{CH}^+$, for any $\alpha<1$, if the initial data along $\uline{H}_{v_0}\cup \mathcal{H}^+$ are suitably regular and decaying, thereby proving Theorem \ref{thm:c11boundphiv1} and Theorem \ref{thm:c11boundphiv2}; see Section~ \ref{sec:decayLphi}.

Moreover, we can integrate the $v$-decaying $L^2(S^2_{u,v})$ norm of $L\phi$ in the $v$-direction, in slowly rotating extremal Kerr--Newman, to obtain boundedness of $\int_{H_u}v^2(L\phi)^2$ and also $\int_{\underline{H}_v} v^2\Omega^2|\snabla\phi|^2$. In this way we get rid of some of the $\epsilon$-loss in the weights that was present in the energy estimates of Part 1 and arose from the obstruction of $\zeta$ to the energy estimates for $\phi$ without the axisymmetry assumption. This improvement comes at the expense of requiring decay of higher order derivatives in the initial data, compared to the estimates in Part 1. See Section~ \ref{sec:decayLphi} for more details. In particular, we can infer Theorem \ref{thm:H1boundv2}.
 
\subsection{Outline}
In Section~ \ref{sec:geometry} we introduce some notation and state estimates relating to the double-null foliation of the interior of extremal Kerr--Newman (Part 0 of Section~ \ref{sec:mainsteps}) that are relevant in the rest of the paper. We state the theorems that are proved in the paper in Section~ \ref{sec:mainresults}. We prove energy estimates for axisymmetric solutions $\phi$ to (\ref{eq:waveqkerr}) in extremal Kerr--Newman in Section~ \ref{sec:energyestimatesaxisymm}. Subsequently, we prove energy estimates in slowly rotating extremal Kerr--Newman in Section~ \ref{sec:energyestimatesslowlyrot}, completing Part 1 of Section~ \ref{sec:mainsteps}. In Section~ \ref{sec:higherorderenergyestimates}, we commute with $L$ and $\underline{L}$ to arrive at energy estimates for higher-order derivatives of $\phi$. Finally, we use the higher-order energy estimates to prove pointwise estimates of $\phi$ (Part 2 of Section~ \ref{sec:mainsteps}). Moreover, we obtain pointwise decay in $v$ of $L\phi$ in Section~ \ref{sec:poinwiseestimates} by making use of higher-order energy estimates, completing Part 3 of Section~ \ref{sec:mainsteps}.

\subsection{Acknowledgments}
I thank Mihalis Dafermos for his guidance and support throughout the completion of this project and Jonathan Luk and Harvey Reall for a multitude of discussions and insights. This work was supported by the European Research Council grant no. ERC-2011-StG 279363-HiDGR.
\section{The geometry of extremal Kerr--Newman}
\label{sec:geometry}
We will first introduce the extremal Kerr--Newman spacetimes in Boyer--Lindquist and Kerr-star coordinates and subsequently present more convenient double-null coordinates, by foliating the spacetime with suitable ingoing and outgoing null hypersurface, covering both the exterior and interior regions of extremal Kerr--Newman. Section \ref{sec:estimatesmetricomp} and \ref{sec:estimatesmetricconn} are based on a more elaborate discussion on double-null foliations of Kerr--New\-man that can be found in \cite{Gajic2015b}.

\subsection{Boyer--Lindquist and Kerr-star coordinates}
Fix the mass parameter $M>0$ and the rotation parameter $a\in \R$, such that $|a|\leq M$, and let moreover the charge parameter $e$ satisfy $e^2=M^2-a^2$. 

We define the \emph{exterior region of extremal Kerr--Newman} as a manifold $\mathcal{M}_{\textnormal{ext}}$, together with a metric $g$, where $\mathcal{M}_{\textnormal{ext}}=\R\times (M,\infty)\times \s^2$ can be equipped with the Boyer--Lindquist coordinate chart $(t,r,\theta,\varphi)$, with $t\in \R$, $r\in (M,\infty)$, $\theta \in(0,\pi)$ and $\varphi \in (0,2\pi)$. In these coordinates, the metric $g$ in $\mathcal{M}_{\rm ext}$ is given by 
\begin{equation}
\label{eq:metricbl}
g=-\left(1-\frac{2Mr}{\rho^2}\right)dt^2+\frac{\rho^2}{\Delta}dr^2+\rho^2 d\theta^2+R^2\sin^2\theta d\varphi^2-2\frac{(2Mr-e^2)a\sin^2\theta}{\rho^2}dtd\varphi.
\end{equation}
Here,
\begin{align*}
\Delta&:=r^2-2Mr+a^2+e^2=(M-r)^2,\\
\rho^2&:=r^2+a^2\cos^2\theta,\\
R^2&:=r^2+a^2+a^2\frac{(2Mr-e^2)\sin^2\theta}{\rho^2}.
\end{align*}

We define the \emph{interior region of extremal Kerr--Newman} as the manifold \\$\mathcal{M}_{\textnormal{int}}=\R\times (0,M)\times \s^2$ equipped with a metric that we also denote by $g$, which can similarly be covered by \emph{Boyer--Lindquist coordinates} $(t,r,\theta,\varphi)$, where now $t\in \R$, $r\in (0,M)$, $\theta \in(0,\pi)$ and $\varphi\in (0,2\pi)$. The components of the metric $g$ in $\mathcal{M}_{\rm int}$ with respect to Boyer--Lindquist coordinates are also given by (\ref{eq:metricbl}).

We define the ingoing \emph{Kerr-star coordinates} $(t_{KS})_*$ and $(\varphi_{KS})_*$ on $\mathcal{M}_{\textnormal{ext}}$ or $\mathcal{M}_{\rm int}$ in the following way:
\begin{align}
\label{eq:deftstar}
(t_{KS})_*(t,r)=\:&t+(r_{KS})_*(r),\\
\label{eq:defvarphistar}
(\varphi_{KS})_*(r,\varphi)=\:&\varphi+\int_{r_0}^r \frac{a}{\Delta(r')}\,dr',
\end{align}
where
\begin{equation}
\label{eq:defrefstar}
(r_{KS})_*(r)=\int_{r_0}^r \frac{{r'}^2+a^2}{\Delta(r')}\,dr',
\end{equation}
and $r_0>0$ is a constant.

We can change from Boyer--Lindquist coordinates to ingoing Kerr-star coordinates
\begin{equation*}
((t_{KS})_*,r,\theta,(\varphi_{KS})_*)
\end{equation*}
to show that the spacetime $\mathcal{M}_{\textnormal{int}}$ can be smoothly patched to the spacetime $\mathcal{M}_{\textnormal{ext}}$, such that $\mathcal{M}_{\rm int}$ embeds as the region $\{0<r<M\}$ of the patched spacetime, and $\mathcal{M}_{\rm ext}$ embeds as the region $\{r>M\}$. The boundary of $\mathcal{M}_{\textnormal{ext}}$ and $\mathcal{M}_{\rm int}$ inside the patched spacetime is given by the level set $\{r=M\}$. This boundary is called the \emph{event horizon} and is denoted by $\mathcal{H}^+$. It lies in the causal past of $\mathcal{M}_{\rm int}$; see also Figure \ref{fig:fullspacetime2}. We denote the patched manifold by $\mathcal{M}:=\mathcal{M}_{\textnormal{int}}\cup \mathcal{M}_{\textnormal{ext}}\cup \mathcal{H}^+$. We can write $$\mathcal{M}=\R\times (0,\infty)\times \s^2,$$ where $(t_{KS})_*\in \R$, $r\in (0,\infty)$, $\theta\in (0,\pi)$ and $(\varphi_{KS})_* \in ((\varphi_{KS})_*(r,0),(\varphi_{KS})_*(r,0)+2\pi)$.

\begin{figure}[h!]
\begin{center}
\includegraphics[width=2.3in]{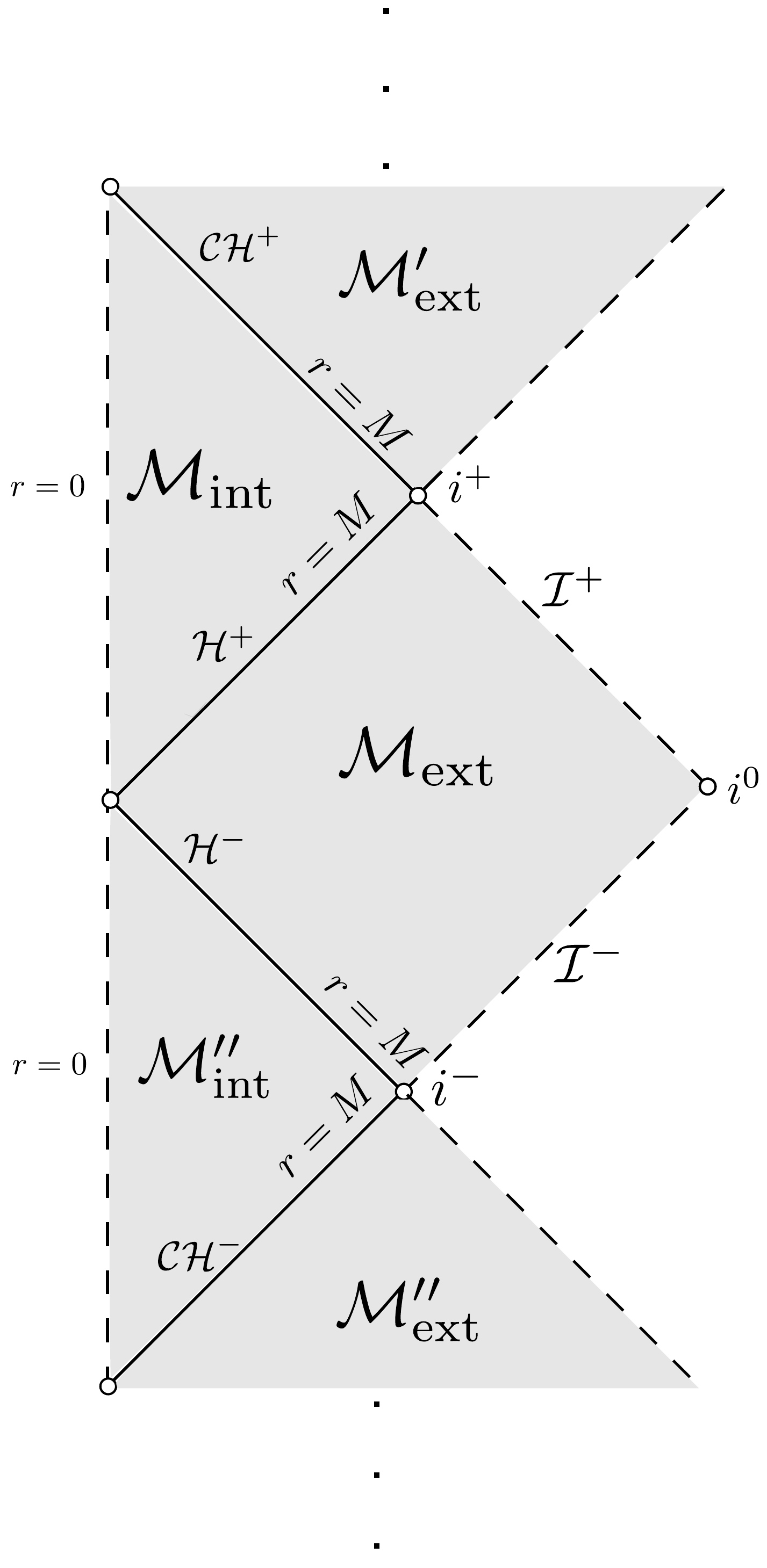}
\caption{\label{fig:fullspacetime2} The Penrose diagram of the maximally analytic extension of extremal Kerr--Newman}
\end{center}
\end{figure}

We can similarly introduce outgoing Kerr-star coordinates $((\tilde{t}_{KS})_{*},r,\theta,(\varphi_{KS})_*)$, where
\begin{equation*}
(\tilde{t}_{KS})_{*}=(t_{KS})_*-2(r_{KS})_*.
\end{equation*}
In these coordinates it is easy to see that $\mathcal{M}_{\textnormal{int}}$ can be smoothly embedded into a bigger spacetime $\mathcal{M}'$, by patching $\mathcal{M}_{\rm int}$ to a spacetime $\mathcal{M}_{\rm ext}'$ that is isometric to $\mathcal{M}_{\rm ext}$. The manifold $\mathcal{M}_{\rm in}$ is embedded in the patched spacetime as the region $\{0<r<M\}$ and $\mathcal{M}'_{\rm ext}$ is embedded as the region $\{r>M\}$. The corresponding boundary $\{r=M\}$ of $\mathcal{M}_{\rm int}$ and $\mathcal{M}'_{\rm ext}$ in the patched spacetime lies in the causal future of $\mathcal{M}_{\rm int}$ and is denoted by $\mathcal{CH}^+$. We refer to this boundary as the \emph{inner horizon}. We can write $\mathcal{M}'=\mathcal{M}_{\textnormal{int}}\cup \mathcal{M}'_{\textnormal{ext}}\cup \mathcal{CH}^+$, or $$\mathcal{M}'=\R\times (0,\infty)\times \s^2,$$ where $(\tilde{t}_{KS})_*\in \R$, $r\in (0,\infty)$, $\theta\in (0,\pi)$ and $(\varphi_{KS})_* \in ((\varphi_{KS})_*(r,0),(\varphi_{KS})_*(r,0)+2\pi)$.

As $\mathcal{M}_{\rm ext}'$ is isometric to $\mathcal{M}_{\rm ext}$, we can repeat the above procedure ad infinitum to extend the manifold $\mathcal{M}\cup \mathcal{M}'$ further and form an infinite sequence of patched manifolds containing regions isometric to either $\mathcal{M}_{\rm ext}$ or $\mathcal{M}_{\rm int}$, glued across horizons. The resulting spacetime $\widetilde{\mathcal{M}}$ is called \emph{maximal analytically extended extremal Kerr--Newman}, and it is depicted in Figure \ref{fig:fullspacetime2}. For the remainder of this paper we will, however, mainly direct our attention to the subset $\mathcal{M}\cup \mathcal{CH}^+$.

\subsection{Double-null coordinates}
\label{sec:estimatesmetricomp}
In the sections below, we will consider energy fluxes along ingoing and outgoing null hypersurfaces in $\mathcal{M}$. It is therefore more natural to work in double-null coordinates in $\mathcal{M}$ rather than Kerr-star coordinates. 

We first consider $\mathcal{M}_{\rm ext}$, covered by Kerr-star coordinates. If we can construct a \emph{tortoise function} $r_*$ to be of the form $r_*(r,\theta)$, such that the functions
\begin{align*}
2v=\:&t+r_*,\\
2u=\:&t-r_*,
\end{align*}
satisfy the eikonal equations
\begin{equation*}
\label{eq:eikonal}
g^{\alpha\beta}\partial_{\alpha}u\partial_{\beta}u=0\:\:\textnormal{and}\:\:g^{\alpha\beta}\partial_{\alpha}v\partial_{\beta}v=0,
\end{equation*}
then the level sets $\{u=\textnormal{constant}\}$ and $\{v=\textnormal{constant}\}$ are null hypersurfaces. We will follow a construction of $r_*$ that was introduced by Pretorius--Israel in \cite{Pretorius1998} and allows for suitable, double-null coordinates. We will assume that $a\neq 0$. In the $a=0$ case we consider Eddington--Finkelstein double-null coordinates; see Section 2 of \cite{gajic2015}. 

In \cite{Gajic2015b} the construction of $r_*$ from \cite{Pretorius1998} is used to extend the local double-null coordinates in $\mathcal{M}_{\rm ext}$ to a smooth, \emph{global} Eddington--Finkelstein-type double-null foliation of $\mathcal{M}_{\textnormal{ext}}$, such that the 2-surfaces 
\begin{equation*}
S^2_{u',v'}=\{u=u'\}\cap \{v=v'\}
\end{equation*}
are diffeomorphic to 2-spheres and we moreover obtain quantitative bounds on the metric components in double-null coordinates.

The metric $g$ on $\mathcal{M}_{\rm ext}$ can then be written as
\begin{equation}
\label{eq:metrickerrdnull2}
g=-4\Omega^2dudv+\sg_{AB}(d\vartheta^A-b^Adv)(d\vartheta^B-b^Bdv),
\end{equation}
where $2u=t-r_*$, $2v=t+r_*$, $\vartheta^1=\theta_*$ and $\vartheta^2={\varphi_*}$, with $u,v\in \R$, $\theta_*\in (0,\pi)$ and $\varphi_*\in (0,2\pi)$. The metric components in (\ref{eq:metrickerrdnull2}) are given by
\begin{align*}
\Omega^2=\:&\Delta R^{-2},\\
b^{\theta_*}=\:&0,\quad b^{{\varphi_*}}=\frac{4M a r}{\rho^2 R^2}-\frac{4M a r}{\rho^2 R^2}\bigg|_{r=M},\\
g_{\theta_*\theta_*}=\:&f_1^2f_2^2(\partial_{\theta_*}F)^2 R^{-2}+(\partial_{\theta_*}f_4)^2R^2\sin^2\theta,\\
g_{\theta_* {\varphi_*}}=\:&(\partial_{\theta_*}f_4)R^2\sin^2\theta,\\
g_{{\varphi_*}{\varphi_*}}=\:&R^2 \sin^2 \theta;
\end{align*}
they are of the same form as the metric components with respect to the double-null coordinates considered in \cite{Dafermos2013b}. We have that
\begin{equation*}
F(r(u,v,\theta_*),\theta(u,v,\theta_*),\theta_*)=\int_{0}^{\theta} f_1^{-1}(\theta',\theta_*)\,d\theta'+\int_r^{\infty} f_2^{-1}(r',\theta_*)\,dr'+\frac{1}{\sin \theta_*\cos\theta_*}\frac{df_3}{d\theta_*}(\theta_*),
\end{equation*}
with
\begin{align*}
f_1(\theta_*,\theta(u,v,\theta_*))=\:&\sqrt{a^2(\sin^2\theta_*-\sin^2\theta)},\\
f_2(\theta_*,r(u,v,\theta_*))=\:&\sqrt{(r^2+a^2)^2-a^2\Delta \sin^2\theta_*},\\
f_3(\theta_*)=\:&-2a^{-1}\int_0^{\theta_*}\sqrt{\sin^2\theta_*-\sin^2 x}\,dx,
\end{align*}
and $f_4$ a function satisfying
\begin{equation*}
\partial_v f_4(u,v,\theta_*)=-\partial_u f_4(u,v,\theta_*)=-\frac{2Mar}{\rho^2 R^2}.
\end{equation*}
In the above expressions we are treating $r$ and $\theta$ as functions of $(u,v,\theta_*)$, via implicit relations.

The angular coordinate $\varphi_*$ is related to the Boyer--Lindquist angular coordinate $\varphi$ in the following way:
\begin{equation*}
\varphi_*=\varphi-f_4(r_*,\theta_*)-\frac{4M a r}{\rho^2R^2}\Big|_{r=M}v,
\end{equation*}
where $f_4$ can be chosen such that $\varphi_*\in (0,2\pi)$.

We can moreover express:
\begin{equation}
\label{eq:detslashedg}
\det \slashed{g}=f_1^2f_2^2(\partial_{\theta_*}F)^2\sin^2\theta,
\end{equation}
which does not depend on the function $f_4$.

In the formal limit $a\to 0$ the double-null coordinates $(u,v,\theta_*,\varphi_*)$ become simply Ed\-ding\-ton--Fin\-kel\-stein double-null coordinates on extremal Reissner--Nordstr\"om.

As we approach $\mathcal{H}^+$, the coordinate $u$ goes to infinity, whereas $v$ remains finite. We can, however, introduce a rescaled ingoing null coordinate in order to further extend the double-null coordinates and additionally cover the region $(\mathcal{M}_{\textnormal{int}}\cup\mathcal{H}^+)\cap\{r>\frac{e^2}{2M}\}\subset \mathcal{M}$.\footnote{If $e\neq 0$ (i.e.\ $0<|a|<M$), we cannot cover the entire region $\mathcal{M}_{\rm int}$ by the double-null coordinates introduced in this section, as the corresponding null generators form caustics in the region $\{0<r<\frac{e^2}{2M}\}$.}

Fix $v_0\in \R$ and define the function $U: \R \to (0,\infty)$ by $U(u)=M-r(u,v_0,\theta_*=\frac{\pi}{2})$. We can interpret $U$ as a smooth, negative function $U:\mathcal{M}_{\rm ext}\to (0,\infty)$. We have that
\begin{equation*}
\frac{du}{dU}=-(\partial_u r)^{-1}|_{v=v_0,\,\theta_*=\frac{\pi}{2}}=(\partial_{r_*} r)^{-1}|_{v=v_0,\,\theta_*=\frac{\pi}{2}}=r^{-2}\Omega^{-2}|_{v=v_0,\,\theta_*=\frac{\pi}{2}}
\end{equation*}

The function $U: \mathcal{M}_{\rm ext}\to\R$ extends smoothly with respect to Kerr-star coordinates to the bigger manifold $\mathcal{M}\cap \{r> \frac{e^2}{2M}\}$, such that $U=0$ along $\mathcal{H}^+$ and $U>0$ in $\mathcal{M}_{\rm int}\cap \{r> \frac{e^2}{2M}\}$.

In \cite{Gajic2015b} it is shown that the metric is well-defined and non-degenerate with respect to the chart $(U,v,\theta_*,\varphi_*)$ on $\mathcal{M}\cap \{r> \frac{e^2}{2M}\}$:
\begin{equation*}
g=-4\frac{\Omega^2(u,v,\theta_*)}{(r^{2}\Omega^{2})\left(u,v=v_0,\theta_*=\frac{\pi}{2}\right)}dUdv+\sg_{AB}(d\vartheta^A-b^Adv)(d\vartheta^B-b^Bdv).
\end{equation*}
Consequently, $(U,v,\theta_*,\varphi_*)$ defines a smooth coordinate chart on $\mathcal{M}\cap \{r> \frac{e^2}{2M}\}$.

In $\mathcal{M}_{\rm int}\cap \{r>\frac{e^2}{2M}\}$, we can moreover revert back to a function $u:\mathcal{M}_{\rm int}\cap~\{r>\frac{e^2}{2M}\}\to \R$, defined by $u=t-v$. Consequently, we can also cover $\mathcal{M}_{\textnormal{int}}\cap~\{r>\frac{e^2}{2M}\}$ by Ed\-ding\-ton--Finkelstein-type double-null coordinates $(u,v,\theta_*,\varphi_*)$, such that the metric is given by (\ref{eq:metrickerrdnull2}), where now $r\in (\frac{e^2}{2M},M)$.

Consider the metric (\ref{eq:metrickerrdnull2}) in $\mathcal{M}_{\textnormal{int}}\cap\{r>\frac{e^2}{2M}\}$. We will shift the angular coordinate $\varphi$ by introducing a function $f_5(r_*,\theta_*)$, such that
\begin{equation*}
\partial_v f_5=-\partial_u f_5=\frac{2Mar}{\rho^2 R^2}.
\end{equation*}
Note that there is still freedom left in the choice of $\partial_{\theta_*}f_5$. Define the new angular coordinate $\widetilde{\varphi}_*\in (0,2\pi)$ as follows:
\begin{equation*}
\widetilde{\varphi}_*=\varphi_*+f_4(r_*,\theta_*)-f_5(r_*,\theta_*)+\frac{4M a r}{\rho^2 R^2}\bigg|_{r=M}r_*.
\end{equation*}
Then the metric can be written in $(\widetilde{u},\tilde{v},\widetilde{\theta}_*,\widetilde{\varphi}_*)$ coordinates, 
\begin{equation}
\label{eq:metrickerrdnull3}
g=-4\Omega^2(\widetilde{u},\tilde{v})d\widetilde{u}d\tilde{v}+\widetilde{\sg}_{AB}(d\widetilde{\vartheta}^A-\tilde{b}^Ad\widetilde{u})(d\widetilde{\vartheta}^B-\tilde{b}^Bd\widetilde{u}),
\end{equation}
where $2\widetilde{u}=2u=t-r_*$, $2\tilde{v}=2v=t+r_*$, $\widetilde{\vartheta}^1=\widetilde{\theta}_*$, $\widetilde{\vartheta}^2=\widetilde{\varphi}_*$ and moreover,
\begin{align*}
\Omega^2(\widetilde{u},\tilde{v})=\:&\Omega^2(r_*)=\Delta R^{-2},\\
\tilde{b}^{\widetilde{\theta}_*}=\:&0,\quad \tilde{b}^{\widetilde{\varphi}_*}=\frac{4M a r}{\rho^2 R^2}-\frac{4M a r}{\rho^2 R^2}\bigg|_{r=M},\\
\widetilde{\sg}_{\widetilde{\theta}_*\widetilde{\theta}_*}=\:&f_1^2f_2^2(\partial_{\theta_*}F)^2 R^{-2}+(\partial_{\theta_*}f_5)^2R^2\sin^2\theta,\\
\widetilde{\sg}_{\widetilde{\theta}_* \widetilde{\varphi}_*}=\:&(\partial_{\theta_*}f_5)R^2\sin^2\theta,\\
\widetilde{\sg}_{\widetilde{\varphi}_*\widetilde{\varphi}_*}=\:&R^2 \sin^2 \theta.
\end{align*}
To distinguish these coordinates from the previous double-null coordinates, we have denoted them with tildes $(\widetilde{u},\tilde{v},\widetilde{\vartheta})$, where $\tilde{v}=v$, $\widetilde{u}=u$ and $\widetilde{\theta}_*=\theta_*$.

Now, fix $u_0\in \R$ and define the function $\widetilde{V}:\mathcal{M}_{\rm int}\cap \{r>\frac{e^2}{2M}\}$ by $\widetilde{V}=(r-M)(u_0,v,\theta_*=\frac{\pi}{2})$. Then
\begin{equation*}
\frac{d\tilde{v}}{d\widetilde{V}}=(\partial_v r)^{-1}|_{u=u_0,\theta_*=\frac{\pi}{2}}=(\partial_{r_*} r)^{-1}|_{u=u_0,\,\theta_*=\frac{\pi}{2}}=r^{-2}\Omega^{-2}|_{u=u_0,\,\theta_*=\frac{\pi}{2}}
\end{equation*}

We can extend $\widetilde{V}$ as a smooth function to the bigger manifold $\mathcal{M}'\cap \{r> \frac{e^2}{2M}\}$, such that $\widetilde{V}=0$ along $\mathcal{CH}^+$ and $\widetilde{V}>0$ in $\mathcal{M}'_{\rm ext}$.

In \cite{Gajic2015b} it is shown that the metric is well-defined and non-degenerate with respect to the chart $(\widetilde{u},\widetilde{V},\widetilde{\theta}_*,\widetilde{\varphi}_*)$ on $\mathcal{M}'\cap \{r>\frac{e^2}{2M}\}$:
\begin{equation*}
g=-4\frac{\Omega^2(\widetilde{u},\tilde{v},\widetilde{\theta}_*)}{(r^{2}\Omega^{2})\left(\widetilde{u}=u_0,\tilde{v},\widetilde{\theta}_*=\frac{\pi}{2}\right)}d\widetilde{u}d\widetilde{V}+\sg_{AB}(d\widetilde{\vartheta}^A-\tilde{b}^Ad\widetilde{u})(d\widetilde{\vartheta}^B-\tilde{b}^Bd\widetilde{u}).
\end{equation*}

We will use the notation $(u,v,\theta_*,\varphi_*)=(-\infty,v_0,\theta_*,\varphi_*)$ and $(u,v,\theta_*,\varphi_*)=(u_0,\infty,\theta_*,\varphi_*)$, with $u_0,v_0<\infty$, for points on $\mathcal{H}^+$ and $\mathcal{CH}^+$, respectively, for the sake of convenience. These points lie in the domain of either the $(U,v)$ or $(\widetilde{u},\widetilde{V})$ double-null coordinates.

In $\mathcal{M}_{\textnormal{int}}\cap\left\{r> \frac{e^2}{2M}\right\}\cup \mathcal{H}^+\cup \mathcal{CH}^+$ we restrict to the region
\begin{equation*}
\begin{split}
{\mathcal{D}_{u_0,v_0}}=&\:\Bigg\{x\in\mathcal{M}_{\textnormal{int}}\cap\left\{r> \frac{e^2}{2M}\right\}\cup \mathcal{H}^+\cup \mathcal{CH}^+\::\: U(x)\in[0,U(u_0)],\:\widetilde{V}(x)\in[\widetilde{V}(v_0),0],\\
&\:(U(x),\widetilde{V}(x))\neq (0,0)\Bigg\}.
\end{split}
\end{equation*}

Let $v'\in[v_0,\infty)$ and $u'=[-\infty,u_0]$. We will consider the following null hypersurfaces:
\begin{align*}
\uline{H}_{v'}&:=\{x\in \mathcal{M}\::\:U(x)\in [0,U(u_0)],\: v(x)=v'\},\\
H_{u'}&:=\{x\in \mathcal{M}\::\:U(x)=U(u'),\: v(x)\in[v_0,\infty)\},
\end{align*}
and we refer to the hypersurfaces $\uline{H}_{v'}$ and $H_{u'}$ as ingoing and outgoing null hypersurfaces, respectively.

\textbf{We will fix $|u_0|$ and $v_0$ to be suitable large such that $$\uline{H}_{v_0}\cup H_{u_0}\subset \left(\mathcal{M}_{\textnormal{int}}\cap\left\{r>\frac{e^2}{2M}\right\}\right)\cup \mathcal{H}^+.$$}

Consider the null vector fields $L$ and $\underline{L}$, which are tangent to the generators of the outgoing and ingoing null hypersurfaces respectively and satisfy $Lv=1$ and $Lu=1$.

The vector field $\underline{L}$ can be naturally expressed in the chart $(u,v,\theta_*,\varphi_*)$. Indeed,
\begin{align*}
L=\:&\partial_v+b^A\partial_{\vartheta^A},\\
\underline{L}=\:&\partial_u.
\end{align*}

The vector field $L$ can similarly be expressed in the chart $(\widetilde{u},\widetilde{V},\widetilde{\theta}_*,\widetilde{\varphi}_*)$:
\begin{align*}
L=\:&\partial_{\tilde{v}},\\
\underline{L}=\:&\partial_{\widetilde{u}}+\tilde{b}^A\partial_{\widetilde{\vartheta}_A}.
\end{align*}

\subsection{Estimates for metric components and connection coefficients in double-null coordinates}
\label{sec:estimatesmetricconn}
In this section, we will present an overview of relevant estimates for the metric components $g_{\alpha \beta}$ in Eddington--Finkelstein-type double-null coordinates, their derivatives and components (and derivatives) of the Jacobian matrix relating Eddington--Finkelstein-type double-null coordinates to Boyer--Lindquist double-null coordinates. All these estimates are obtained in \cite{Gajic2015b}.

We first define the following notation to separate out leading-order terms in $v+|u|$.
\begin{definition}
Let $f: \mathcal{M}\cap {\mathcal{D}_{u_0,v_0}}\to \R$ be a $C^0$ function. We say that $f\in \mathcal{O}((v+|u|)^{-l})$, where $u$ and $v$ are Eddington--Finkelstein-type double-null coordinates in \\$\mathcal{M}_{\textnormal{int}}\cap\{r>r_0>\frac{e^2}{2M}\}$, if there exists a constant $C=C(M,a,r_0,\Sigma)>0$, such that
\begin{equation*}
|f|(u,v,\theta_*,\varphi_*)\leq C (v+|u|)^{-l}.
\end{equation*}
\end{definition}

We obtain in \cite{Gajic2015b} the following estimates for the metric components $g_{\alpha \beta}$ in \\$\mathcal{M}_{\textnormal{int}}\cap\{r>r_0>\frac{e^2}{2M}\}$ in the Eddington--Finkelstein-type coordinates $(u,v,\theta_*,\varphi_*)$ introduced above:

\begin{theorem}[Estimates for metric components in double-null coordinates, \cite{Gajic2015b}]
\label{thm:estmetric}
Let $r_0>\frac{e^2}{2M}$ and consider $\mathcal{M}_{\textnormal{int}}\cap\{r>r_0\}$ covered by the double-null coordinates $(u,v,\theta_*,\varphi_*)$ introduced in Section \ref{sec:estimatesmetricomp}. 
\begin{itemize}
\item[(i)]
There exist constants $c=c(r_0,a,M)>0$ and $C=C(N,r_0,a,M)>0$, such that for all $n \in \Z$, with $n\leq N$, where $N\in \N_0$,
\begin{align*}
|\partial_{r_*}^n\slashed{g}_{\theta_*\theta_*}|\leq \:& C(v-u)^{-2n},\\
|\partial_{r_*}^n\slashed{g}_{\theta_*\varphi_*}|\leq \:& C(v-u)^{-2n},\\
|\partial_{r_*}^n\slashed{g}_{\varphi_*\varphi_*}|\leq \:& C\sin^2\theta(v-u)^{-2n},\\
|\partial_{\theta_*}\slashed{g}_{\theta_*\theta_*}|\leq \:&C,\\
|\partial_{\theta_*}\slashed{g}_{\theta_*{\varphi_*}}|\leq \:&C\sin^2\theta,\\
|\partial_{\theta_*}\slashed{g}_{{\varphi_*}{\varphi_*}}|\leq \:& C\sin\theta,\\
c\sin^2\theta\leq \textnormal{det}\,\slashed{g}\leq \:& C\sin^2\theta.
\end{align*}
\item[(ii)]
We can expand
\begin{align*}
v+|u|=\:&r_*(r,\theta)=\frac{a^2+M^2}{M-r}+2M\log(M-r)+\mathcal{O}(1),\\
\Omega^{-2}=\:&\frac{1}{M^2+a^2\cos^2\theta}\left[(v+|u|)^2+4M(v+|u|)\log(M-r)\right]+\mathcal{O}(v+|u|),\\
b^{\varphi_*}=\:&\frac{4Ma}{(M^2+a^2)^2}(3M^2-a^2)(v+|u|)^{-1}+\log(v+|u|)\mathcal{O}((v+|u|)^{-2})
\end{align*}
and estimate for $n\leq N$, with $N\in \N$,
\begin{align*}
|\partial_{r_*}^{n+1}b^{\varphi_*}|\leq C(v+|u|)^{-2(n+1)},\\
|\partial_{r_*}^n\partial_{\theta_*}b^{\varphi_*}|\leq C(v+|u|)^{-2-2n}\log(v+|u|),\\
|\partial_{r*}\Omega^2|\leq C(v+|u|)^{-2},\\
|\partial_{r*}\partial_{\theta_*}\Omega^2|\leq C(v+|u|)^{-2},\\
\end{align*}
where $C=C(N,r_0,a,M)>0$.

\item[(iii)]
There exist $c=c(r_0,a,M)>0$ and $C=C(r_0,a,M)>0$ such that for all $n\in \Z$, with $n\leq N$, where $N\in \N_0$:
\begin{align}
\label{eq:estJac4}
|\partial_{r_*}\theta|\leq \:& C(v-u)^{-2}\sin \theta,\\
\label{eq:estJac5}
|\partial_{r_*}r|\leq \:& C (v-u)^{-2},\\
\label{eq:estJac1}
c(v-u)^{-2n}\leq |\partial_{r_*}^n\partial_{\theta}\theta_*|\leq \:& C(v-u)^{-2n},\\
\label{eq:estJac2}
|\partial_{r_*}^n\partial_{\theta}^2\theta_*|\leq \:& C(v-u)^{-2n},\\
\label{eq:estJac3}
|\partial_{r_*}^n\partial_{\theta}r_*|\leq \:& C\sin\theta (v-u)^{-2n}\leq C \sin \theta_*(v-u)^{-2n},\\
\label{eq:estJac6}
|\partial_{\theta}^2r_*|\leq \:& C.
\end{align}
\end{itemize}
\end{theorem}

We define the connection coefficients
\begin{align*}
\chi_{AB}&:=g(\nabla_{\partial_{\vartheta_A}} e_4,\partial_{\vartheta_B}),\\
\underline{\chi}_{AB}&:=g(\nabla_{\partial_{\vartheta_A}} e_3,\partial_{\vartheta_B}),\\
\omega&:=-\frac{1}{4}g(\nabla_{e_4}e_3,e_4),\\
\underline{\omega}&:=-\frac{1}{4}g(\nabla_{e_3}e_4,e_3),\\
\zeta_A&:=\frac{1}{2}g(\nabla_{\partial_{\vartheta_A}}e_4,e_3),
\end{align*}
where $A=1,2$ and $e_3=\Omega^{-1}\uline{L}$ and $e_4=\Omega^{-1}L$ are renormalised null vector fields, such that $g(e_3,e_4)=-2$. We have the following relations between connection coefficients and metric derivatives:
\begin{align*}
2\Omega\chi_{AB}=\:&L(\slashed{g}_{AB})+\partial_Ab^C\slashed{g}_{CB}+\partial_Bb^C\slashed{g}_{CA},\\
2\Omega\underline{\chi}_{AB}=\:&\underline{L}(\slashed{g}_{AB}),\\
\Omega \tr \chi=\:&\frac{L(\sqrt{\det \slashed{g}})}{\sqrt{\det \slashed{g}}},\\
\Omega \tr \underline{\chi} =\:& \frac{\underline{L}(\sqrt{\det \slashed{g}})}{\sqrt{\det \slashed{g}}},\\
4\Omega \omega=\:&\Omega^{-2}L(\Omega^2),\\
4\Omega \underline{\omega}=\:&\Omega^{-2}\underline{L}(\Omega^2),\\
\zeta^A=\:&\frac{1}{4}\Omega^{-2}[L,\underline{L}]^A=\Omega^{-2}\partial_vb^A.
\end{align*}

See Appendix \ref{sec:bulkest} for the derivations of the above identities and for further properties the connection coefficients and their expressions in terms of derivatives of $g_{\alpha \beta}$.

\begin{theorem}[Estimates for connection coeffients in double-null coordinates, \cite{Gajic2015b}]
\label{thm:estconncoef}
Let $r_0>\frac{e^2}{2M}$ and consider $\mathcal{M}_{\textnormal{int}}\cap\{r>r_0\}$.
\begin{itemize}
\item[(i)]
Let $A,B=1,2$ and $n\in \Z$, with $n\leq N$, where $N\in \N_0$. There exist a constant $C=C(N,r_0,a,M)>0$, such that
\begin{align*}
|\Omega \hat{\chi}_{AB}|\leq \:& C(v+|u|)^{-2}\log(v+|u|)|\slashed{g}_{AB}|,\\
|\Omega \hat{\underline{\chi}}_{AB}|\leq \:& C(v+|u|)^{-2}|\slashed{g}_{AB}|,\\
0<\Omega \tr \chi\leq \:& C(v+|u|)^{-2},\\
0<-\Omega \tr \underline{\chi}\leq \:& C(v+|u|)^{-2},\\
|\partial_{r_*}^n (\Omega \tr \chi)|\leq \:& C(v+|u|)^{-2-n},
\end{align*}
where we made use the following notation
\begin{align*}
\hat{\chi}_{AB}=&\:\chi_{AB}-\frac{1}{2}\slashed{g}_{AB}\tr \chi,\\
\hat{\uline{\chi}}_{AB}=&\:\uline{\chi}_{AB}-\frac{1}{2}\slashed{g}_{AB}\tr \uline{\chi},\\
\end{align*}
\item[(ii)]
Moreover, we can expand
\begin{align*}
4\Omega \omega =\:&-\frac{2}{v+|u|}+\log(v+|u|)\mathcal{O}((v+|u|)^{-2}),\\
4\Omega \underline{\omega}=\:&\frac{2}{v+|u|}+\log(v+|u|)\mathcal{O}((v+|u|)^{-2}),\\
\Omega^2\zeta^{\varphi_*}=\:&\frac{Ma}{(M^2+a^2)^2}(3M^2-a^2)(v+|u|)^{-2}+\log(v+|u|)\mathcal{O}((v+|u|)^{-3}),\\
\Omega^2 \zeta^{\theta_*}=\:&0.
\end{align*}
\end{itemize}
\end{theorem}

\subsection{Killing vector fields}
\label{sec:hawkvf}

The vector field $T=\frac{\partial}{\partial {(t_{KS})_*}}$ in $\mathcal{M}$, as expressed in Kerr-star coordinates ($T=\partial_t$ in Boyer--Lindquist coordinates on $\mathcal{M}_{\rm int}$), is a Killing vector field; it corresponds to time-translation symmetry in extremal Kerr--Newman. Note that $T$ is not causal everywhere in $\mathcal{M}$. The subset of $\mathcal{M}_{\textnormal{ext}}$ in which $T$ is not causal is called the \emph{ergoregion}. Similarly, there is a subset of $\mathcal{M}_{\textnormal{int}}$ in which $T$ fails to be causal everywhere (cf.\ $T$ is timelike \emph{everywhere} away from the horizons in extremal Reissner--Nordstr\"om, where $a=0$).

We denote the Killing vector field corresponding to axial symmetry in extremal Kerr--New\-man by $\Phi$. In Kerr-star coordinates, we can write $\Phi=\frac{\partial}{\partial {(\varphi_{KS})_*}}$. However, we can also write $\Phi=\partial_{\varphi_*}$ in Eddington-Finkelstein-type double-null coordinates, or $\Phi=\partial_{\varphi}$ in Boyer--Lindquist coordinates.

The Carter operator is a second order differential operator that can be expressed as follows:
\begin{equation*}
Q={\Delta}_{\s^2}+(a^2\sin\theta) T^2-\Phi^2.
\end{equation*}
Since $T$ and $\Phi$ are Killing vector fields, we have that
\begin{equation*}
[\square_g,T]=[\square_g,\Phi]=0. 
\end{equation*}
It turns out that the Carter operator also commutes with the wave operator:
\begin{equation*}
[\square_g,Q]=0.
\end{equation*}
See \cite{andblue1} for a derivation of the above commutator.

We can define the \emph{Hawking vector field} $H$ in ${\mathcal{D}_{u_0,v_0}}$ by
\begin{equation*}
H=\frac{1}{2}(\partial_u+\partial_v)=\frac{1}{2}(\underline{L}+L-b^{{\varphi_*}}\Phi).
\end{equation*}
We can also express $H$ by as a linear combination of the Killing vector fields $T$ and $\Phi$,
\begin{equation*}
H=T+\omega_{\mathcal{H}^+}\Phi,
\end{equation*}
where
\begin{equation*}
\begin{split}
\omega_{\mathcal{H}^+}&:=\frac{2M a r}{\rho^2 R^2}\bigg|_{r=r_+}=\frac{2aM^2}{(M^2+a^2)^2}.
\end{split}
\end{equation*}
Indeed, in $\mathcal{M}_{\rm int}\cap \{r> \frac{e^2}{2M}\}$ we can write
\begin{equation*}
\begin{split}
T=\partial_t=\:&\partial_tu\partial_u+\partial_tv \partial_v+\partial_t{\varphi_*}\partial_{{\varphi_*}}\\
=\:&\frac{1}{2}\left[\partial_u+\partial_v-2\omega_{\mathcal{H}^+}\Phi\right].
\end{split}
\end{equation*}
In the literature, the constant $\omega_{\mathcal{H}^+}$ is commonly referred to as the \emph{angular velocity of the Kerr--Newman black hole}.

In order for the energy fluxes with respect to $H$ along null hypersurfaces to be non-negative definite, we need $H$ to be causal. We have that
\begin{equation*}
g(H,H)=g(L+\underline{L}-b^{{\varphi_*}}\partial_{{\varphi_*}},L+\underline{L}-b^{{\varphi_*}}\partial_{{\varphi_*}})=-4\Omega^{-2}+(b^{{\varphi_*}})^2R^2\sin^2\theta.
\end{equation*}
The maximum value of $R^2\sin^2\theta$ is obtained at $\theta=\frac{\pi}{2}$,
\begin{equation*}
R^2\sin^2\theta|_{\theta=\frac{\pi}{2}}=\frac{(M^2+a^2)^2}{M^2}+\mathcal{O}((v+|u|)^{-1}).
\end{equation*}
Consequently, by applying the estimates in Theorem \ref{thm:estmetric}, we obtain
\begin{equation*}
\begin{split}
g(H,H)|_{\theta=\frac{\pi}{2}}=\:&\left[-8+R^2\left(\frac{4Ma}{(M^2+a^2)^2}(3M^2-a^2)\right)^2\right](v+|u|)^{-2}\\
&+\log(v+|u|)\mathcal{O}((v+|u|)^{-3})\\
=\:&\left[-8+\left(\frac{4a}{M^2+a^2}(3M^2-a^2)\right)^2\right](v+|u|)^{-2}+\log(v+|u|)\mathcal{O}((v+|u|)^{-3}).
\end{split}
\end{equation*}
Therefore, $g(H,H)\leq 0$ everywhere for $v+|u|$ suitably large, or equivalently, $M-r$ suitably small, if
\begin{equation*}
2a^2(3M^2-a^2)^2-(M^2+a^2)^2< 0.
\end{equation*}
We rescale $x=\left(\frac{a}{M}\right)^2$, with $x\in[0,1]$, to obtain an equivalent inequality:
\begin{equation*}
2x(3-x)^2-(1+x)^2<0.
\end{equation*}
One can solve the above cubic equation to obtain $0<a_{c}(M)<M$, such that $g(H,H)< 0$ for all $0\leq |a|< a_c$ and $v+|u|$ suitably large.

We define \emph{slowly rotating extremal Kerr--Newman spacetimes} to be the subfamily of extremal Kerr--Newman spacetimes satisfying $0\leq |a|< a_c$. Note that extremal Kerr ($|a|=M$) is \emph{not} a slowly rotating extremal Kerr--Newman spacetime.

\subsection{The divergence theorem and integration norms}
\label{sec:vfmethod}
In this section, we will introduce some basic notation regarding integration in $\mathcal{M}\cap {\mathcal{D}_{u_0,v_0}}$. We will state the divergence theorem, which is the main ingredient of the vector field method; see also the discussion in Section~ 1.3 of \cite{gajic2015}.

Let $V$ be a vector field in a Lorentzian manifold $(\mathcal{N},g)$. We consider the stress-energy tensor $\mathbf{T}[\phi]$ corresponding to (\ref{eq:waveqkerr}), with components
\begin{equation*}
\mathbf{T}_{\alpha \beta}[\phi]=\partial_{\alpha}\phi \partial_{\beta} \phi-\frac{1}{2}g_{\alpha \beta} \partial^{\gamma} \phi \partial_{\gamma} \phi.
\end{equation*}

Let $J^V[\phi]$ denote the energy current corresponding to $V$, which is obtained by applying $V$ as a \emph{vector field multiplier}, i.e.\ in components
\begin{equation*}
J^V_{\alpha}[\phi]=\mathbf{T}_{\alpha \beta}[\phi] V^{\beta}.
\end{equation*}

An energy flux is an integral of $J^V[\phi]$ contracted with the normal to a hypersurface with the natural volume form corresponding to the metric induced on the hypersurface. We apply the divergence theorem to relate the energy flux along the boundary of a spacetime region to the spacetime integral of the divergence of the energy current $J^V$. If the boundary has a null segment, there is no natural volume form or normal; these are assumed compatible with the divergence theorem.

That is to say, if we take $-\infty\leq u_1<u_2\leq u_0$ and $v_0\leq v_1<v_2\leq \infty$, the divergence theorem in the open rectangle $\{u_1<u<u_2,\:v_1<v<v_2\}$ in $\mathcal{M}_{\rm int}\cap \{r>\frac{e^2}{2M}\}$ gives the following identity:
\begin{equation}
\label{divergenceidentity}
\begin{split}
\int_{\{u_1<u<u_2,\:v_1<v<v_2\}} \textnormal{div}J^V[\phi]&=-\int_{H_{u_2}\cap \{v_1\leq v\leq v_2\}}J^V[\phi]\cdot L+\int_{H_{u_1}\cap \{v_1\leq v\leq v_2\}}J^V[\phi]\cdot L\\
&\quad - \int_{\uline{H}_{v_2}\cap \{u_1\leq u\leq u_2\}}J^V[\phi]\cdot \uline{L}+\int_{\uline{H}_{v_1}\cap \{u_1\leq u\leq u_2\}}J^V[\phi]\cdot \uline{L}.
\end{split}
\end{equation}

Here, we introduced the following notation:
\begin{equation*}
J^V[\phi]\cdot W=\mathbf{T}(V,W),
\end{equation*}
for vector fields $V$ and $W$. Moreover, in the notation on the left-hand side of (\ref{divergenceidentity}), we integrate over spacetime with respect to the standard volume form, i.e.\ let $f: \mathcal{M}\cap {\mathcal{D}_{u_0,v_0}}\to \R$ be a suitably regular function and $U$ an open subset of $\mathcal{M}$, then
\begin{equation*}
\int_{U}f:=\int_{U}f(u,v,\theta_*,\varphi_*)\,\sqrt{-\det g}d\theta_* d\varphi_*dudv= \int_{U}f(u,v,\theta_*,\varphi_*)\,2\Omega^2\sqrt{\det \slashed{g}}d\theta_*d\varphi_*dudv,
\end{equation*}
where $\det \slashed{g}$ is expressed in (\ref{eq:detslashedg}).

When integrating over $H_u$ and $\uline{H}_{v}$ we used the following convention in the notation on the right-hand side of (\ref{divergenceidentity}):
\begin{align*}
\int_{\uline{H}_v}f&:=\int_{-\infty}^{u_0}\int_{S^2_{u,v}}f\,\sqrt{\det \slashed{g}}d\theta_* d\varphi_*du=\int_{0}^{U(u_0)}\int_{S^2_{U,v}}f\frac{du}{dU}\,\sqrt{\det \slashed{g}}d\theta_* d\varphi_*dU,\\
\int_{{H}_u}f&:=\int_{v_0}^{\infty}\int_{S^2_{u,v}}f\sqrt{\det \slashed{g}}d\theta_* d\varphi_*dv=\int_{\widetilde{V}(v_0)}^0\int_{S^2_{u,\widetilde{V}}}f\frac{d\tilde{v}}{d\widetilde{V}}\,\sqrt{\det \slashed{g}}d{\theta}_* d{\varphi}_*d\widetilde{V}.
\end{align*}

In the notation of \cite{Christodoulou2008} we decompose the divergence term appearing in (\ref{divergenceidentity}) in the following way:
\begin{equation*}
\textnormal{div} J^V[\phi]=K^V[\phi]+\mathcal{E}^V[\phi],
\end{equation*}
where
\begin{align*}
K^V[\phi]&:=\mathbf{T}^{\alpha \beta}[\phi]\nabla_{\alpha}V_{\beta},\\
\mathcal{E}^V[\phi]&:=V(\phi)\square_g\phi.
\end{align*}
In particular, $\mathcal{E}^V[\phi]=0$ if $\phi$ is a solution to (\ref{eq:waveqkerr}). We can also replace $\phi$ by $W\phi$, where $W$ is a vector field that is referred to as a \emph{commutation vector field}. The expression $\mathcal{E}^V[W \phi]$ now does not need to vanish.

Furthermore, we can write
\begin{align*}
\int_{0}^{U(u_0)}\int_{S^2_{U,v}}(\partial_Uf)^2\,\sqrt{\det \slashed{g}}d\theta_*d\varphi_{*}dU=&\: \int_{-\infty}^{u_0}\int_{S^2_{u,v}}\frac{du}{dU}(\partial_uf)^2\,\sqrt{\det \slashed{g}}d\theta_*d\varphi_{*}du=\int_{\uline{H}_v}\frac{du}{dU}(\partial_uf)^2,\\
\int_{\widetilde{V}(v_0)}^0\int_{S^2_{u,v}}(\partial_{\widetilde{V}}f)^2\,\sqrt{\det \slashed{g}}d\theta_*d\varphi_{*}d\widetilde{V}= &\:\int_{v_0}^{\infty}\int_{S^2_{u,v}}\frac{d\tilde{v}}{d\widetilde{V}}(\partial_vf)^2\,\sqrt{\det \slashed{g}}d\theta_*d\varphi_{*}dv=\int_{{H}_u}\frac{dv}{d\widetilde{V}}(\partial_vf)^2.
\end{align*}

We can estimate in $\mathcal{M}\cap {\mathcal{D}_{u_0,v_0}}$, with $|u_0|,v_0\geq 1$ without loss of generality,
\begin{align*}
C_1 u^2\leq \frac{du}{dU}\leq C_2 u^2,\\
C_1 v^2\leq\frac{d\tilde{v}}{d\widetilde{V}}\leq C_2 v^2,
\end{align*}
for $C_1=C_1(a,M,u_0,v_0)>0$ and $C_2=C(a,M,u_0,v_0)>0$ uniform constants.

We rewrite the estimates above by using the following notation:
\begin{align}
\label{est:utoU}
\frac{du}{dU}\sim u^2,\\
\label{est:vtoV}
\frac{d\tilde{v}}{d\widetilde{V}}\sim v^2,
\end{align}
so that
\begin{align*}
\int_{\uline{H}_v}u^2(\partial_uf)^2 &\sim \int_{0}^{U(u_0)}\int_{S^2_{U,v}}(\partial_Uf)^2\,\sqrt{\det \slashed{g}}d\theta_*d\varphi_{*}dU,\\
\int_{{H}_u}v^2(\partial_vf)^2 & \sim \int_{\widetilde{V}(v_0)}^0\int_{S^2_{u,v}}(\partial_{\widetilde{V}}f)^2\,\sqrt{\det \slashed{g}}d\theta_*d\varphi_{*}d\widetilde{V}.
\end{align*}
Let use introduce the following natural $L^2$ norms:
\begin{align*}
||f||^2_{L^2(S^2_{u,v})}:=&\:\int_{S^2_{u,v}}f^2\,d\mu_{\slashed{g}}, \: \textnormal{where}\: d\mu_{\slashed{g}}:=\sqrt{\det \slashed{g}}d\theta_*d\varphi_{*},\\
||f||^2_{L^2(H_u)}:=&\:\int_{H_u}\frac{du}{dU}f^2,\\
||f||^2_{L^2(\uline{H}_v)}:=&\:\int_{\uline{H}_v}\frac{d\tilde{v}}{d\widetilde{V}}f^2.
\end{align*}
Now consider a compact subset $\mathcal{K}\subset \mathcal{M}'\cap\{r>\frac{e^2}{2M}\}$, such that moreover $\mathcal{K}\subset {\mathcal{D}_{u_0,v_0}}$. Then we define the following spacetime $L^2$ norms:
\begin{align*}
||f||^2_{L^2(\mathcal{K})}&:=\int_{\mathcal{K}\cap \mathcal{M}_{\rm int}} f^2,\\
||\partial f||^2_{L^2(\mathcal{K})}&:=\int_{\mathcal{K}\cap \mathcal{M}_{\rm int}} (\partial_{\widetilde{V}}f)^2+(\partial_U f)^2+|\snabla f|^2,\\
\end{align*}
where $\snabla$ denotes the induced covariant derivative on $S^2_{u,v}$.

We can in particular estimate
\begin{equation}
\label{est:energytoH1}
\begin{split}
||\partial f||_{L^2(\mathcal{K})}^2&=\int_{\mathcal{K}\cap \mathcal{M}_{\rm int}}(\partial_{\widetilde{V}}f)^2+(\partial_U f)^2+|\snabla f|^2\\
&\leq \int_{u_{\mathcal{K}}}^{u_0}\int_{\widetilde{V}_{\mathcal{K}}}^0\int_{\s^2}\left[(\partial_{\widetilde{V}}f)^2+(\partial_U f)^2+|\snabla f|^2\right]\,2\Omega^2\partial_{\widetilde{V}} \tilde{v} d\mu_{\slashed{g}}d\widetilde{V}du\\
&\leq C|\widetilde{V}_{\mathcal{K}}||u_{\mathcal{K}}|^2\sup_{\tilde{v}(\widetilde{V}_{\mathcal{K}})\leq v<\infty} \int_{\uline{H}_v\cap\{|u|\leq |u_{\mathcal{K}}|\}}(\partial_u f)^2+\Omega^2v^2|\snabla f|^2\\
&\quad+ C(u_0-u_{\mathcal{K}})^2\sup_{u_{\mathcal{K}}\leq u<u_0}\int_{{H}_u}v^2(\partial_vf)^2,
\end{split}
\end{equation}
where $\mathcal{K}\subset [u_{\mathcal{K}},u_0]\times[\widetilde{V}_{\mathcal{K}},0]\times \s^2$, with $-\infty<u_{\mathcal{K}}<u_0$, $\widetilde{V}(v_0)<\widetilde{V}_{\mathcal{K}}<0$ and $C=C(u_0,v_0)>0$. 

We define the weighted null-directed vector field $N_{p,q}$ in $\mathcal{M}_{\rm int}\cap {\mathcal{D}_{u_0,v_0}}$ as follows:
\begin{equation*}
N_{p,q}=|u|^p\underline{L}+v^qL=|u|^p\partial_u+v^q(\partial_v+b^A\partial_{\vartheta_A})=|\widetilde{u}|^p(\partial_{\widetilde{u}}+b^A\partial_{\widetilde{\vartheta}_A})+\tilde{v}^p\partial_{\tilde{v}},
\end{equation*}
with $0\leq p,q\leq 2$. In particular, in $(U,v,\vartheta)$ coordinates, we can express
\begin{equation*}
N_{p,q}=|u(U)|^p(r^2\Omega^2)\big|_{v=v_0,\,\theta_*=\frac{\pi}{2}}\partial_U+v^q(\partial_v+b^A\partial_{\vartheta_A}).
\end{equation*}
If $p\leq 2$, $N_{p,q}$ can be extended as a smooth vector field across $\mathcal{H}^+$ into $\mathcal{M}_{\rm ext}$. In $(\widetilde{u},\widetilde{V},\widetilde{\vartheta})$ coordinates, we have that
\begin{equation*}
N_{p,q}=|\widetilde{u}|^p(\partial_{\widetilde{u}}+b^A\partial_{\widetilde{\vartheta}_A})+\tilde{v}^q(\widetilde{V})(r^2\Omega^2)\big|_{u=u_0,\,\theta_*=\frac{\pi}{2}}\partial_{\widetilde{V}}.
\end{equation*}
If $q\leq 2$, $N_{p,q}$ can be extended as a smooth vector field beyond $\mathcal{CH}^+$ in $\mathcal{M}'_{\rm ext}$. The energy currents with respect to the constant $u$ and constant $v$ null hypersurfaces are given by
\begin{align*}
J^{N_{p,q}}[\phi]\cdot L=v^q\mathbf{T}(L,L)+|u|^p\mathbf{T}(L,\underline{L})=v^q(L\phi)^2+|u|^p\Omega^2|\snabla\phi|^2,\\
J^{N_{p,q}}[\phi]\cdot \underline{L}=v^q\mathbf{T}(L,\underline{L})+|u|^p\mathbf{T}(\underline{L},\underline{L})=|u|^p(\underline{L}\phi)^2+v^q\Omega^2|\snabla\phi|^2,
\end{align*}
where we inserted the expressions for $\mathbf{T}_{\alpha\beta}$ from Appendix \ref{sec:bulkest}.

In Appendix \ref{sec:bulkest} we show that the current $K^{N_{p,q}}$, compatible to $J^{{N}_{p,q}}$, is given by
\begin{equation*}
K^{N_{p,q}}[\phi]=K^{N_{p,q}}_{\textnormal{null}}[\phi]+K^{N_{p,q}}_{\textnormal{angular}}[\phi]+K^{N_{p,q}}_{\textnormal{mixed}}[\phi],
\end{equation*}
with
\begin{align}
\label{eq:errorterms1}
K^{N_{p,q}}_{\textnormal{null}}[\phi]=\:&\frac{1}{2}\Omega^{-2}(v^q\Omega\tr \chi+|u|^{p}\Omega \tr \underline{\chi})L\phi\underline{L}\phi,\\
\label{eq:errorterms2}
K^{N_{p,q}}_{\textnormal{angular}}[\phi]=\:&-\frac{1}{2}\left[-p|u|^{p-1}+qv^{q-1}+4\Omega(v^q\omega+|u|^p\underline{\omega})\right]|\snabla\phi|^2\\\nonumber
&+\left[v^q\Omega \hat{\chi}^{AB}+|u|^p\Omega\hat{\underline{\chi}}^{AB}\right](\partial_A\phi)(\partial_B\phi),\\
\label{eq:errorterms3}
K^{N_{p,q}}_{\textnormal{mixed}}[\phi]=\:&2[v^q(L\phi)-|u|^p(\underline{L}\phi)]\zeta^{\varphi_*}\partial_{\varphi_*}\phi
\end{align}

In extremal Kerr--Newman spacetimes with $|a|<a_c$, we consider moreover the vector field $Y_p$ in the region $\mathcal{M}_{\rm int}\cap (\{v\geq v_1\}\cup\{u\leq u_1\})$, which is defined by
\begin{equation*}
Y_p=|u|^pH.
\end{equation*}
From Section~ \ref{sec:hawkvf} it follows that $H$ is timelike in $\mathcal{M}_{\rm int}\cap (\{v\geq v_1\}\cup\{u\leq u_1\})$, if $|u_1|$ and $v_1$ are chosen suitably large. We use moreover that $H$ is a Killing vector field to easily obtain an expression for $K^{Y_p}$,
\begin{equation*}
\begin{split}
K^{Y_p}[\phi]=\:&g^{\alpha\beta}\nabla_{\beta}(J^{Y_p}_{\alpha})=g^{\alpha\beta}\nabla_{\beta}(|u|^p)J^H_{\alpha}[\phi]+|u|^pK^H[\phi]\\
=\:&\frac{p}{2}\Omega^{-2}|u|^{p-1}J^H[\phi]\cdot L \geq 0,
\end{split}
\end{equation*}
where non-negativity, in the case that $\phi$ is not axisymmetric, follows from the timelike character of $H$.

\section{Precise statements of the main theorems}
\label{sec:mainresults}
In this section we present more precise versions of the main results proved in this paper, which are stated in Section~ \ref{sec:roughversionsresults}. Let $\mathcal{M}$ denote extremal Kerr--Newman with $0\leq |a|\leq M$, unless otherwise stated. We first give a formulation of the standard global existence and uniqueness for the characteristic initial value problem for (\ref{eq:waveqkerr}) in $\mathcal{M}\cap {\mathcal{D}_{u_0,v_0}}$.

\begin{proposition}
\label{prop:wellposedness}
Let $\phi$ be a continuous function on the union of null hypersurfaces 
\begin{equation*}
\left(\mathcal{H}^+\cap\{v\geq v_0\}\right)\cup \uline{H}_{v_0},
\end{equation*}
such that the restriction to $\mathcal{H}^+$ and the restriction to $\uline{H}_{v_0}$ are smooth functions. Then there exists a unique, smooth extension of $\phi$ to $\mathcal{M}_{\rm int}\cup\mathcal{H}^+\cap {\mathcal{D}_{u_0,v_0}}$ that satisfies (\ref{eq:waveqkerr}) in extremal Kerr--Newman. We also denote this extension by $\phi$. We refer to the restriction $\phi|_{\left(\mathcal{H}^+\cap\{v\geq v_0\}\right)\cup \uline{H}_{v_0}}$ as characteristic initial data. If $\phi|_{\left(\mathcal{H}^+\cap\{v\geq v_0\}\right)\cup \uline{H}_{v_0}}$ is axisymmetric, the extension $\phi$ to $\mathcal{M}_{\rm int}\cup \mathcal{H}^+\cap {\mathcal{D}_{u_0,v_0}}$ must also be axisymmetric.
\end{proposition}

Observe that Proposition \ref{prop:wellposedness} does not provide any information about the asymptotic behaviour of $\phi$ towards $\mathcal{CH}^+$. We will state in the subsections below further quantitative and qualitative properties of $\phi$, relating to boundedness and extendibility of $\phi$ and its derivatives beyond $\mathcal{CH}^+$, under the assumption of suitable additional decay requirements along $\mathcal{H}^+$.

\subsection{Energy estimates along null hypersurfaces}
Consider solutions $\phi$ to (\ref{eq:waveqkerr}) that arise from the characteristic initial data in Proposition \ref{prop:wellposedness}. We will first show that we can prove boundedness of weighted $L^2$ norms for $\phi$ along null hypersurfaces, under additional assumptions on suitable initial $L^2$ norms along $\mathcal{H}^+$. We will treat separately the case of axisymmetric solutions $\phi$ on extremal Kerr--Newman with $0\leq |a|\leq M$, and the case of general solutions $\phi$ on slowly rotating extremal Kerr--Newman, with $0\leq |a|<a_c$. In the next section, we will give an overview of the theorems regarding $L^{\infty}$ estimates for $\phi$. Unless specified differently, we consider (\ref{eq:waveqkerr}) on an extremal Kerr--Newman background with $0\leq |a|\leq M$.

\begin{theorem}
\label{thm:eestaxisymm}
Take $0<q\leq 2$. Let $\phi$ be a solution to (\ref{eq:waveqkerr}) corresponding to axisymmetric initial data from Proposition \ref{prop:wellposedness} satisfying
\begin{equation*}
E_q[\phi]:=\int_{\mathcal{H}^+\cap\{v\geq v_0\}} v^q(L\phi)^2+|\snabla \phi|^2+ \int_{\underline{H}_{v_0}} |u|^2(\uline{L}\phi)^2+\Omega^2|\snabla\phi|^2<\infty.
\end{equation*}

Then there exists a constant $C=C(a,M,u_0,v_0,q)>0$ such that for all $H_u$ and $\underline{H}_v$
\begin{equation*}
\begin{split}
&\int_{H_u}v^q(L\phi)^2+|u|^2|\snabla \phi|^2 +\int_{\underline{H}_v}  |u|^2(\uline{L}\phi)^2+\Omega^2v^q|\snabla\phi|^2\leq CE_q[\phi].
\end{split}
\end{equation*}
\end{theorem}
Theorem \ref{thm:eestaxisymm} is proved in Proposition \ref{prop:mainenergyestimate}. Theorem \ref{thm:H1boundv1} follows immediately by using moreover Theorem \ref{thm:pointwiseboundkn} below and the estimate (\ref{est:energytoH1}).

\begin{theorem}
\label{thm:energyestsmallakn}
Let $\phi$ be a solution to (\ref{eq:waveqkerr}), with $|a|<a_c$, corresponding to initial data from Proposition \ref{prop:wellposedness} satisfying
\begin{equation*}
E_{q}[\phi]:=\int_{\mathcal{H}^+\cap\{v\geq v_0\}} v^q(L\phi)^2+|\snabla \phi|^2+ \int_{\underline{H}_{v_0}} |u|^2(\uline{L}\phi)^2+\Omega^2|\snabla\phi|^2<\infty,
\end{equation*}
where $0< q\leq 2$.

Let $0\leq p<2$ and let $\epsilon>0$ be arbitrarily small. Then there exists a constant \\$C=C(a,M,u_0,v_0,p,q,\epsilon)>0$, such that for all $H_u$ and $\underline{H}_v$,
\begin{equation*}
\begin{split}
&\int_{H_u}v^{q-\epsilon}(L\phi)^2+|u|^p|\snabla \phi|^2 +\int_{\underline{H}_v}  |u|^p(\uline{L}\phi)^2+\Omega^2v^{q-\epsilon}|\snabla\phi|^2\leq CE_{q}[\phi].
\end{split}
\end{equation*}
\end{theorem}

Theorem \ref {thm:energyestsmallakn} is proved in Proposition \ref{energyestsmallakn}. 

We can remove the $\epsilon$ in Theorem \ref {thm:energyestsmallakn} at the cost of losing derivatives on the right-hand side of the estimate.
\begin{theorem}
\label{thm:improvedestoutgoingenergy}
Let $\phi$ be a solution to (\ref{eq:waveqkerr}), with $|a|<a_c$, corresponding to initial data from Proposition \ref{prop:wellposedness} and denote
\begin{equation*}
D=||\partial_U\phi||^2_{L^{\infty}\left(\uline{H}_{v_0}\right)}+||\snabla \phi||^2_{L^{\infty}\left(\uline{H}_{v_0}\right)}.
\end{equation*}

Assume further that
\begin{equation*}
\begin{split}
E_{\rm extra; \eta}[\phi]:=&\: \int_{\mathcal{H}^+\cap \{v\geq v_0\}}v^2(L\phi)^2\\
&+\sum_{0\leq j_1+j_2\leq 4}\int_{\mathcal{H}^+\cap\{v\geq v_0\}} v^{\eta} \left(|\snabla^{j_1} L^{j_2+1}\phi|^2+ |\snabla^{j_1+1} L^{j_2}\phi|^2+|\snabla^{j_1+2} L^{j_2}\phi|^2\right)<\infty,
\end{split}
\end{equation*}
for $\eta>0$ arbitrarily small. Then there exists a constant $C=C(a,M,v_0,u_0,\eta)>0$ such that,
\begin{equation*}
\int_{H_u}v^2(L\phi)^2+u^2\Omega^2|\snabla\phi|^2+\int_{\underline{H}_v}v^2\Omega^2|\snabla\phi|^2\leq C(D+E_{\rm extra; \eta}[\phi]).
\end{equation*}

\end{theorem}
Theorem \ref{thm:improvedestoutgoingenergy} is proved in Corollary \ref{cor:improvedestoutgoingenergy}. Theorem \ref{thm:H1boundv2} now follows from Theorem \ref{thm:improvedestoutgoingenergy}, combined with Theorem \ref{thm:pointwiseboundkn} below and the estimate (\ref{est:energytoH1}).

\subsection{Pointwise estimates and continuous extendibility beyond $\mathcal{CH}^+$}
We can use the energy estimates in the subsection above to obtain $L^{\infty}$ estimates in $\mathcal{M}\cap {\mathcal{D}_{u_0,v_0}}$, and we can moreover show that $\phi$ is continuously extendible beyond $\mathcal{CH}^+$. Here, we treat the restriction to axisymmetric $\phi$ and the restriction to slowly rotating extremal Kerr--Newman simultaneously.
\begin{theorem}
\label{thm:pointwiseboundkn}
Either take $0<p<2$ and let $\phi$ be a solution to (\ref{eq:waveqkerr}), with $|a|<a_c$, corresponding to initial data from Proposition \ref{prop:wellposedness} without any symmetry assumptions, or take $0\leq p\leq 2$ and let $\phi$ be a solution to (\ref{eq:waveqkerr}), with $0\leq |a|\leq M$, corresponding to axisymmetric initial data from Proposition \ref{prop:wellposedness}.

Assume that, for $\epsilon>0$ arbitrarily small,
\begin{align*}
\sum_{|k|\leq 2}\int_{S^2_{-\infty,v}}|\snabla^k\phi|^2<&\infty,\\
\sum_{0\leq j_1+j_2\leq 4}\int_{\mathcal{H}^+\cap\{v\geq v_0\}}v^{\epsilon}|\snabla^{j_1}L^{j_2+1}\phi|^2+|\snabla^{j_1+1} L^{j_2}\phi|^2<&\infty.
\end{align*}
Then there exists a constant $C=C(a,M,v_0,u_0,\epsilon)>0$ such that
\begin{equation*}
\begin{split}
\phi^2(u,v,\theta_*,{\varphi_*})\leq &\: \sum_{|k|\leq 2}\int_{S^2_{-\infty,v}}|\snabla^k\phi|^2+C|u|^{1-p}\sum_{0\leq j_1+j_2\leq 4}\int_{\mathcal{H}^+\cap\{v\geq v_0\}}v^{\epsilon}|\snabla^{j_1}L^{j_2+1}\phi|^2+|\snabla^{j_1+1} L^{j_2}\phi|^2.
\end{split}
\end{equation*}
\end{theorem}

Theorem \ref{thm:pointwiseboundkn} follows from Proposition \ref{pointwiseboundkn}.  

\begin{theorem}
\label{thm:C0extension}
Let $\phi$ be a solution to (\ref{eq:waveqkerr}), with $|a|<a_c$, corresponding to initial data from Proposition \ref{prop:wellposedness} without any symmetry assumptions, or let $\phi$ be a solution to (\ref{eq:waveqkerr}), with $0\leq |a|\leq M$, corresponding to axisymmetric initial data from Proposition \ref{prop:wellposedness}.

Assume furthermore that
\begin{equation}
\label{eq:requireddecayext}
\sum_{0\leq j_1+j_2\leq 4}\int_{\mathcal{H}^+\cap\{v\geq v_0\}}v^{q}|\snabla^{j_1}L^{j_2+1}\phi|^2+|\snabla^{j_1+1} L^{j_2}\phi|^2<\infty,
\end{equation}
for some $q>1$.

Then $\phi$ can be extended as a $C^0$ function beyond $\mathcal{CH}^+$.
\end{theorem}

Theorem \ref{thm:C0extension} is proved in Proposition \ref{prop:C0extension}. We can infer Theorems \ref{thm:linftyboundv1} and \ref{thm:linftyboundv2} from Theorem \ref{thm:pointwiseboundkn} and Theorem \ref{thm:C0extension}. 

As Theorem \ref{thm:linftyboundv0} is formulated in terms of Cauchy initial data for $\phi$ on an asymptotically flat hypersurface $\Sigma$ in extremal Kerr, we also need to appeal to the decay estimates in the exterior of extremal Kerr. In particular, boundedness of a non-degenerate energy and $\tau^{-1-\epsilon}$-decay of a degenerate $T$-energy for axisymmetric solutions, with respect to a suitable spacelike foliation $\Sigma_{\tau}$ of the extremal Kerr exterior, which are proved in Theorems 2 and 3 of \cite{are6}, are sufficient to show that (\ref{eq:requireddecayext}) holds for suitable Cauchy data for $\phi$, so that Theorem \ref{thm:linftyboundv0} can be viewed as a corollary of Theorem \ref{thm:linftyboundv1}.

\begin{theorem}
\label{thm:alphaholdercont}
Let $\phi$ be a solution to (\ref{eq:waveqkerr}) corresponding to initial data from Proposition \ref{prop:wellposedness} without any symmetry assumptions. Let $k\in \N_0$ and denote
\begin{equation*}
\begin{split}
D_{2k}:=&\:\sum_{j_1+j_2+2j_3+j_4\leq 2k}||\partial_UL^{j_1}\uline{L}^{j_2}Q^{j_3}\Phi^{j_4}\phi||^2_{L^{\infty}\left(\uline{H}_{v_0}\right)}+||\snabla L^{j_1}\uline{L}^{j_2}Q^{j_3}\Phi^{j_4}\phi||^2_{L^{\infty}\left(\uline{H}_{v_0}\right)}\\
&+\sum_{j_1+2j_2\leq n}||\partial_U\Phi^{j_1+1}Q^{j_2}\phi||^2_{L^{\infty}\left(\uline{H}_{v_0}\right)}+||\snabla \Phi^{j_1+1}Q^{j_2}\phi||^2_{L^{\infty}\left(\uline{H}_{v_0}\right)}.
\end{split}
\end{equation*}
Assume that
\begin{equation*}
\int_{S^2_{-\infty,v}}v^4(L\phi)^2+v^2|\snabla\phi|^2+v^2|\snabla^2\phi|^2\,d\mu_{\slashed{g}}<\infty.
\end{equation*}
\begin{itemize}
\item[(i)]
Let $|a|<a_c$ and assume also that
\begin{equation*}
\sum_{0\leq j_1+j_2\leq 4}\int_{\mathcal{H}^+\cap\{v\geq v_0\}} v \left(|\snabla^{j_1} L^{j_2+1}\phi|^2+ |\snabla^{j_1+1} L^{j_2}\phi|^2+|\snabla^{j_1+2} L^{j_2}\phi|^2\right)<\infty.
\end{equation*}
Then we can estimate
\begin{equation*}
\begin{split}
\int_{S^2_{u,v}}&v^4(L\phi)^2(u,v,\theta_*,{\varphi_*})\,d\mu_{\slashed{g}}\\
\leq \:& \int_{S^2_{-\infty,v}}v^4(L\phi)^2\,d\mu_{\slashed{g}}+C\int_{S^2_{-\infty,v}}v^2|\snabla\phi|^2+v^2|\snabla^2\phi|^2\,d\mu_{\slashed{g}}\\
&+Cv^{\epsilon}\left[D_2+\sum_{0\leq j_1+j_2\leq 4}\int_{\mathcal{H}^+\cap\{v\geq v_0\}} v \left(|\snabla^{j_1} L^{j_2+1}\phi|^2+ |\snabla^{j_1+1} L^{j_2}\phi|^2+|\snabla^{j_1+2} L^{j_2}\phi|^2\right)\right].
\end{split}
\end{equation*}
\item[(ii)]
Restrict to axisymmetric data from Proposition \ref{prop:wellposedness} and assume that
\begin{equation*}
\sum_{0\leq j_1+j_2\leq 4}\int_{\mathcal{H}^+\cap\{v\geq v_0\}} v^{1+\epsilon} \left(|\snabla^{j_1} L^{j_2+1}\phi|^2+ |\snabla^{j_1+1} L^{j_2}\phi|^2+|\snabla^{j_1+2} L^{j_2}\phi|^2\right)<\infty,
\end{equation*}
for $\epsilon>0$ arbitrarily small. Then we can estimate
\begin{equation*}
\begin{split}
&\int_{S^2_{u,v}}v^4(L\phi)^2(u,v,\theta_*,{\varphi_*})\,d\mu_{\slashed{g}}\\
\leq \:& \int_{S^2_{-\infty,v}}v^4(L\phi)^2\,d\mu_{\slashed{g}}+C\int_{S^2_{-\infty,v}}v^2|\snabla\phi|^2+v^2|\snabla^2\phi|^2d\mu_{\slashed{g}}\\
&+C\log\left(\frac{v+|u|}{|u|}\right)\\
&\quad\cdot \Bigg[D_2+\sum_{0\leq j_1+j_2\leq 4}\int_{\mathcal{H}^+\cap\{v\geq v_0\}} v^{1+\epsilon} \left(|\snabla^{j_1} L^{j_2+1}\phi|^2+ |\snabla^{j_1+1} L^{j_2}\phi|^2+|\snabla^{j_1+2} L^{j_2}\phi|^2\right)\Bigg].
\end{split}
\end{equation*}
\item[(iii)]
Either restrict to axisymmetric data in from Proposition \ref{prop:wellposedness}, or let $|a|<a_c$. Assume that
\begin{align*}
\int_{S^2_{-\infty,v}}v^4\sum_{j_1+j_2\leq 2}(LL^{j_1}\uline{L}^{j_2}\phi)^2+v^4\sum_{j_1+j_2\leq 1} \sum_{\Gamma \in \{\Phi,\Phi^2,T^2,Q\}}(LL^{j_1}\uline{L}^{j_2} \Gamma \phi)^2\,d\mu_{\slashed{g}}<&\infty,\\
\sum_{j_1+j_2\leq 2}\int_{S^2_{-\infty,v}}v^2\left(|\snabla L^{j_1}\uline{L}^{j_2}\phi|^2+|\snabla^2 L^{j_1}\uline{L}^{j_2}\phi|^2\right)\,d\mu_{\slashed{g}}<&\infty,\\
\sum_{0\leq j_1+j_2\leq 8}\int_{\mathcal{H}^+\cap\{v\geq v_0\}} v \left(|\snabla^{j_1} L^{j_2+1}\phi|^2+ |\snabla^{j_1+1} L^{j_2}\phi|^2+|\snabla^{j_1+2} L^{j_2}\phi|^2\right)<&\infty.
\end{align*}
Then $\phi$ can be extended in $C^{0,\alpha}$, for all $\alpha<1$.
\end{itemize}
\end{theorem}
Theorem \ref{thm:alphaholdercont} follows from Proposition \ref{prop:decayLphi} and Proposition \ref{prop:c0alphaextendibility} and implies Theorem \ref{thm:c11boundphiv1} and Theorem \ref{thm:c11boundphiv2}.

\section{Energy estimates for axisymmetric solutions}
\label{sec:energyestimatesaxisymm}
We will first restrict to axisymmetric solutions to (\ref{eq:waveqkerr}) on extremal Kerr--Newman spacetimes with $0\leq |a|\leq M$. \textbf{In this section we will always use $\phi$ to denote a solution to} (\ref{eq:waveqkerr})\textbf{, with $0\leq |a|\leq M$, corresponding to \uline{axisymmetric} initial data from Proposition} \ref{prop:wellposedness}.

We will frequently make use of a Gr\"onwall-type lemma.
\begin{lemma}
\label{lm:gronwall}
Let $-\infty \leq u_1<u_2\leq \infty$ and $-\infty \leq v_1<v_2\leq \infty$. Consider continuous, non-negative functions $f,g: [u_1,u_2]\times [v_1,v_2] \to \R$ and continuous, non-negative functions $h: [u_1,u_2]\to \R$ and $k: [v_1,v_2]\to \R$. Suppose
\begin{equation}
\label{eq:gronassumption}
f(u,v)+g(u,v)\leq A+ B\left[\int_{u_1}^u h(u')f(u',v)\,du'+\int_{v_1}^v k(v')g(u,v')\,dv'\right],
\end{equation}
for all $u\in [u_1,u_2]$ and $v\in [v_1,v_2]$, where $A,B>0$ are constants. Then:
\begin{equation}
\label{eq:gronba}
f(u,v)+g(u,v)\leq (1+\eta)A e^{\beta B\left[\int_{u_1}^u h(u')\,du'+\int_{v_1}^v k(v')\,dv'\right]},
\end{equation}
for all $u\in [u_1,u_2]$ and $v\in [v_1,v_2]$, where $\eta>0$ can be taken arbitrarily small and $\beta\geq \frac{2(1+\eta)}{\eta}$.
\end{lemma}
\begin{proof}
See Section~ 4 of \cite{gajic2015}.
\end{proof}

We can use the vector field $N_{p,q}$ with $p=2$, defined in Section~ \ref{sec:vfmethod}, as a vector field multiplier to obtain weighted energy estimates.
\begin{proposition}
\label{prop:mainenergyestimate}
Fix $p=2$ and let $0<q\leq 2$. There exists a constant \\$C=C(a,M,u_0,v_0,q)>0$ such that for all $H_u$ and $\underline{H}_v$ in ${{\mathcal{D}_{u_0,v_0}}}$
\begin{equation}
\label{eq:weightfluxbound}
\begin{split}
\int_{H_u} &J^{N_{2,q}}[\phi]\cdot L+\int_{\underline{H}_v} J^{N_{2,q}}[\phi]\cdot \underline{L}\\
\leq \:& C\left[ \int_{\mathcal{H}^+\cap\{v\geq v_0\}} J^{N_{2,q}}[\phi]\cdot L+ \int_{\underline{H}_{v_0}} J^{N_{2,q}}[\phi]\cdot \underline{L}\right]=:CE_q[\phi].
\end{split}
\end{equation}
\end{proposition}
\begin{proof}
By applying the divergence theorem from Section~ \ref{sec:vfmethod} in ${{\mathcal{D}_{u_0,v_0}}}$, we can estimate
\begin{equation}
\label{eq:stokesaxikn}
\begin{split}
\int_{H_u}& J^{N_{2,q}}[\phi]\cdot L+\int_{\underline{H}_v} J^{N_{2,q}}[\phi]\cdot \underline{L}\\
=\:& \int_{\mathcal{H}^+\cap\{v\geq v_0\}} J^{N_{2,q}}[\phi]\cdot L+ \int_{\underline{H}_{v_0}} J^{N_{2,q}}[\phi]\cdot \underline{L}-\int_{{{\mathcal{D}_{u_0,v_0}}}} K^{N_{p,q}}_{\textnormal{null}}[\phi]+K^{N_{p,q}}_{\textnormal{angular}}[\phi]+K^{N_{p,q}}_{\textnormal{mixed}}.
\end{split}
\end{equation}
By the assumption that $\phi$ is axisymmetric, we have that $K^{N_{p,q}}_{\textnormal{mixed}}=0$.

We first consider $K^{N_{p,q}}_{\textnormal{null}}[\phi]$ and apply the estimates of Section~ \ref{sec:estimatesmetricconn} to (\ref{eq:errorterms1}) to find that
\begin{equation*}
\Omega^2|K^{N_{p,q}}_{\textnormal{null}}[\phi]|\leq C (v+|u|)^{-2}(v^q-|u|^p)|L\phi||\underline{L}\phi|.
\end{equation*}
By applying a (weighted) Cauchy--Schwarz inequality, we can further estimate for $\eta>0$,
\begin{equation*}
\begin{split}
v^q(v+|u|)^{-2}|L\phi||\underline{L}\phi|\leq &\: Cv^q(v+|u|)^{-1-q-\eta}|u|^p(\uline{L}\phi)^2+C|u|^{-p}(v+|u|)^{q+\eta-3}v^q(L\phi)^2\\
\leq &\: Cv^q\sup_{u\leq u' \leq u_0}\left[(v+|u'|)^{-1-q-\eta}\right]|u|^p(\underline{L}\phi)^2\\
&+C|u|^{-p}\sup_{v_0\leq v' \leq v}\left[(v'+|u|)^{q+\eta-3}\right]v^q(L\phi)^2\\
\leq &\: C v^{-1-\eta}|u|^p(\underline{L}\phi)^2+C|u|^{q+\eta-p-3}v^q(L\phi)^2,
\end{split}
\end{equation*}
for $\eta<3-q$.

Similarly, by reversing the roles of $u$ and $v$, we obtain
\begin{equation*}
|u|^p(v+|u|)^{-2}|L\phi||\underline{L}\phi|\leq C|u|^{-1-\eta}v^q(L\phi)^2+Cv^{p+\eta-q-3}|u|^p(\underline{L}\phi)^2,
\end{equation*}
for $\eta<3-p$.

We will now estimate $K^{N_{p,q}}_{\textnormal{angular}}$ by applying the estimates of Section~ \ref{sec:estimatesmetricconn} to (\ref{eq:errorterms2}). We obtain
\begin{equation}
\label{eq:KNangularkerrn1}
\begin{split}
K^{N_{p,q}}_{\textnormal{angular}}=\:&-\frac{1}{2}\left[qv^{q-1}-p|u|^{p-1}+4\Omega\omega(v^q-|u|^p)\right]|\snabla\phi|^2\\
&+(v^q+|u|^p)\log(v+|u|)\mathcal{O}((v+|u|)^{-2})|\snabla\phi|^2.
\end{split}
\end{equation}
Recall from (ii) of Theorem \ref{thm:estconncoef} that we can expand
\begin{equation*}
4\Omega \omega=-\frac{2}{v+|u|}+\log (v+|u|)\mathcal{O}((v+|u|)^{-2}).
\end{equation*}
Consequently, we can rewrite (\ref{eq:KNangularkerrn1}) to obtain
\begin{equation}
\label{eq:KNangularkerrn2}
\begin{split}
K^{N_{p,q}}_{\textnormal{angular}}=\:&-\frac{1}{2}\left[\left(q-2\frac{v^{\frac{q}{2}}+|u|^{\frac{p}{2}}}{v+|u|}v^{1-\frac{q}{2}}\right)v^{q-1}+\left(2\frac{v^{\frac{q}{2}}+|u|^{\frac{p}{2}}}{v+|u|}|u|^{1-\frac{p}{2}}-p\right)|u|^{p-1}\right]|\snabla\phi|^2\\
&+(v^q+|u|^p)\log(v+|u|)\mathcal{O}((v+|u|)^{-2})|\snabla\phi|^2.
\end{split}
\end{equation}
First, let $0\leq p<2$. Then the term between square brackets in front of $|\snabla \phi|^2$ will become positive in the region $|u|>v$, as we approach $\mathcal{H}^+$, which means that $K^{N_{p,q}}_{\textnormal{angular}}$ will be negative, and we are not able to control it. We therefore restrict to $p=2$. 

If $p=2$ and $q<2$, the term inside the square brackets is negative for suitably large $v$, so we can estimate
\begin{equation*}
K^{N_{p,q}}_{\textnormal{angular}}\geq Cv^{q-1}|\snabla\phi|^2+\mathcal{O}((v^q+|u|^2)(v+|u|)^{-2}\log(v+|u|))|\snabla\phi|^2.
\end{equation*}

If $p=2$ and $q=2$, a cancellation occurs in the leading-order terms between square brackets, so we can estimate
\begin{equation*}
K^{N_{p,q}}_{\textnormal{angular}}=(v^2+|u|^2)\log(v+|u|)\mathcal{O}((v+|u|)^{-2})|\snabla\phi|^2.
\end{equation*}

If we fix $p=2$, we can therefore estimate for all $0\leq q\leq 2$,
\begin{equation*}
\begin{split}
\Omega^2 K^{N_{p,q}}_{\textnormal{angular}}\geq&\: (v^q+|u|^2)\log(v+|u|)\mathcal{O}((v+|u|)^{-2})\Omega^2|\snabla\phi|^2\\
\geq&\: -C_{\epsilon}|u|^{-2+\epsilon}v^q\Omega^2|\snabla\phi|^2-C_{\epsilon}v^{-2+\epsilon}|u|^2\Omega^2|\snabla\phi|^2,
\end{split}
\end{equation*}
with $\epsilon>0$ arbitrarily small and $C_{\epsilon}=C_{\epsilon}(M,u_0,v_0,\epsilon)>0$. We will fix $0<\epsilon<1$.

We combine the estimates above for $K^{N_{2,q}}_{\textnormal{null}}$ and $K^{N_{2,q}}_{\textnormal{angular}}$ to obtain, for $0\leq q\leq 2$,
\begin{equation*}
\begin{split}
-\Omega^2 (K^{N_{2,q}}_{\textnormal{angular}}+K^{N_{2,q}}_{\textnormal{null}})\leq \:& C\Big[(v^{-1-\eta}+v^{\eta-q-1})|u|^2(\underline{L}\phi)^2+(|u|^{-1-\eta}+|u|^{q+\eta-5})v^q(L\phi)^2\\
&+|u|^{-2+\epsilon}v^q\Omega^2|\snabla\phi|^2+v^{-2+\epsilon}|u|^2\Omega^2|\snabla\phi|^2\Big].
\end{split}
\end{equation*}
Finally, we can apply Lemma \ref{lm:gronwall} with the choices
\begin{align*}
A=\:& \int_{\mathcal{H}^+\cap\{v\geq v_0\}} J^{N_{2,q}}[\phi]\cdot L+ \int_{\uline{H}_{v_0}} J^{N_{2,q}}[\phi]\cdot \underline{L},\\
f(u,v)=\:&\int_{H_{u}}v^q(L\phi)^2+|u|^2\Omega^2|\snabla\phi|^2,\\
g(u,v)=\:&\int_{\underline{H}_v}|u|^2(\underline{L}\phi)^2+v^q\Omega^2|\snabla \phi|^2,\\
h(u)=\:&|u|^{-1-\eta}+|u|^{q+\eta-5}+|u|^{-2+\epsilon},\\
k(v)=\:&v^{-1-\eta}+v^{\eta-q-1}+v^{-2+\epsilon},
\end{align*}
where we use that $h$ and $k$ are integrable for $0<\eta<\min\{q,1\}$ and $0<\epsilon<1$, to arrive at the estimate in the proposition. We therefore need the restriction $q>0$ if $p=2$.
\end{proof}

We have now proved Theorem \ref{thm:eestaxisymm}.

\section{Energy estimates in slowly rotating extremal Kerr--Newman}
\label{sec:energyestimatesslowlyrot}
We now drop the axisymmetry assumptions on solutions to (\ref{eq:waveqkerr}) on extremal Kerr--Newman. We do however restrict to the subfamily of slowly rotating extremal Kerr--Newman spacetimes, with $0\leq |a|<a_c$; see Section~ \ref{sec:vfmethod}.

\textbf{In this section we will always use $\phi$ to denote a solution to} (\ref{eq:waveqkerr})\textbf{, with \\$0\leq |a|<a_c$, corresponding to initial data from Proposition} \ref{prop:wellposedness} \textbf{\uline{without symmetry assumptions}}.

Even without an axisymmetry assumption on $\phi$, we can still obtain energy estimates with respect to vector fields $N_{p,q}$ if we restrict to subsets of ${\mathcal{D}_{u_0,v_0}}$ with a finite spacetime volume. We introduce the hypersurfaces $\underline{\gamma}_{\alpha}$ and $\gamma_{\beta}$, with $\alpha\geq 1$ and $\beta \geq 1$, such that
\begin{align*}
\underline{\gamma}_{\alpha}:=\{(u,v,\theta_*,{\varphi_*})\in {{\mathcal{D}_{u_0,v_0}}}\,:\,f_{\alpha}(u,v)=0\},\\
\gamma_{\beta}:=\{(u,v,\theta_*,{\varphi_*})\in {{\mathcal{D}_{u_0,v_0}}}\,:\,\underline{f}_{\beta}(u,v)=0\}.
\end{align*}
We define $f_{\alpha}(u,v)$ as follows:
\begin{align*}
f_{\alpha}(u,v)=\:&|u|-v^{\alpha},\quad |u|>|u_1|,\\
=\:&h_{\alpha}(u,v),\quad |u|\leq u_1,
\end{align*}
where $|u_1|$ is taken suitably large, such that $-g(df_{\alpha},df_{\alpha})\geq C$, for $|u|>|u_1|$, with $C>0$ a constant. Moreover, we can choose $h_{\alpha}$ such that $f_{\alpha}$ is a smooth function on $(-\infty,u_0]\times[v_0,\infty)$ and for all $(u,v)$ such that $h_{\alpha}(u,v)=0$, we can uniformly bound $-g(dh_{\alpha},dh_{\alpha})(u,v)\geq C$.

\begin{figure}[h!]
\begin{center}
\includegraphics[width=2.5in]{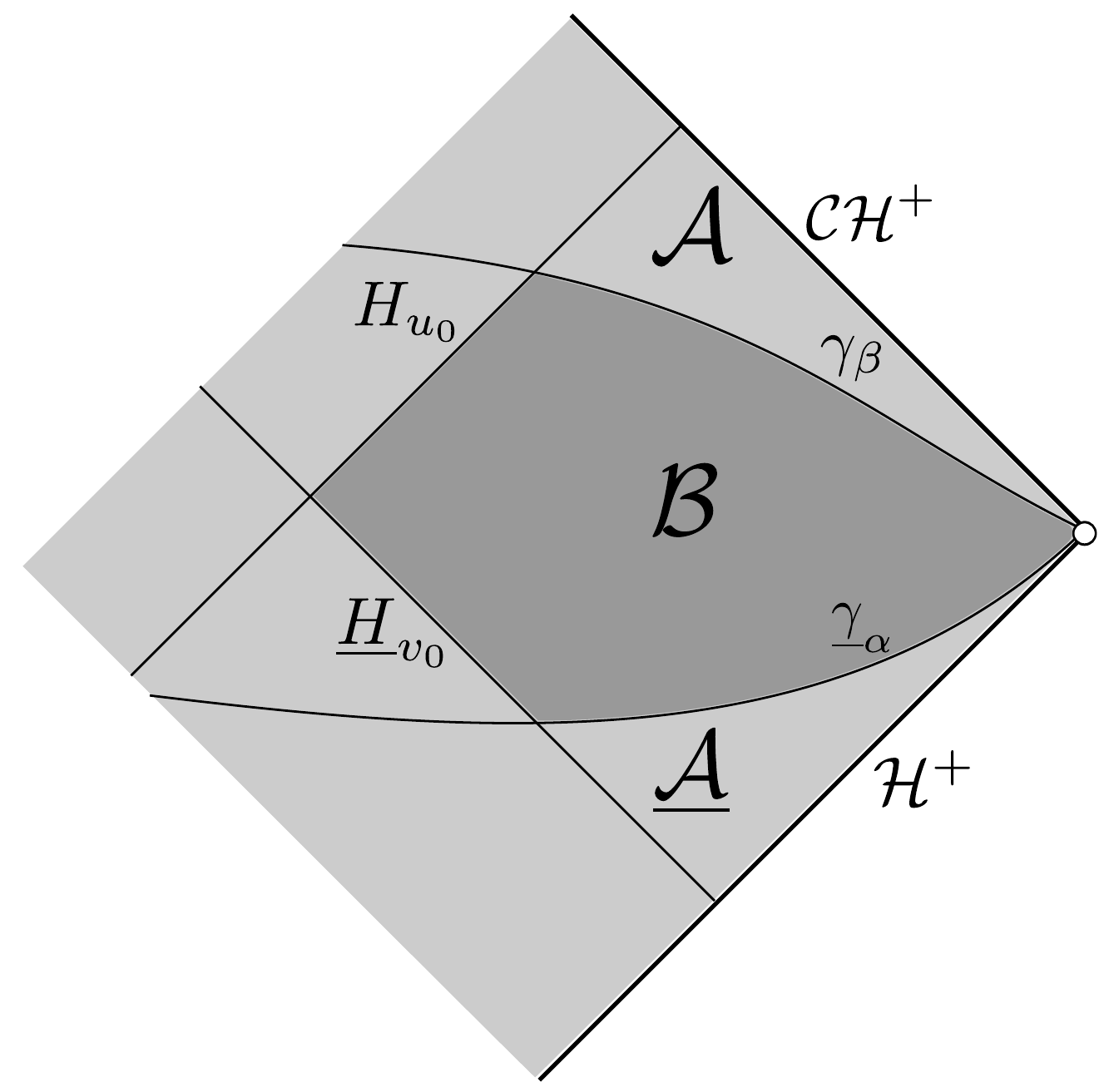}
\caption{\label{fig:finitevolume} }
\end{center}
\end{figure}

We define $\underline{f}_{\beta}(u,v)$ as follows:\\
\\
\begin{align*}
\underline{f}_{\beta}(u,v)=\:&v-|u|^{\beta},\quad v>v_1,\\
=\:&\underline{h}_{\beta}(u,v),\quad v\leq v_1,
\end{align*}
where $v_1$ is taken suitably large, such that $-g(d\underline{f}_{\beta},d\underline{f}_{\beta})\geq C$, for $v>v_1$, with $C>0$ a constant. Moreover, we can choose $\underline{h}_{\beta}$ such that $\underline{f}_{\beta}$ is a smooth function on $(-\infty,u_0]\times[v_0,\infty)$ and for all $(u,v)$ such that $h_{\beta}(u,v)=0$, we can uniformly bound $-g(dh_{\beta},dh_{\beta})(u,v)\geq C$.

Consequently, $\underline{\gamma}_{\alpha}$ and $\gamma_{\beta}$ are spacelike hypersurfaces. Denote
\begin{align*}
\uline{\mathcal{A}}&:=J^-(\underline{\gamma}_{\alpha})\cap {\mathcal{D}_{u_0,v_0}},\\
\mathcal{A}&:=J^+({\gamma}_{\beta})\cap {\mathcal{D}_{u_0,v_0}},\\
\mathcal{B}&:=J^+(\underline{\gamma}_{\alpha})\cap J^-({\gamma}_{\beta})\cap {\mathcal{D}_{u_0,v_0}}.
\end{align*}
It is easy to verify that the spacetime volumes of $\mathcal{A}$ and $\uline{\mathcal{A}}$ are finite, if we take $\alpha>1$ and $\beta>1$.

\subsection{Energy estimates in $\uline{\mathcal{A}}\cup \mathcal{A}$}
We first consider energy estimates with respect to $N_{p,q}$. Since $\phi$ is no longer assumed to be axisymmetric, $K_{\textnormal{mixed}}[\phi]$ does not necessarily vanish. To deal with a non-vanishing $K_{\textnormal{mixed}}[\phi]$, we first consider the regions $\mathcal{A}$ and $\uline{\mathcal{A}}$ of finite spacetime volume.
\begin{proposition}
\label{prop:finitevolumeenergyestimates}
\hspace{1pt}
\begin{itemize}\setlength\itemsep{1em}
\item[(i)]
Let $p=2$ and $0<q<2$. Fix $\alpha>1$. Then there exists a constant \\$C=C(a,M,u_0,v_0,q,\alpha)>0$ such that
\begin{equation}
\label{eq:weightfluxboundfinitevolume1}
\begin{split}
\int_{H_u\cap \uline{\mathcal{A}}}& J^{N_{2,q}}[\phi]\cdot L+\int_{\underline{H}_v\cap \uline{\mathcal{A}}} J^{N_{2,q}}[\phi]\cdot \underline{L}+\int_{\uline{\gamma}_{\alpha}} J^{N_{p,q}}[\phi]\cdot n_{\uline{\gamma}_{\alpha}}\\
\leq \:& C\left[ \int_{\mathcal{H}^+\cap\{v\geq v_0\}} J^{N_{2,q}}[\phi]\cdot n_{\mathcal{H}^+}+ \int_{\underline{H}_{v_0}\cap\{|u|\geq |u_{\underline{\gamma}_{\alpha}}|(v_0)\}} J^{N_{2,q}}[\phi]\cdot \underline{L}\right].
\end{split}
\end{equation}

\item[(ii)]
Let $p=2$ and $0< q\leq 2$, or let $0\leq p<2$ and $0\leq q< 2$. Fix $\beta>1$. Then there exists a constant $C=C(a,M,u_0,v_0,p,q,\beta)>0$ such that
\begin{equation}
\label{eq:weightfluxboundfinitevolume2}
\begin{split}
\int_{H_u\cap \mathcal{A}}& J^{N_{p,q}}[\phi]\cdot L+\int_{\underline{H}_v\cap \mathcal{A}} J^{N_{p,q}}[\phi]\cdot \underline{L}\\
\leq \:& C\left[\int_{\gamma_{\beta}} J^{N_{p,q}}[\phi]\cdot n_{\gamma_{\beta}}+\int_{\underline{H}_{v_0}\cap\{|u|\leq |u_{\gamma_{\beta}}|(v_0)\}} J^{N_{p,q}}[\phi]\cdot \underline{L}\right].
\end{split}
\end{equation}
\end{itemize}
\end{proposition}
\begin{proof}
We first apply the divergence theorem in region $\uline{\mathcal{A}}$, to obtain
\begin{equation}
\label{eq:stokespastgammaalpha}
\begin{split}
\int_{H_u\cap \uline{\mathcal{A}}}& J^{N_{2,q}}[\phi]\cdot L+\int_{\underline{H}_v\cap \uline{\mathcal{A}}} J^{N_{2,q}}[\phi]\cdot \underline{L}\\
=\:& \int_{\mathcal{H}^+\cap\{v\geq v_0\}} J^{N_{2,q}}[\phi]\cdot L+ \int_{\underline{H}_{v_0}\cap \{|u|\geq |u_{\underline{\gamma}_{\alpha}}|(v_0)\}} J^{N_{2,q}}[\phi]\cdot \underline{L}\\
&-\int_{ \uline{\mathcal{A}}} K^{N_{p,q}}_{\textnormal{null}}[\phi]+K^{N_{p,q}}_{\textnormal{angular}}[\phi]+K^{N_{p,q}}_{\textnormal{mixed}}[\phi].
\end{split}
\end{equation}
We can estimate $K^{N_{p,q}}_{\textnormal{null}}[\phi]$ and $K^{N_{p,q}}_{\textnormal{angular}}[\phi]$ in exactly the same way as in Proposition \ref{prop:mainenergyestimate}. We are left with estimating $K^{N_{p,q}}_{\textnormal{mixed}}[\phi]$. If we apply the estimates in Section~ \ref{sec:estimatesmetricconn} to (\ref{eq:errorterms3}), we can estimate,
\begin{equation*}
\Omega^2|K^{N_{p,q}}_{\textnormal{mixed}}[\phi]|\leq C(v+|u|)^{-2}(v^q|L\phi||\partial_{{\varphi_*}}\phi|+ |u|^p|\underline{L}\phi||\partial_{{\varphi_*}}\phi|).
\end{equation*}
By applying Cauchy--Schwarz, we can further estimate, for $\eta>0$,
\begin{equation*}
\begin{split}
(v+|u|)^{-2}v^q|L\phi||\partial_{\varphi_*}\phi|\leq &\: C|u|^{-1-\eta}v^q(L\phi)^2+C(v+|u|)^{-2}v^q|u|^{1-p+\eta}|u|^p\Omega^2|\snabla\phi|^2\\
\leq &\: C|u|^{-1-\eta}v^q(L\phi)^2+C|u|^{\frac{q}{\alpha}+\eta-p-1}|u|^p\Omega^2|\snabla\phi|^2,
\end{split}
\end{equation*}
where we used that $v\leq C|u|^{\frac{1}{\alpha}}$ in $J^-(\underline{\gamma}_{\alpha})$, for some constant $C>0$.

Similarly, we apply Cauchy--Schwarz to estimate
\begin{equation*}
\begin{split}
(v+|u|)^{-2}|u|^p|\underline{L}\phi||\partial_{\varphi_*}\phi|\leq &\: C v^{-1-\eta}|u|^p(\underline{L}\phi)^2 +C(v+|u|)^{-2}v^{1+\eta}|u|^p\Omega^2|\snabla\phi|^2\\
\leq &\: C v^{-1-\eta}|u|^p(\underline{L}\phi)^2+ C|u|^{\frac{1+\eta}{\alpha}-2}|u|^p\Omega^2|\snabla\phi|^2.
\end{split}
\end{equation*}

Combined with the estimates from Proposition \ref{prop:mainenergyestimate} for $K^{N_{p,q}}_{\textnormal{angular}}$ and $K^{N_{p,q}}_{\textnormal{null}}$, we can apply Lemma \ref{lm:gronwall} with
\begin{align*}
A=\:& \int_{\mathcal{H}^+\cap\{v\geq v_0\}} J^{N_{2,q}}[\phi]\cdot L+ \int_{\uline{H}_{v_0}} J^{N_{2,q}}[\phi]\cdot \underline{L},\\
f(u,v)=\:&\int_{H_{u}\cap \uline{\mathcal{A}}}v^q(L\phi)^2+|u|^2\Omega^2|\snabla\phi|^2,\\
g(u,v)=\:&\int_{\underline{H}_v\cap \uline{\mathcal{A}}}|u|^2(\underline{L}\phi)^2+v^q\Omega^2|\snabla \phi|^2,\\
h(u)=\:&|u|^{-1-\eta}+|u|^{q+\eta-5}+|u|^{-2+\epsilon}+|u|^{\frac{q}{\alpha}+\eta-3}+|u|^{\frac{1+\eta}{\alpha}-2},\\
k(v)=\:&v^{-1-\eta}+v^{\eta-q-1}+v^{-2+\epsilon},
\end{align*}
where we use that $h$ and $k$ are integrable for $0<\eta<\min\{q,1,2-\frac{q}{\alpha},\alpha-1\}$ and $0<\epsilon<1$. Hence, we need $\alpha>1$ to conclude the estimate in (i).

Consider now the region $\mathcal{A}$. The estimates above can be repeated here, where the roles of $|u|$ and $v$ are replaced when estimating $K_{\textnormal{mixed}}^{N_{p,q}}$:
\begin{equation*} 
\Omega^2|K^{N_{p,q}}_{\textnormal{mixed}}[\phi]|\leq Cv^{-1-\eta}|u|^p(\underline{L}\phi)^2+C|u|^{-1-\eta}v^q(L\phi)^2+C(v^{\frac{p}{\beta}+\eta-q-1}+v^{\frac{1+\eta}{\beta}-2})v^q\Omega^2|\snabla\phi|^2.
\end{equation*}

Furthermore, we can actually improve the estimate for $K_{\textnormal{angular}}^{N_{p,q}}$ from Proposition \ref{prop:mainenergyestimate} when restricted to $\mathcal{A}$, by including the cases $0<p\leq 2$, with $0\leq q<2$. 

Indeed, we can estimate
\begin{equation*}
\begin{split}
-K^{N_{p,q}}_{\textnormal{angular}}=\:& \frac{1}{2(v+|u|)}\left[\left((q-2)v+q|u|\right)v^{q-1}+\left((2-p)|u|-pv\right)|u|^{p-1}\right]|\snabla\phi|^2\\
&+\mathcal{O}((v^q+|u|^p)(v+|u|)^{-2}\log(v+|u|))|\snabla\phi|^2\\
\leq \:& C(v+|u|)^{-1}|u|^p|\snabla\phi|^2+\mathcal{O}((v^q+|u|^p)(v+|u|)^{-2}\log(v+|u|))|\snabla\phi|^2,
\end{split}
\end{equation*}
where in the second inequality we used that $(q-2)v+q|u|<0$ in $\mathcal{A}$ if $v+|u|$ is suitably large and $q<2$, which follows from the inequality $|u|<\frac{2-q}{q}v$, which holds in $\mathcal{A}$ if $v+|u|$ is suitably large and $\beta> 1$.

We can now apply Lemma \ref{lm:gronwall} with
\begin{align*}
A=\:& \int_{\mathcal{H}^+\cap\{v\geq v_0\}} J^{N_{p,q}}[\phi]\cdot L+ \int_{\uline{H}_{v_0}} J^{N_{p,q}}[\phi]\cdot \underline{L},\\
f(u,v)=\:&\int_{H_{u}\cap \mathcal{A}}v^q(L\phi)^2+|u|^p\Omega^2|\snabla\phi|^2,\\
g(u,v)=\:&\int_{\underline{H}_v\cap \mathcal{A}}|u|^p(\underline{L}\phi)^2+v^q\Omega^2|\snabla \phi|^2,\\
h(u)=\:&|u|^{-1-\eta}+|u|^{q+\eta-3-p}+|u|^{-2+\epsilon},\\
k(v)=\:&v^{-1-\eta}+v^{p+\eta-q-3}+v^{-2+\epsilon}+v^{\frac{p}{\beta}+\eta-1-q}+v^{\frac{1+\eta}{\beta}-2},
\end{align*}
where $h$ and $k$ are integrable for for $0<\eta<\min\{q+2-p,p+2-q,2-\frac{p}{\beta},\beta-1\}$, $0<\epsilon<1$. For consistency, we therefore require $\beta>1$, $p<q+2$ and $q<p+2$. In particular, if $p=2$, we need $q>0$ and if $q=2$, we need $p>0$.
\end{proof}

\subsection{Energy estimates in $\mathcal{B}$}

We are left with proving a suitable energy estimate in the region $\mathcal{B}$. In Kerr--Newman spacetimes with $0\leq |a|<a_c$, we can obtain an energy estimate away from $\mathcal{H}^+$ with respect to the vector field $Y_p$, defined by
\begin{equation*}
Y_p=|u|^pH,
\end{equation*}
if we restrict to a region $\{v\geq v_1\}$, where $v_1$ is taken suitably large, such that $M-r$ is sufficiently small, so as to ensure that $Y_p$ is a causal vector field everywhere in $\mathcal{B}\cap \{v\geq v_1\}$; see the discussion in Section~ \ref{sec:hawkvf}.
\begin{proposition}
\label{prop:energyestimatehawking}
Let $0\leq p\leq 2$ and let $v_1>v_0$ be suitably large. Then there exists a constant $C=C(a,M,u_0,v_0,v_1,p)>0$ such that
\begin{equation}
\label{eq:weightfluxboundv2}
\begin{split}
\int_{H_u\cap \mathcal{B}\cap \{v\geq v_1\}}& J^{Y_p}[\phi]\cdot L+\int_{\underline{H}_v\cap \mathcal{B}\cap \{v\geq v_1\}} J^{Y_p}[\phi]\cdot \underline{L}+\int_{\gamma_{\beta}\cap\{v\geq v_1\}}J^{Y_p}[\phi]\cdot n_{\gamma_{\beta}}\\
&+\int_{\mathcal{B}\cap \{v\geq v_1\}}K^{Y_p}[\phi]\\
\leq& \: C\left[ \int_{\underline{\gamma}_{\alpha}\cap\{v\geq v_1\}} J^{Y_p}[\phi]\cdot n_{\mathcal{H}^+}+ \int_{\underline{H}_{v_1}\cap\{|u|\leq u_{\underline{\gamma}_{\alpha}}(v_1)\}} J^{Y_p}[\phi]\cdot \underline{L}\right].
\end{split}
\end{equation}
\end{proposition}
\begin{proof}
We consider the region $\mathcal{B}\cap \{v\geq v_1\}$, where $v_1$ can be chosen suitably large, such that $H$ is causal for $|a|<a_c$ everywhere in $\mathcal{B}\cup \{v\geq v_1\}$. See Section~ \ref{sec:hawkvf}. 

We use that $H$ is a Killing vector field to easily obtain an expression for $K^{Y_p}$,
\begin{equation*}
\begin{split}
K^{Y_p}[\phi]=\:&g^{\alpha\beta}\nabla_{\beta}(J^{Y_p}_{\alpha})=g^{\alpha\beta}\nabla_{\beta}(|u|^p)J^H_{\alpha}[\phi]+|u|^pK^H[\phi]\\
=\:&\frac{p}{2}\Omega^{-2}|u|^{p-1}J^H[\phi]\cdot L \geq 0,
\end{split}
\end{equation*}
where non-negativity in the case that $\phi$ is not axisymmetric requires that $H$ is causal.

Moreover, there exists a constant $C>0$ such that we can estimate
\begin{equation*}
\begin{split}
J^H[\phi]\cdot L=\:&(L\phi)^2+\Omega^2|\snabla\phi|^2-b^{{\varphi_*}}L\phi\partial_{{\varphi_*}}\phi\\
\geq&\: C\left[(L\phi)^2+\Omega^2|\snabla\phi|^2\right].
\end{split}
\end{equation*}
If we apply Stokes' theorem in the region $\mathcal{B}\cap \{v\geq v_1\}$, the bulk term is therefore of a good sign.
\end{proof}

We can now obtain energy estimates in the entire region ${{\mathcal{D}_{u_0,v_0}}}$ by combining the results from Propositions \ref{prop:finitevolumeenergyestimates} and \ref{prop:energyestimatehawking}.
\begin{proposition}
\label{energyestsmallakn}
Let $0\leq p<2$ and $0\leq q<2$. Then there exist $\alpha=\alpha(p,q)>1$ and $\beta=\beta(p,q)>1$, such that for all $H_u$ and $\underline{H}_v$ in ${{\mathcal{D}_{u_0,v_0}}}$,
\begin{equation*}
\begin{split}
\int_{H_u}& J^{N_{p,q}}[\phi]\cdot L+\int_{\underline{H}_v} J^{N_{p,q}}[\phi]\cdot \underline{L}\\
\leq \:& C\left[ \int_{\mathcal{H}^+\cap\{v\geq v_0\}} J^{N_{2,q\beta \alpha }}[\phi]\cdot L+ \int_{\underline{H}_{v_0}} J^{N_{2,q \beta \alpha}}[\phi]\cdot \underline{L}\right]=:CE_{q\alpha \beta}[\phi],
\end{split}
\end{equation*}
with $0<q\beta \alpha \leq 2$ and $C=C(a,M,u_0,v_0,p,q,\alpha,\beta)>0$.
\end{proposition}
\begin{proof}
We first restrict to the region $\{v\geq v_1\}$. Note that
\begin{equation*}
\begin{split}
\int_{H_u\cap \mathcal{B}\cap \{v\geq v_1\}}& J^{N_{p,q}}[\phi]\cdot L+\int_{\underline{H}_v\cap \mathcal{B}\cap \{v\geq v_1\}} J^{N_{p,q}}[\phi]\cdot \underline{L}\\
\leq \:& C\left[\int_{H_u\cap \mathcal{B}\cap \{v\geq v_1\}} J^{Y_{p'}}[\phi]\cdot L+\int_{\underline{H}_v\cap \mathcal{B}\cap \{v\geq v_1\}} J^{Y_{p'}}[\phi]\cdot \underline{L}\right],
\end{split}
\end{equation*}
for $p' \geq p$ and $p'\geq q\beta$.

Furthermore,
\begin{equation*}
\begin{split}
\int_{H_u\cap \mathcal{B}\cap \{v\geq v_1\}}& J^{Y_p'}[\phi]\cdot L+\int_{\underline{H}_v\cap \mathcal{B}\cap \{v\geq v_1\}} J^{Y_p'}[\phi]\cdot \underline{L}\\
\leq \:& C\left[\int_{H_u\cap \mathcal{B}\cap \{v\geq v_1\}} J^{N_{2,q''}}[\phi]\cdot L+\int_{\underline{H}_v\cap\mathcal{B}\cap \{v\geq v_1\}} J^{N_{2,q''}}[\phi]\cdot \underline{L}\right],
\end{split}
\end{equation*}
for $q''\geq p'\alpha$.

Combining Proposition \ref{prop:finitevolumeenergyestimates} and \ref{prop:energyestimatehawking} we can therefore estimate
\begin{equation*}
\begin{split}
\int_{H_u\cap \mathcal{A}\cap \{v\geq v_1\}}& J^{N_{p,q}}[\phi]\cdot L+\int_{\underline{H}_v\cap \mathcal{A}\cap \{v\geq v_1\}} J^{N_{p,q}}[\phi]\cdot \underline{L}\\
\leq \:& C\int_{\gamma_{\beta}\cap \{v\geq v_1\}} J^{N_{p,q}}[\phi]\cdot n_{\gamma_{\beta}}+\int_{\underline{H}_{v_1}\cap\{|u|\leq |u_{\gamma_{\beta}}(v_0)|\}} J^{N_{p,q}}[\phi]\cdot \underline{L}\\
\leq \:& C\int_{\gamma_{\beta}\cap \{v\geq v_1\}} J^{Y_{p'}}[\phi]\cdot n_{\gamma_{\beta}}+C\int_{\underline{H}_{v_1}\cap\{|u|\leq |u_{\gamma_{\beta}}(v_1)|\}} J^{N_{p,q}}[\phi]\cdot \underline{L}\\
\leq \:& C\Bigg[ \int_{\underline{\gamma}_{\alpha}\cap\{v\geq v_1\}} J^{Y_{p'}}[\phi]\cdot n_{\underline{\gamma}_{\alpha}}+ \int_{\underline{H}_{v_1}\cap\{|u_{\gamma_{\beta}}(v_1)|\leq |u|\leq |u_{\underline{\gamma}_{\alpha}}(v_1)|\}} J^{Y_{p'}}[\phi]\cdot \underline{L}\\
&+\int_{\underline{H}_{v_1}\cap\{|u|\leq |u_{\gamma_{\beta}}(v_1)|\}} J^{N_{p,q}}[\phi]\cdot \underline{L}\Bigg],
\end{split}
\end{equation*}
where we need $\beta>1$, $0\leq p \leq 2$ and $0<q\leq 2$, or $0< p \leq 2$ and $0\leq q\leq 2$. Moreover, we need
\begin{align*}
p'\geq&\: p,\\
p'\geq&\: q\beta.
\end{align*}

Similarly, we can estimate
\begin{equation*}
\begin{split}
\int_{H_u\cap \mathcal{B}\cap \{v\geq v_1\}}& J^{N_{p,q}}[\phi]\cdot L+\int_{\underline{H}_v\cap \mathcal{B}\cap \{v\geq v_1\}} J^{N_{p,q}}[\phi]\cdot \underline{L}\\
\leq \:& C\Bigg[ \int_{\underline{\gamma}_{\alpha}\cap\{v\geq v_1\}} J^{Y_{p'}}[\phi]\cdot n_{\underline{\gamma}_{\alpha}}+ \int_{\underline{H}_{v_1}\cap\{|u_{\gamma_{\beta}}(v_1)|\leq |u|\leq |u_{\underline{\gamma}_{\alpha}}(v_1)|\}} J^{Y_{p'}}[\phi]\cdot \underline{L},
\end{split}
\end{equation*}
where $p'\geq p$ and $p'\geq q\beta$.

Now we apply Proposition \ref{prop:finitevolumeenergyestimates} in the region $\uline{\mathcal{A}}$ to estimate
\begin{equation*}
\begin{split}
\int_{\underline{\gamma}_{\alpha}\cap\{v\geq v_1\}}& J^{Y_{p'}}[\phi]\cdot n_{\underline{\gamma}_{\alpha}}\\
\leq \:& C\left[ \int_{\mathcal{H}^+\cap\{v\geq v_1\}} J^{N_{2,q''}}[\phi]\cdot L+ \int_{\underline{H}_{v_1}\cap\{|u|\geq |u_{\underline{\gamma}_{\alpha}}(v_1)|\}} J^{N_{2,q''}}[\phi]\cdot \underline{L}\right],
\end{split}
\end{equation*}
where we need $\alpha>1$, and we require
\begin{align*}
p'\leq \:& 2,\\
p'\alpha \leq \:& q''\leq 2.
\end{align*}
If we combine the restrictions on $p,q,p',q'$ and $q''$, we obtain
\begin{align*}
p\leq \:& p'\leq \frac{2}{\alpha},\\
q\leq \:& \frac{p'}{\beta}\leq \frac{2}{\alpha \beta}.
\end{align*}

We now consider the region $\{v_0 \leq v\leq v_1\}$. Since the region $\mathcal{B}\cap \{v\leq v_1\}$ is compact, we do not need to appeal to the estimates with respect to the vector fields $Y_p$ from Proposition \ref{prop:energyestimatehawking}. Instead, we use the vector fields $N_{p,q}$, as in Proposition \ref{prop:finitevolumeenergyestimates}, making use of the compactness of $\mathcal{B}\cap \{v\leq v_1\}$ to in particular estimate the previously problematic $K^{N_{p,q}}_{\textnormal{mixed}}[\phi]$ error term. 

We arrive at the estimate:
\begin{equation*}
\begin{split}
\int_{H_u\cap \{v\leq v_1\}}& J^{N_{p,q}}[\phi]\cdot L+\int_{\underline{H}_v\cap \{v\leq v_1\}} J^{N_{p,q}}[\phi]\cdot \underline{L}\\
\leq \:& C\int_{\mathcal{H}^+\cap \{v\leq v_1\}} J^{N_{p,q}}[\phi]\cdot L+C\int_{\underline{H}_{v_0}} J^{N_{p,q}}[\phi]\cdot \underline{L},
\end{split}
\end{equation*}
for any $0\leq p,q\leq 2$.

The estimate in the proposition now follows by adding the estimates in $\{v\geq v_1\}$ and $\{v\leq v_1\}$ together.
\end{proof}

\begin{remark}
For $\epsilon>0$ arbitrarily small we can always choose $\alpha$ and $\beta$ in Proposition \ref{energyestsmallakn} suitably close to 1, so that we can take $p=2-\epsilon$ and $q=2-\epsilon$.
\end{remark}

We have now proved Theorem \ref {thm:energyestsmallakn}.

\section{Higher-order energy estimates}
In order to obtain $L^{\infty}$ bounds from the $L^2$ bounds derived in Section~ \ref{sec:energyestimatesaxisymm} and \ref{sec:energyestimatesslowlyrot}, we need to derive similar $L^2$ bounds for higher-order derivatives of $\phi$. \textbf{In this section we will use $\phi$ to denote a solution to} (\ref{eq:waveqkerr})\textbf{ corresponding to initial data from Proposition} \ref{prop:wellposedness}. We will always specify whether we are assuming $\phi$ arises from axisymmetric data in Proposition \ref{prop:wellposedness}, or the rotation parameter $a$ is restricted to the range $0\leq |a|<a_c$.

\label{sec:higherorderenergyestimates}
\subsection{Elliptic estimates on $S^2_{u,v}$}
We will first show that the angular derivatives on the Eddington--Finkelstein-type spheres $S^2_{u,v}$ can be controlled by derivatives with respect to the Killing vector field $\Phi$, the null directed vector fields $L$ and $\uline{L}$ and the Carter operator $Q$; see Section~ \ref{sec:hawkvf}. The lemma below can be found in \cite{are6}.
\begin{lemma}
Given a function $f: \mathcal{M}\cap {\mathcal{D}_{u_0,v_0}}\to \R$, there exists a $C=C(a)>0$ such that
\begin{equation}
\label{est:ellipticBLspheres}
\int_{\s^2} |\nabla_{\s^2}f|^2(t,r,\theta,\varphi)+|\nabla^2_{\s^2}f|^2(t,r,\theta,\varphi)\, d\mu_{\s^2}\leq C \int_{\s^2}(Qf)^2+(\Phi^2 f)^2+(T^2 f)^2\,d\mu_{\s^2},
\end{equation}
where $\nabla_{\s^2}$ denotes the covariant derivative on $\s^2$ and $d\mu_{\s^2}=\sin \theta d\theta d\varphi$.
\end{lemma}
\begin{proof}
By decomposing $f$ into spherical harmonics $f_l$ on $\s^2$, one can show that
\begin{equation*}
\begin{split}
\int_{\s^2}(\Delta_{\s^2}f)^2\, d\mu_{\s^2}=\:&\sum_{l=0}^{\infty}\int_{\s^2}(l(l+1))^2f_l^2\, d\mu_{\s^2}=\sum_{l=0}^{\infty}\int_{\s^2} \sum_{1\leq |k|\leq 2}(O^kf_l)^2\, d\mu_{\s^2}\\
=\:& \sum_{1\leq |k|\leq 2}\int_{\s^2} (O^kf)^2\, d\mu_{\s^2}\geq \int_{\s^2}|\nabla_{\s^2}f|^2+|\nabla^2_{\s^2}f|^2\,d\mu_{\s^2},
\end{split}
\end{equation*}
where $\Delta_{\s^2}$ denotes the Laplacian on $\s^2$.

The estimate (\ref{est:ellipticBLspheres}) follows by using the definition of $Q$ to rewrite the left-hand side above and applying Cauchy--Schwarz.
\end{proof}
We would similarly like to control the angular derivatives in the coordinates $(\theta_*,{\varphi_*})$ by using the operators $Q,T$ and $\Phi$ that commute with $\square_g$. However, since the tangent spaces to the Boyer--Lindquist spheres and the spheres $S^2_{u,v}$ are not spanned by the same tangent vectors, we need to include $L$ and $\underline{L}$ derivatives in our estimate.

For the sake of convenience, we change from the chart $(\theta_*,\varphi_*)$ on the 2-spheres $S^2_{u,v}$ to the chart $(\theta_*,\varphi)$, because the induced metric on $S^2_{u,v}$ then becomes diagonal:
\begin{equation*}
\slashed{g}=f_1^2f_2^2(\partial_{\theta_*}F)^2 R^{-2}d\theta_*^2+R^2 \sin^2 \theta d\varphi^2.
\end{equation*}

\begin{proposition}
\label{prop:ellipticsphere1}
Given a suitably regular function $f: \mathcal{M}\cap {{\mathcal{D}_{u_0,v_0}}}\to \R$, there exists a $C=C(M,a,u_0,v_0)>0$ such that
\begin{equation}
\begin{split}
\label{est:ellipticspherev1}
\int_{S^2_{u,v}}|\snabla f|^2+|\snabla^2f|^2\,\sqrt{\det \slashed{g}}d\theta_*d\varphi_*\leq \:& C\sum_{k=0}^1\sum_{\Gamma\in\{\textnormal{id},{\Phi},{\Phi}^2,T^2,Q\}}\int_{S^2_{u,v}}(\Gamma \underline{L}^kf)^2+(\Gamma L^k f)^2+(\underline{L}^2f)^2\\
&+({L}^2f)^2+(L\underline{L} f)^2+(b^{{\varphi_*}})^2[(Q\Phi f)^2+(T^2\Phi f)^2\\
&+(\Phi^3f)^2]\,\sqrt{\det \slashed{g}}d\theta_*d\varphi_*,
\end{split}
\end{equation}
for all $|u|\geq |u_0|$ and $v\geq v_0$.

Let $v_1>v_0$ and $|u_1|>u_0$ be suitably large. Then we can estimate
\begin{equation}
\begin{split}
\label{est:ellipticspherev2}
\int_{S^2_{u,v}}&|\snabla^2f|^2\,\sqrt{\det \slashed{g}}d\theta_*d\varphi_*\\
\leq \:& C \int_{S^2_{u,v}} (Qf)^2+(\Phi^2 f)^2 +(T^2f)^2 + (L^2f)^2+(\uline{L}^2f)^2+(L\uline{L}f)^2+(Lf)^2+(\uline{L}f)^2\\
&+|\snabla Lf|^2+|\snabla \uline{L}f|^2+|\snabla f|^2\,\sqrt{\det \slashed{g}}d\theta_*d\varphi_*,
\end{split}
\end{equation}
if either $|u|\geq |u_1|$ or $v\geq v_1$, where $C=C(M,a,u_0,v_0,u_1,v_1)>0$.
\end{proposition}
\begin{proof}
By Theorem \ref{thm:estmetric} there exist uniform constants $C, c>0$ such that
\begin{equation*}
c\sin\theta \leq \det \slashed{g} \leq C \sin \theta,
\end{equation*}
where $\det \slashed{g}$ is the determinant with respect to the coordinate basis corresponding to the chart $(\theta_*,\varphi_*)$, which is equal to the determinant of the matrix of $\slashed{g}$ with respect to the coordinate basis corresponding to the chart $(\theta_*,\varphi)$.

Consider the first-order angular derivatives. We can write
\begin{equation*}
|\snabla f|^2=\slashed{g}^{\theta_*\theta_*}(\partial_{\theta_*}f)^2+\slashed{g}^{\varphi\varphi}(\partial_{\varphi}f )^2\leq C\left[(\partial_{\theta_*}f)^2+\sin^{-2}\theta(\partial_{\varphi}f)^2\right].
\end{equation*}
By the chain rule, we have that
\begin{equation*}
\begin{split}
\partial_{\theta}f=\:&(\partial_{\theta}\theta_*)\partial_{\theta_*}f+\frac{1}{2}(\partial_{\theta}r_*)(\partial_v f-\partial_u f)\\
=\:&(\partial_{\theta}\theta_*)\partial_{\theta_*}f+\frac{1}{2}(\partial_{\theta}r_*)(Lf-b^{{\varphi_*}}\partial_{{\varphi_*}}f-\underline{L}f).
\end{split}
\end{equation*}
By applying (\ref{eq:estJac1}) and (\ref{eq:estJac3}), we find that
\begin{equation}
\label{est:mainest1storderv0}
(\partial_{\theta_*}f)^2\leq C\left[(\partial_{\theta}f)^2+\sin^2 \theta (Lf)^2+\sin^2\theta(\underline{L}f)^2+\sin^2\theta(b^{{\varphi_*}})^2({\Phi} f)^2\right].
\end{equation}
We can now conclude the following:
\begin{equation}
\label{est:mainest1storder}
|\snabla f|^2\leq C\left[|\nabla_{\s^2}f|^2+\sin^2\theta(Lf)^2+\sin^2\theta(\underline{L}f)^2\right].
\end{equation}

Now consider the second-order angular derivatives. We can estimate
\begin{equation*}
\begin{split}
|\snabla^2f|^2=\:& \slashed{g}^{AB}\slashed{g}^{CD}(\nabla_A\partial_Cf)(\nabla_B\partial_Df)\\
=\:&(\slashed{g}^{\varphi\varphi})^2(\partial_{\varphi}^2f)^2+(\slashed{g}^{\theta_*\theta_*})^2(\partial_{\theta_*}^2f)^2+g^{\theta_*\theta_*}g^{\varphi\varphi}(\partial_{\theta_*}\partial_{\varphi}f)^2\\
&+\slashed{g}^{AB}\slashed{g}^{CD}(\slashed{\nabla}_A\partial_C)^E(\slashed{\nabla}_B \partial_D)^F \snabla_Ef \snabla_Ff,\\
\leq \:&C\left[ \sin^{-4}\theta (\partial_{\varphi}^2f)^2+(\partial_{\theta_*}^2f)^2+\sin^{-2}\theta(\partial_{\theta_*}\partial_{\varphi}f)^2\right]\\
&+\slashed{g}^{AB}\slashed{g}^{CD}(\slashed{\nabla}_A\partial_C)^E(\slashed{\nabla}_B \partial_D)^F \snabla_Ef \snabla_Ff.
\end{split}
\end{equation*}

By applying the chain rule we find that
\begin{equation*}
\begin{split}
\partial_{\theta}^2f=\:&\partial_{\theta}^2\theta_* \partial_{\theta_*}f+(\partial_{\theta}\theta_*)\left[\partial_{\theta}\theta_* \partial_{\theta_*}^2f+\partial_{\theta}r_*\left(\partial_{\theta_*}Lf-\partial_{\theta_*}(b^{{\varphi_*}}\Phi f)-\partial_{\theta_*}\underline{L}f\right)\right]\\
&+\frac{1}{4}(\partial_{\theta}r_*)^2((L-b^{{\varphi_*}}\Phi)^2f-2(L-b^{{\varphi_*}}\Phi)\underline{L}f+\underline{L}^2f)\\
&+\frac{1}{2}\partial_{\theta}^2r_*(Lf-b^{{\varphi_*}}\partial_{{\varphi_*}}f-\underline{L}f).
\end{split}
\end{equation*}
Consequently, by applying the estimates from Theorem \ref{thm:estmetric}, we obtain
\begin{equation}
\label{est:mainest2ndorder}
\begin{split}
(\partial_{\theta_*}^2f)^2\leq \:& (\partial_{\theta_*}f)^2+(\partial_{\theta}^2 f)^2+\sin^2\theta(\partial_{\theta_*}(\underline{L}f))^2+\sin^2\theta(\partial_{\theta_*}(b^{{\varphi_*}}\Phi f))^2+\sin^2\theta(\partial_{\theta_*}(Lf))^2\\
&+(\underline{L}f)^2+(Lf)^2+\sin^4\theta(L^2f)^2+\sin^4\theta(L\underline{L}f)^2+\sin^4\theta(\underline{L}^2f)^2+\sin^4\theta(b^{{\varphi_*}})^2(L\Phi f)^2\\
&+\sin^4\theta(b^{{\varphi_*}})^2(\underline{L}\Phi f)^2+\sin^4\theta(b^{{\varphi_*}})^2(\underline{L}\Phi f)^2+(b^{{\varphi_*}})^4(\Phi^2f)^2+(\Phi f)^2,
\end{split}
\end{equation}
where we used that we can freely commute with $\Phi$.

We can further estimate
\begin{equation}
\label{est:mainest2ndordera}
\begin{split}
(\partial_{\theta_*}(b^{{\varphi_*}}\Phi f))^2\leq \:& (b^{{\varphi_*}})^2(\Phi\partial_{\theta_*}f)^2+(\partial_{\theta_*}b^{{\varphi_*}})^2(\Phi f)^2\\
\leq \:& C\left[(\partial_{\theta}\Phi f)^2+(L\Phi f)^2+(\underline{L}\Phi f)^2+(\Phi^2f)^2\right]\\
\leq \:& C\left[(Q\Phi f)^2+(T^2\Phi f)^2+(L\Phi f)^2+(\underline{L}\Phi f)^2+(\Phi f)^2+(\Phi^2f)^2+(\Phi^3f)^2\right].
\end{split}
\end{equation}

We now turn to $(\slashed{\nabla}_A\partial_C)^E$. The only non-vanishing components are given by:
\begin{align*}
\left|(\slashed{\nabla}_{\theta_*}\partial_{\theta_*})^{\theta_*}\right|=\:&\left|\frac{1}{2}g^{{\theta_*}{\theta_*}}\partial_{\theta_*}g_{{\theta_*}{\theta_*}}\right| \leq C,\\
\left|(\slashed{\nabla}_{\varphi}\partial_{\varphi})^{\theta_*}\right|=\:&\left|\frac{1}{2}g^{{\theta_*}{\theta_*}}\partial_{\theta_*}g_{{\varphi}{\varphi}}\right|\leq C\sin \theta,\\
\left|(\slashed{\nabla}_{\theta_*}\partial_{\varphi})^{\varphi}\right|=\:&\left|\frac{1}{2}g^{{\varphi}{\varphi}}\partial_{\theta_*}g_{{\varphi}{\varphi}}\right|\leq C\sin \theta,
\end{align*}
where we used the estimates from Theorem \ref{thm:estmetric} to arrive at the inequalities on the right-hand sides.

We can now estimate, by applying (\ref{est:mainest1storderv0}),
\begin{equation}
\label{est:mainest2ndorderb}
\begin{split}
\slashed{g}^{AB}\slashed{g}^{CD}(\slashed{\nabla}_A\partial_C)^E(\slashed{\nabla}_B \partial_D)^F \snabla_Ef \snabla_Ff=\:&\left(\slashed{g}^{\theta_*\theta_*}\right)^2 ((\slashed{\nabla}_{\theta_*}\partial_{\theta_*})^{\theta_*})^2(\partial_{\theta_*}f)^2+2\slashed{g}^{\theta_*\theta_*}\slashed{g}^{\varphi \varphi}((\slashed{\nabla}_{\theta_*}\partial_{\varphi})^{\varphi})^2(\partial_{\varphi}f)^2\\
&+(\slashed{g}^{\varphi\varphi})^2((\slashed{\nabla}_{\varphi}\partial_{\varphi})^{\theta_*})^2(\partial_{\theta_*}f)^2\\
\leq \:& C\sin^{-2}\theta(\partial_{\theta_*}f)^2+C\sin^{-4}\theta (\partial_{\varphi}f)^2\\
\leq \:& C\slashed{g}_{\s^2}^{AB}\slashed{g}_{\s^2}^{CD}(({\nabla_{\s^2}})_A\partial_C)^E(({\nabla_{\s^2}})_B \partial_D)^F \partial_Ef \partial_Ff\\
&+C(Lf)^2+C(\uline{L}f)^2+C(b^{\varphi_*})^2(\Phi f)^2.
\end{split}
\end{equation}

By combining (\ref{est:mainest2ndorder}), the first inequality in (\ref{est:mainest2ndordera}), and (\ref{est:mainest2ndorderb}) and making use of the smallness of $|b^{\varphi_*}|$ for suitably large $v_1$ and $|u_1|$, we obtain (\ref{est:ellipticspherev2}). We can also estimate
\begin{equation*}
\begin{split}
\int_{S^2_{u,v}}|\snabla f|^2+|\snabla^2f|^2\,\sqrt{\det \slashed{g}}d\theta_*d\varphi\leq \:& C\sum_{k=0}^1\sum_{\Gamma\in\{\textnormal{id},{\Phi},{\Phi}^2,T^2\}}\int_{S^2_{u,v}}(\underline{L}^k\Gamma f)^2+(L^k\Gamma f)^2+(\underline{L}^2f)^2\\
&+({L}^2f)^2+(L\underline{L} f)^2+(\Phi^2f)^2+(b^{{\varphi_*}})^2[(Q\Phi f)^2+(T^2\Phi f)^2\\
&+(\Phi^3f)^2]+(Qf)^2+\sin^2\theta(Q\underline{L}f)^2\\
&+\sin^2\theta(QLf)^2\,\sqrt{\det \slashed{g}}d\theta_*d\varphi_*.
\end{split}
\end{equation*}
We have now obtained the estimate (\ref{est:ellipticspherev1}).

We can easily commute $L$ and $\underline{L}$ with $\Gamma$ above.
\end{proof}

\subsection{Commutator estimates}
\label{sec:commutatorestimates}
We can use the elliptic estimates in Proposition \ref{prop:ellipticsphere1} to control angular error terms that arise from commuting $\square_g$ with $L$ and $\underline{L}$. We first derive a general expression for the commutator $[\square_g,W^mV^n]$, where $V$ and $W$ are vector fields.
\begin{lemma}
\label{lm:commwithV}
Let $V$ and $W$ be vector fields and $n\geq 1$, then
\begin{equation*}
\begin{split}
2\Omega^2\square_g(W^mV^n\phi)=\:&[W^mV^n,\underline{L}]L\phi+\underline{L}([W^mV^n,L]\phi)+[W^mV^n,L]\underline{L}\phi+L([W^mV^n,\underline{L}]\phi)\\
&\sum_{l=1}^m\sum_{k=1}^{n} \binom{m}{l}\binom{n}{k} W^lV^k(\Omega\tr \underline{\chi})W^{m-l}V^{n-k}L\phi+\Omega \tr {\underline{\chi}}[W^mV^n,L]\phi\\
&+ \sum_{l=1}^m\sum_{k=1}^{n} \binom{m}{l}\binom{n}{k} W^lV^k(\Omega\tr {\chi})W^{m-l}V^{n-k}\underline{L}\phi+\Omega \tr {\chi}[W^mV^{n},\underline{L}]\phi\\
&-2\sum_{l=1}^m\sum_{k=1}^n  \binom{m}{l}\binom{n}{k} W^lV^k(\slashed{g}^{AB}\partial_A\Omega^2)W^{m-l}V^{n-k}\partial_B\phi\\
&-2\slashed{g}^{AB}\partial_A\Omega^2[W^mV^n,\partial_B]\phi\\
&-2\sum_{l=1}^m\sum_{k=1}^{n}  \binom{m}{l}\binom{n}{k} W^lV^k(\Omega^2)W^{m-l}V^{n-k}\slashed{\Delta}\phi-2\Omega^2[W^mV^n,\slashed{\Delta}]\phi.
\end{split}
\end{equation*}
\end{lemma}
\begin{proof}
Use the expression for the wave equation (\ref{eq:finalexpressionwaveeq}) in Appendix \ref{app:treq}, together with 
\begin{equation*}
W^mV^n(2\Omega^2\square_g\phi)=0.
\end{equation*}
\end{proof}

\begin{proposition}
\label{commenergyestimateskn1}
Either restrict to $|a|<a_c$ and let $0\leq p<2$ and $0\leq q<2$, or restrict to axisymmetric $\phi$, with $p=2$ and $0<q\leq 2$. Then there exist $\alpha=\alpha(p,q)>1$ and $\beta=\beta(p,q)>1$, such that for $u_1$ suitably large and $|u|\geq |u_1|$,
\begin{equation*}
\begin{split}
\int_{\underline{H}_{v}\cap\{|u|\geq |u_1|\}}&J^{N_{p,q}}[\underline{L}\phi]\cdot \underline{L}+J^{N_{p,q}}[L\phi]\cdot \underline{L}+\int_{{H}_{u}}J^{N_{p,q}}[\underline{L}\phi]\cdot L+J^{N_{p,q}}[L\phi]\cdot L\\
\leq \:& C \Bigg[\int_{\underline{H}_{v_0}\cap\{|u|\geq |u_1|}J^{N_{2,q\beta \alpha}}[\underline{L}\phi]\cdot \underline{L}+J^{N_{2,q\beta\alpha}}[L\phi]\cdot \underline{L}\\
&+\int_{\mathcal{H}^+\cap\{v\geq v_0\}}J^{N_{2,q\beta \alpha}}[\underline{L}\phi]\cdot L+J^{N_{2,q\beta \alpha}}[L\phi]\cdot L\Bigg]\\
&+C\sum_{\Gamma\in\{\textnormal{id},\Phi^2,Q\}}\left[\int_{\underline{H}_{v_0}\cap\{|u|\geq |u_1|\}}J^{N_{2,q\beta \alpha}}[\Gamma\phi]\cdot \underline{L}+\int_{\mathcal{H}^+\cap\{v\geq v_0\}}J^{N_{2,q\beta \alpha}}[\Gamma\phi]\cdot L\right]\\
&+C\sum_{\Gamma\in\{\Phi^2,Q\}}\int_{\mathcal{H}^+\cap\{v\geq v_0\}}v^{-2+\epsilon}(\Gamma \phi)^2,
\end{split}
\end{equation*}
where $0<q\beta \alpha \leq 2$, $\epsilon>0$ can be taken arbitrarily small and $$C=C(a,M,u_0,v_0,u_1,v_1,p,q,\alpha,\beta,\epsilon)>0.$$ For axisymmetric $\phi$, we can replace $N_{2,q\beta \alpha}$ on the right-hand side by $N_{2,q}$.

Moreover, for $v_1>v_0$ suitably large and $v\geq v_1$,
\begin{equation*}
\begin{split}
\int_{\underline{H}_{v}}&J^{N_{p,q}}[\underline{L}\phi]\cdot \underline{L}+J^{N_{p,q}}[L\phi]\cdot \underline{L}+\int_{{H}_{u}\cap\{v\geq v_1\}}J^{N_{p,q}}[\underline{L}\phi]\cdot L+J^{N_{p,q}}[L\phi]\cdot L\\
\leq \:& C \left[\int_{\underline{H}_{v_1}}J^{N_{2,q\beta \alpha}}[\underline{L}\phi]\cdot \underline{L}+J^{N_{2,q\beta\alpha}}[L\phi]\cdot \underline{L}+\int_{\mathcal{H}^+\cap\{v\geq v_1\}}J^{N_{2,q\beta \alpha}}[\underline{L}\phi]\cdot L+J^{N_{2,q\beta \alpha}}[L\phi]\cdot L\right]\\
&+C\sum_{\Gamma\in\{\textnormal{id},\Phi^2,Q\}}\left[\int_{\underline{H}_{v_1}}J^{N_{2,q\beta \alpha}}[\Gamma\phi]\cdot \underline{L}+\int_{\mathcal{H}^+\cap\{v\geq v_1\}}J^{N_{2,q\beta \alpha}}[\Gamma\phi]\cdot L\right]\\
&+C\sum_{\Gamma\in\{\Phi^2,Q\}}\int_{\mathcal{H}^+\cap\{v\geq v_1\}}v^{-2+\epsilon}(\Gamma \phi)^2,
\end{split}
\end{equation*}
where $0<q\beta \alpha \leq 2$, $\epsilon>0$ can be taken arbitrarily small and $$C=C(a,M,u_0,v_0,u_1,v_1,p,q,\alpha,\beta,\epsilon)>0.$$ For axisymmetric $\phi$, we can replace $N_{2,q\beta \alpha}$ on the right-hand side by $N_{2,q}$.
\end{proposition}
\begin{proof}
For any vector field $V$, we have that
\begin{equation*}
\mathcal{E}^{N_{p,q}}[V\phi]=N_{p,q}(\phi)\square_g(V\phi)=(|u|^p\underline{L}V\phi+v^qLV\phi)\square_g(V\phi).
\end{equation*}
We can commute $\square_g$ with $L$ and $\underline{L}$ and apply Lemma \ref{lm:commwithV} to obtain
\begin{equation*}
\begin{split}
2\Omega^2\square_g(\underline{L}\phi)=\:&\underline{L}([\underline{L},L]\phi)+[\underline{L},L]\underline{L}\phi+\underline{L}(\Omega\tr \underline{\chi})L\phi+\Omega \tr {\underline{\chi}}[\underline{L},L]\phi\\
&+ \underline{L}(\Omega\tr {\chi})\underline{L}\phi-2\underline{L}(\slashed{g}^{AB}\partial_A\Omega^2)\partial_B\phi-2\underline{L}(\Omega^2)\slashed{\Delta}\phi-2\Omega^2[\underline{L},\slashed{\Delta}]\phi.
\end{split}
\end{equation*}
and
\begin{equation*}
\begin{split}
2\Omega^2\square_g(L\phi)=\:&L([L,\underline{L}]\phi)+[L,\underline{L}]L\phi+L(\Omega\tr {\chi})\underline{L}\phi+\Omega \tr {{\chi}}[\underline{L},L]\phi\\
&+ L(\Omega\tr \underline{\chi})L\phi-2L(\slashed{g}^{AB}\partial_A\Omega^2)\partial_B\phi-2L(\Omega^2)\slashed{\Delta}\phi-2\Omega^2[L,\slashed{\Delta}]\phi\\
&-2\slashed{g}^{AB}\partial_A\Omega^2[L,\partial_B]\phi,
\end{split}
\end{equation*}

Moreover, we have that
\begin{align*}
[\slashed{\Delta},\underline{L}]\phi=\:&-\underline{L}\left(\frac{1}{\sqrt{\det \slashed{g}}}\right)\partial_A\left(\slashed{g}^{AB}\sqrt{\det \slashed{g}}\partial_B\phi\right)-\frac{1}{\sqrt{\det \slashed{g}}}\partial_A\left(\underline{L}\left(\slashed{g}^{AB}\sqrt{\det \slashed{g}}\right)\partial_B\phi\right),\\
[\slashed{\Delta},L]\phi=\:&-L\left(\frac{1}{\sqrt{\det \slashed{g}}}\right)\partial_A\left(\slashed{g}^{AB}\sqrt{\det \slashed{g}}\partial_B\phi\right)-\frac{1}{\sqrt{\det \slashed{g}}}\partial_A\left(L\left(\slashed{g}^{AB}\sqrt{\det \slashed{g}}\right)\partial_B\phi\right)\\
&-\frac{1}{\sqrt{\det \slashed{g}}}[L,\partial_A]\left(\slashed{g}^{AB}\sqrt{\det \slashed{g}}\partial_B\phi\right)-\slashed{g}^{AB}\partial_A([L,\partial_B]\phi).
\end{align*}

By applying the estimates from Section~ \ref{sec:estimatesmetricconn}, we obtain
\begin{equation*}
\begin{split}
|2\Omega^2\square_g(\underline{L}\phi)|&\lesssim (v+|u|)^{-2}|\underline{L}(\partial_{{\varphi_*}}\phi)|+(v+|u|)^{-3}(|L\phi|+|\underline{L}\phi|)+(v+|u|)^{-3}|\snabla\phi|\\
& +(v+|u|)^{-3}|\snabla^2\phi|.
\end{split}
\end{equation*}
and
\begin{equation*}
\begin{split}
|2\Omega^2\square_g(L\phi)|&\lesssim (v+|u|)^{-2}|L(\partial_{{\varphi_*}}\phi)|+(v+|u|)^{-3}(|L\phi|+|\underline{L}\phi|)+(v+|u|)^{-3}|\snabla\phi|\\
& +(v+|u|)^{-3}|\snabla^2\phi|.
\end{split}
\end{equation*}
Consequently, we can apply Cauchy--Schwarz to obtain
\begin{equation*}
\begin{split}
2\Omega^2|\mathcal{E}^{N_{p,q}}[\underline{L}\phi]|&\lesssim v^{-1-\eta}|u|^p(\underline{L}\underline{L}\phi)^2+|u|^{-1-\eta}v^q(L\underline{L}\phi)^2\\
&+(v^{1+\eta}|u|^p+|u|^{1+\eta}v^q)\Big[ (v+|u|)^{-4}|\underline{L}(\partial_{{\varphi_*}}\phi)|^2+(v+|u|)^{-6}(|L\phi|^2+|\underline{L}\phi|^2)\\
&+(v+|u|)^{-6}|\snabla\phi|^2 +(v+|u|)^{-6}|\snabla^2\phi|^2\Big].
\end{split}
\end{equation*}
and
\begin{equation*}
\begin{split}
2\Omega^2|\mathcal{E}^{N_{p,q}}[L\phi]|&\lesssim v^{-1-\eta}|u|^p(\underline{L}L\phi)^2+|u|^{-1-\eta}v^q(LL\phi)^2\\
&+(v^{1+\eta}|u|^p+|u|^{1+\eta}v^q)\Big[ (v+|u|)^{-4}|L(\partial_{{\varphi_*}}\phi)|^2+(v+|u|)^{-6}(|L\phi|^2+|\underline{L}\phi|^2)\\
&+(v+|u|)^{-6}|\snabla\phi|^2 +(v+|u|)^{-6}|\snabla^2\phi|^2\Big].
\end{split}
\end{equation*}
We obtain similar estimates for $2\Omega^2|\mathcal{E}^{Y_p}[\underline{L}\phi]|$ and $2\Omega^2|\mathcal{E}^{Y_p}[L\phi]|$, where we replace the weight $v^q$ by $|u|^p$.

Using (\ref{est:ellipticspherev2}) of Proposition \ref{prop:ellipticsphere1}, we can further estimate
\begin{equation}
\label{est:2ndorderangder}
\begin{split}
(v+|u|)^{-6}\int_{S^2_{u,v}}|\snabla^2\phi|^2\,d\mu_{\slashed{g}}\leq &\:C (v+|u|)^{-6}\int_{S^2_{u,v}} (Q\phi)^2+(\Phi^2 \phi)^2 +(T^2\phi)^2 + (L^2\phi)^2+(\uline{L}^2\phi)^2\\
&+(L\uline{L}\phi)^2+(L\phi)^2+(\uline{L}\phi)^2\\
&+|\snabla L\phi|^2+|\snabla \uline{L}\phi|^2+|\snabla \phi|^2\,d\mu_{\slashed{g}}.
\end{split}
\end{equation}
The $T^2\phi$ term can be absorbed into the remaining terms on the right-hand side of the above inequality. The energy estimates of Proposition \ref{prop:mainenergyestimate} and Proposition \ref{energyestsmallakn} apply also to $Q \phi$ and $\Phi^2 \phi$, so we can estimate, by applying the fundamental theorem of calculus together with Cauchy--Schwarz in the region $\{|u|\geq |u_1|\}$:
\begin{equation*}
\begin{split}
\int_{S^2_{u,v}}(\Gamma\phi)^2\,d\mu_{\slashed{g}}\leq \:& \int_{S^2_{-\infty,v}}(\Gamma\phi)^2\,d\mu_{\slashed{g}}+C|u|^{1-p'}\sup_{v_0\leq v'\leq v}\int_{\underline{H}_v'\cap\{|u|\geq |u_1|\}}|u|^{p'}(\underline{L}\Gamma\phi)^2\\
\leq \:& C|u|^{1-p'}\left[\int_{\underline{H}_{v_0}\cap\{|u|\geq |u_1|\}}J^{N_{2,\eta}}[\Gamma\phi]\cdot \underline{L}+\int_{\mathcal{H}^+\cap\{v\geq v_0\}}J^{N_{2,\eta}}[\Gamma\phi]\cdot L\right]\\
&+\int_{S^2_{-\infty,v}}(\Gamma\phi)^2\,d\mu_{\slashed{g}},
\end{split}
\end{equation*}
where $\Gamma \in \{\Phi^2,Q\}$ and we can take $p'=2$ if $\phi$ is axisymmetric and $p'=2-\eta$, with $\eta>0$ arbitrarily small, if $\phi$ is not axisymmetric.

We similarly apply the fundamental theorem of calculus together with Cauchy--Schwarz to obtain:
\begin{equation*}
\begin{split}
\int_{S^2_{u,v}}(\Gamma\phi)^2\,d\mu_{\slashed{g}}\leq \:& \int_{S^2_{-\infty,v}}(\Gamma\phi)^2\,d\mu_{\slashed{g}}+C|u|^{1-p'}\sup_{v_1\leq v'\leq v}\int_{\underline{H}_v'}|u|^{p'}(\underline{L}\Gamma\phi)^2\\
\leq \:& \int_{S^2_{-\infty,v}}(\Gamma\phi)^2\,d\mu_{\slashed{g}}\\
&+C|u|^{1-p'}\left[\int_{\underline{H}_{v_1}}J^{N_{2,\eta}}[\Gamma\phi]\cdot \underline{L}+\int_{\mathcal{H}^+\cap\{v\geq v_1\}}J^{N_{2,\eta}}[\Gamma\phi]\cdot L\right],
\end{split}
\end{equation*}

Therefore, we can estimate in $\{|u|\geq |u_1|\}$,
\begin{equation*}
\begin{split}
\int_{v_0}^{\infty}\int_{-\infty}^{u_1}&(v^q|u|^{1+\eta}+|u|^pv^{1+\eta})(v+|u|)^{-6}\int_{S^2_{u,v}}(\Gamma\phi)^2\,d\mu_{\slashed{g}}dudv\\
\leq \:& C\int_{v_0}^{\infty}v^{-4+\eta+\max\{p,q\}}\int_{S^2_{-\infty,v}}(\Gamma\phi)^2\,d\mu_{\slashed{g}}dv\\
&+C\int_{v_0}^{\infty}\int_{-\infty}^{u_1}(v^q|u|^{\eta}+|u|^{p-1+\eta}v^{1+\eta})(v+|u|)^{-6}\,dudv \\
&\times \left[\int_{\underline{H}_{v_0}\cap\{|u|\geq |u_1|\}}J^{N_{2,\eta}}[\Gamma\phi]\cdot \underline{L}+\int_{\mathcal{H}^+\cap\{v\geq v_0\}}J^{N_{2,\eta}}[\Gamma\phi]\cdot L\right]\\
\leq \:& C \int_{\mathcal{H}^+\cap\{v\geq v_0\}}v^{-4+\eta+\max\{p,q\}}(\Gamma\phi)^2\\
&+C\left[\int_{\underline{H}_{v_0}\cap\{|u|\geq |u_1|\}}J^{N_{2,\eta}}[\Gamma\phi]\cdot \underline{L}+\int_{\mathcal{H}^+\cap\{v\geq v_0\}}J^{N_{2,\eta}}[\Gamma\phi]\cdot L\right].
\end{split}
\end{equation*}

Similarly, we can estimate in $\{v\geq v_1\}$
\begin{equation*}
\begin{split}
\int_{v_1}^{\infty}\int_{-\infty}^{u_0}&(v^q|u|^{1+\eta}+|u|^pv^{1+\eta})(v+|u|)^{-6}\int_{S^2_{u,v}}(\Gamma\phi)^2\,d\mu_{\slashed{g}}dudv\\
\leq \:& C \int_{\mathcal{H}^+\cap\{v\geq v_1\}}v^{-4+\eta+\max\{p,q\}}(\Gamma\phi)^2\\
&+C\left[\int_{\underline{H}_{v_1}}J^{N_{2,\eta}}[\Gamma\phi]\cdot \underline{L}+\int_{\mathcal{H}^+\cap\{v\geq v_1\}}J^{N_{2,\eta}}[\Gamma\phi]\cdot L\right].
\end{split}
\end{equation*}

The remaining terms on the right-hand side of (\ref{est:2ndorderangder}) can be estimated by energy fluxes through $H_u$ and $\underline{H}_v$, multiplied by integrable functions $h(u)$ or $k(v)$. We can apply Lemma \ref{lm:gronwall} in the regions $\uline{\mathcal{A}}$ and $\mathcal{B}$, with $|u|\geq |u_1|$ or $v\geq v_1$. Furthermore, we can estimate in a similar way the terms in $\mathcal{E}^{Y_p}[L\phi]$ and $\mathcal{E}^{Y_p}[\underline{L}\phi]$ in the region $\mathcal{B}$, with $|u|\geq |u_1|$ or $v\geq v_1$.

We combine in the $|a|<a_c$ case, as in Proposition \ref{energyestsmallakn}, the estimates with respect to the multipliers $N_{p,q}$ and $Y_p$.
\end{proof}

We can easily commute $\square_g$ with $L$, $\underline{L}$ and $\partial_{\theta_*}$ in the region $\{|u|\leq |u_1|,\,v\leq v_1\}$. As $|u|$ and $v$ are both finite in this region, we do not need to keep track of the behaviour in $v+|u|$ of the error terms.

\begin{proposition}
\label{prop:commfinitevolregion}
Let $0\leq p,q \leq 2$ and $v_1>v_0$, $u_1<u_0$. Then there exists a \\$C=C(a,M,p,q,u_0,v_0,u_1,v_1)>0$ such that for all $k\in \N$,
\begin{equation*}
\begin{split}
\sum_{0\leq j_1+j_2+j_3+j_4\leq k}&\int_{\underline{H}_{v}\cap \{|u|\leq |u_1\}}J^{N_{p,q}}[L^{j_1}\underline{L}^{j_2}\partial_{\theta_*}^{j_3}\Phi^{j_4}\phi]\cdot \underline{L}+\int_{{H}_{u}\cap\{v\leq v_1\}}J^{N_{p,q}}[L^{j_1}\underline{L}^{j_2}\partial_{\theta_*}^{j_3}\Phi^{j_4}\phi]\cdot L\\
\leq \:& C \sum_{0\leq j_1+j_2+j_3+j_4\leq k}\int_{\underline{H}_{v_1}\cap\{|u|\leq |u_1|\}}J^{N_{p,q}}[L^{j_1}\underline{L}^{j_2}\partial_{\theta_*}^{j_3}\Phi^{j_4}\phi]\cdot \underline{L}\\
&+\int_{H_{u_1}\cap\{v\leq v_1\}}J^{N_{p,q}}[L^{j_1}\underline{L}^{j_2}\partial_{\theta_*}^{j_3}\Phi^{j_4}\phi]\cdot L.
\end{split}
\end{equation*}
\end{proposition}
\begin{proof}
It easily follows that
\begin{equation*}
\begin{split}
\sum_{0\leq j_1+j_2+j_3+j_4\leq k}&|\mathcal{E}^{N_{p,q}}[L^{j_1}\underline{L}^{j_2}\partial_{\theta_*}^{j_3}\Phi^{j_4}\phi]|+|K^{N_{p,q}}[L^{j_1}\underline{L}^{j_2}\partial_{\theta_*}^{j_3}\Phi^{j_4}\phi]|\\
&\lesssim \sum_{0\leq j_1+j_2+j_3+j_4\leq k}J^{N_{p,q}}[L^{j_1}\underline{L}^{j_2}\partial_{\theta_*}^{j_3}\Phi^{j_4}\phi]\cdot \underline{L}+ J^{N_{p,q}}[L^{j_1}\underline{L}^{j_2}\partial_{\theta_*}^{j_3}\Phi^{j_4}\phi]\cdot L,
\end{split}
\end{equation*}
as the energy fluxes together control all derivatives, since $|u|$ and $v$ are bounded. We can therefore directly apply Lemma \ref{lm:gronwall} in $\{|u|\leq |u_1|,\,v\leq v_1\}$ to obtain the estimate in the proposition.
\end{proof}
We also commute with higher-order derivatives along null vector fields in the region $\{v\geq v_1\}\cup \{|u|\geq |u_1|\}$. In this case, we do need to keep track of the behaviour in $v+|u|$ of the error terms arising from commuting with $L$ and $\uline{L}$.
\begin{lemma}
\label{lm:waveoperatorcomm}
Let $n\in \N_0$. Then there exists a constant $C=C(a,M,v_0,u_0,n)>0$, such that
\begin{equation}
\label{est:hocommutation}
\begin{split}
\left|\sum_{j_1+j_2=n}2\Omega^2\square_g(L^{j_1}\uline{L}^{j_2}\phi)\right|\leq &\:C (v+|u|)^{-2} \sum_{j_1+j_2\leq n} |L^{j_1}\uline{L}^{j_2} \Phi \phi|\\
&+ C(v+|u|)^{-3}\sum_{j_1+j_2\leq n-1} |\snabla L^{j_1}\uline{L}^{j_2}\phi|+|\snabla^2L^{j_1}\uline{L}^{j_2}\phi|\\
&+C(v+|u|)^{-3}\sum_{j_1+j_2+j_3\leq n-1} |L^{j_1}\uline{L}^{j_2}\Phi^{j_3+1}\phi|.
\end{split}
\end{equation}
\end{lemma}
\begin{proof}
By Lemma \ref{lm:commwithV}, we have that
\begin{equation*}
\begin{split}
2\Omega^2\square_g(L^{j_1}\uline{L}^{j_2}\phi)=&\:[L^{j_1}\uline{L}^{j_2},\underline{L}]L\phi+\underline{L}([L^{j_1}\uline{L}^{j_2},L]\phi)+[L^{j_1}\uline{L}^{j_2},L]\underline{L}\phi+L([L^{j_1}\uline{L}^{j_2},\underline{L}]\phi)\\
&\sum_{l=1}^{j_1}\sum_{k=1}^{j_2} \binom{j_1}{l}\binom{j_2}{k} L^l\uline{L}^k(\Omega\tr \underline{\chi})L^{j_1-l}\uline{L}^{j_2-k}L\phi+\Omega \tr {\underline{\chi}}[L^{j_1}\uline{L}^{j_2},L]\phi\\
&+ \sum_{l=1}^{j_1}\sum_{k=1}^{j_2} \binom{j_1}{l}\binom{j_2}{k} L^l\uline{L}^k(\Omega\tr {\chi})L^{j_1-l}\uline{L}^{n-k}\underline{L}\phi+\Omega \tr {\chi}[L^{j_1}\uline{L}^{j_2},\underline{L}]\phi\\
&-2\sum_{l=1}^{j_1}\sum_{k=1}^{j_2}  \binom{j_1}{l}\binom{j_2}{k} L^l\uline{L}^k(\slashed{g}^{AB}\partial_A\Omega^2)L^{j_1-l}\uline{L}^{j_2-k}\partial_B\phi-2\slashed{g}^{AB}\partial_A\Omega^2[L^{j_1}\uline{L}^{j_2},\partial_B]\phi\\
&-2\sum_{l=1}^{j_1}\sum_{k=1}^{j_2}  \binom{j_1}{l}\binom{j_2}{k} L^l\uline{L}^k(\Omega^2)L^{j_1-l}\uline{L}^{j_2-k}\slashed{\Delta}\phi-2\Omega^2[L^{j_1}\uline{L}^{j_2},\slashed{\Delta}]\phi.
\end{split}
\end{equation*}
We have that
\begin{equation*}
[L,\uline{L}]=L(b^{\varphi_*})\Phi.
\end{equation*}
It follows that
\begin{align*}
[L^{j_1}\uline{L}^{j_2},L](f)=&\:\sum_{k=1}^{j_2}L^{j_1}\uline{L}^{j_2-k}[\uline{L},L](\uline{L}^{k-1}f)\\
=&\:-\sum_{k=1}^{j_2}L^{j_1}\uline{L}^{j_2-k}(L(b^{\varphi_*})\uline{L}^{k-1}\Phi f),\\
[L^{j_1}\uline{L}^{j_2},\uline{L}](f)=&\:\sum_{k=1}^{j_1}{L}^{j_1-k}[L,\uline{L}](L^{k-1}\uline{L}^{j_2}f)\\
=&\:\sum_{k=1}^{j_1}{L}^{j_1-k}(L(b^{\varphi_*})L^{k-1}\uline{L}^{j_2}\Phi f)
\end{align*}

By making use of the estimate for $\partial_{r_*}^n b^{\varphi_*}$ from Theorem \ref{thm:estmetric}, it follows that
\begin{equation*}
\begin{split}
\sum_{j_1+j_2=n}&|[L^{j_1}\uline{L}^{j_2},\underline{L}]L\phi+\underline{L}([L^{j_1}\uline{L}^{j_2},L]\phi)+[L^{j_1}\uline{L}^{j_2},L]\underline{L}\phi+L([L^{j_1}\uline{L}^{j_2},\underline{L}]\phi)|\\
\leq &\: C (v+|u|)^{-2} \sum_{j_1+j_2\leq n} |L^{j_1}\uline{L}^{j_2} \Phi \phi|+ C(v+|u|)^{-3}\sum_{j_1+j_2+j_3\leq n-1} |L^{j_1}\uline{L}^{j_2}\Phi^{j_3+1}\phi|.
\end{split}
\end{equation*}

Furthermore, we have that
\begin{equation*}
[L,\partial_{\theta_*}]=-\partial_{\theta_*}b^{\varphi_*}\partial_{\varphi_*},
\end{equation*}
so we obtain:
\begin{equation*}
[L^{j_1}\uline{L}^{j_2},\partial_{\theta_*}](f)=\sum_{k=1}^{j_1}L^{j_1-k}\left(\Phi b^{\varphi_*}\partial_{\theta_*}L^{k-1}\uline{L}^{j_2}f\right).
\end{equation*}

Now, we make use of the estimate for $\partial_{r_*}^n\partial_{\theta_*} b^{\varphi_*}$ from Theorem \ref{thm:estmetric}, with $k\leq 4$, to estimate
\begin{equation*}
\begin{split}
\sum_{j_1+j_2=n}|[L^{j_1}\uline{L}^{j_2},\partial_{\theta_*}]\phi|\leq  C(v+|u|)^{-1}\sum_{j_1+j_2\leq n-1} |L^{j_1}\uline{L}^{j_2}\Phi \phi|^2.
\end{split}
\end{equation*}

Recall from Proposition \ref{commenergyestimateskn1} that
\begin{align*}
[\slashed{\Delta},\underline{L}]\phi=\:&-\underline{L}\left(\frac{1}{\sqrt{\det \slashed{g}}}\right)\partial_A\left(\slashed{g}^{AB}\sqrt{\det \slashed{g}}\partial_B\phi\right)-\frac{1}{\sqrt{\det \slashed{g}}}\partial_A\left(\underline{L}\left(\slashed{g}^{AB}\sqrt{\det \slashed{g}}\right)\partial_B\phi\right),\\
[\slashed{\Delta},L]\phi=\:&-L\left(\frac{1}{\sqrt{\det \slashed{g}}}\right)\partial_A\left(\slashed{g}^{AB}\sqrt{\det \slashed{g}}\partial_B\phi\right)-\frac{1}{\sqrt{\det \slashed{g}}}\partial_A\left(L\left(\slashed{g}^{AB}\sqrt{\det \slashed{g}}\right)\partial_B\phi\right)\\
&-\frac{1}{\sqrt{\det \slashed{g}}}[L,\partial_A]\left(\slashed{g}^{AB}\sqrt{\det \slashed{g}}\partial_B\phi\right)-\slashed{g}^{AB}\partial_A([L,\partial_B]\phi).
\end{align*}
We can therefore estimate
\begin{equation*}
\begin{split}
[L^{j_1}\uline{L}^{j_2},\slashed{\Delta}](f)=&\:\sum_{k=1}^{j_2}L^{j_1}\uline{L}^{j_2-k}[\uline{L},\slashed{\Delta}](\uline{L}^{k-1}f)+\sum_{k=1}^{j_1}L^{j_1-k}[L,\slashed{\Delta}](L^{k-1}\uline{L}^{j_1}f).
\end{split}
\end{equation*}

Hence, we obtain by using the estimates for $\partial_{r_*}^k\slashed{g}_{AB}$ from Theorem \ref{thm:estmetric},
\begin{equation*}
\sum_{j_1+j_2=n} |[L^{j_1}\uline{L}^{j_2},\slashed{\Delta}]\phi|\leq (v+|u|)^{-1}\sum_{j_1+j_2\leq n-1} |\snabla L^{j_1}\uline{L}^{j_2}\phi|+|\snabla^2 L^{j_1}\uline{L}^{j_2}\phi|.
\end{equation*}

We can easily estimate the remaining terms in $2\Omega^2\square_g(L^{j_1}\uline{L}^{j_2}\phi)$, applying the estimates from Section~ \ref{sec:estimatesmetricconn}, to conclude that (\ref{est:hocommutation}) must hold.
\end{proof}

\begin{proposition}
\label{commutedenergyestkn}
Let $k\in \N_0$. Either restrict to $|a|<a_c$ and let $0\leq p<2$ and $0\leq q<2$, or restrict to axisymmetric $\phi$, with $p=2$ and $0<q\leq 2$. Then there exist $\alpha=\alpha(p,q)>1$ and $\beta=\beta(p,q)>1$, such that for $u_1$ suitably large, $|u|\geq |u_1|$:\begin{equation*}
\begin{split}
\sum_{j_1+j_2=2k+1}& \int_{\uline{H}_v\cap \{|u|\geq |u_1|\}}J^{N_{p,q}}[L^{j_1}\uline{L}^{j_2}\phi]\cdot \uline{L}+\int_{H_u}J^{N_{p,q}}[L^{j_1}\uline{L}^{j_2}\phi]\cdot L\\
\leq &\: C\sum_{j_1+j_2+2j_3+j_4\leq 2k+1}\int_{\uline{H}_{v_0}}J^{N_{2,q\beta \alpha}}[L^{j_1}\uline{L}^{j_2}Q^{j_3}\Phi^{j_4}\phi]\cdot \uline{L}+\int_{\mathcal{H}^+}J^{N_{2,q\beta \alpha}}[L^{j_1}\uline{L}^{j_2}Q^{j_3}\Phi^{j_4}\phi]\cdot L\\
&+C\sum_{2j_1+j_2\leq 2k+1}\int_{\uline{H}_{v_0}}J^{N_{2,q\beta \alpha}}[Q^{j_1}\Phi^{j_2+1}\phi]\cdot \uline{L}+\int_{\mathcal{H}^+}J^{N_{2,q\beta \alpha}}[Q^{j_1}\Phi^{j_2+1}\phi]\cdot L\\
&+C\sum_{2j_1+j_2\leq 2k}\int_{\uline{H}_{v_0}}J^{N_{2,q\beta \alpha}}[Q^{j_1+1}\Phi^{j_2}\phi]\cdot \uline{L}+\int_{\mathcal{H}^+}J^{N_{2,q\beta \alpha}}[Q^{j_1}\Phi^{j_2}\phi]\cdot L\\
&+C\sum_{j_1+2j_2\leq 2k+2}\int_{\mathcal{H}^+}v^{-2+\epsilon}(\Phi^{j_1}Q^{j_2}\phi)^2,
\end{split}
\end{equation*}
and
\begin{equation*}
\begin{split}
\sum_{j_1+j_2=2k}& \int_{\uline{H}_v\cap \{|u|\geq |u_1|\}}J^{N_{p,q}}[L^{j_1}\uline{L}^{j_2}\phi]\cdot \uline{L}+\int_{H_u}J^{N_{p,q}}[L^{j_1}\uline{L}^{j_2}\phi]\cdot L\\
\leq &\: C\sum_{j_1+j_2+2j_3+j_4\leq 2k}\int_{\uline{H}_{v_0}}J^{N_{2,q\beta \alpha}}[L^{j_1}\uline{L}^{j_2}Q^{j_3}\Phi^{j_4}\phi]\cdot \uline{L}+\int_{\mathcal{H}^+}J^{N_{2,q\beta \alpha}}[L^{j_1}\uline{L}^{j_2}Q^{j_3}\Phi^{j_4}\phi]\cdot L\\
&+C\sum_{2j_1+j_2\leq 2k}\int_{\uline{H}_{v_0}}J^{N_{2,q\beta \alpha}}[Q^{j_1}\Phi^{j_2+1}\phi]\cdot \uline{L}+\int_{\mathcal{H}^+}J^{N_{2,q\beta \alpha}}[Q^{j_1}\Phi^{j_2+1}\phi]\cdot L\\
&+C\sum_{j_1+2j_2\leq 2k+1}\int_{\mathcal{H}^+}v^{-2+\epsilon}(\Phi^{j_1}Q^{j_2}\phi)^2,
\end{split}
\end{equation*}
where $0<q\beta \alpha \leq 2$, $\epsilon>0$ can be taken arbitrarily small and $$C=C(k,a,M,u_0,v_0,u_1,v_1,p,q,\alpha,\beta,\epsilon)>0.$$ For axisymmetric $\phi$, we can replace $N_{2,q\beta \alpha}$ on the right-hand side by $N_{2,q}$.

Moreover, for $v_1>v_0$ suitably large and $v\geq v_1$,
\begin{equation*}
\begin{split}
\sum_{j_1+j_2=2k+1}& \int_{\uline{H}_v}J^{N_{p,q}}[L^{j_1}\uline{L}^{j_2}\phi]\cdot \uline{L}+\int_{H_u\cap\{v\geq v_1\}}J^{N_{p,q}}[L^{j_1}\uline{L}^{j_2}\phi]\cdot L\\
\leq &\: C\sum_{j_1+j_2+2j_3+j_4\leq 2k+1}\int_{\uline{H}_{v_1}}J^{N_{2,q\beta \alpha}}[L^{j_1}\uline{L}^{j_2}Q^{j_3}\Phi^{j_4}\phi]\cdot \uline{L}+\int_{\mathcal{H}^+\cap\{v\geq v_1\}}J^{N_{2,q\beta \alpha}}[L^{j_1}\uline{L}^{j_2}Q^{j_3}\Phi^{j_4}\phi]\cdot L\\
&+C\sum_{2j_1+j_2\leq 2k+1}\int_{\uline{H}_{v_1}}J^{N_{2,q\beta \alpha}}[Q^{j_1}\Phi^{j_2+1}\phi]\cdot \uline{L}+\int_{\mathcal{H}^+\cap\{v\geq v_1\}}J^{N_{2,q\beta \alpha}}[Q^{j_1}\Phi^{j_2+1}\phi]\cdot L\\
&+C\sum_{2j_1+j_2\leq 2k}\int_{\uline{H}_{v_1}}J^{N_{2,q\beta \alpha}}[Q^{j_1+1}\Phi^{j_2}\phi]\cdot \uline{L}+\int_{\mathcal{H}^+\cap\{v\geq v_1\}}J^{N_{2,q\beta \alpha}}[Q^{j_1}\Phi^{j_2}\phi]\cdot L\\
&+C\sum_{j_1+2j_2\leq 2k+2}\int_{\mathcal{H}^+\cap\{v\geq v_1\}}v^{-2+\epsilon}(\Phi^{j_1}Q^{j_2}\phi)^2,
\end{split}
\end{equation*}
and
\begin{equation*}
\begin{split}
\sum_{j_1+j_2=2k}& \int_{\uline{H}_v}J^{N_{p,q}}[L^{j_1}\uline{L}^{j_2}\phi]\cdot \uline{L}+\int_{H_u\cap\{v\geq v_1\}}J^{N_{p,q}}[L^{j_1}\uline{L}^{j_2}\phi]\cdot L\\
\leq &\: C\sum_{j_1+j_2+2j_3+j_4\leq 2k}\int_{\uline{H}_{v_1}}J^{N_{2,q\beta \alpha}}[L^{j_1}\uline{L}^{j_2}Q^{j_3}\Phi^{j_4}\phi]\cdot \uline{L}+\int_{\mathcal{H}^+\cap\{v\geq v_1\}}J^{N_{2,q\beta \alpha}}[L^{j_1}\uline{L}^{j_2}Q^{j_3}\Phi^{j_4}\phi]\cdot L\\
&+C\sum_{2j_1+j_2\leq 2k}\int_{\uline{H}_{v_1}}J^{N_{2,q\beta \alpha}}[Q^{j_1}\Phi^{j_2+1}\phi]\cdot \uline{L}+\int_{\mathcal{H}^+\cap\{v\geq v_1\}}J^{N_{2,q\beta \alpha}}[Q^{j_1}\Phi^{j_2+1}\phi]\cdot L\\
&+C\sum_{j_1+2j_2\leq 2k+1}\int_{\mathcal{H}^+\cap\{v\geq v_1\}}v^{-2+\epsilon}(\Phi^{j_1}Q^{j_2}\phi)^2,
\end{split}
\end{equation*}
where $0<q\beta \alpha \leq 2$, $\epsilon>0$ can be taken arbitrarily small and $$C=C(k,a,M,u_0,v_0,u_1,v_1,p,q,\alpha,\beta,\epsilon)>0.$$ For axisymmetric $\phi$, we can replace $N_{2,q\beta \alpha}$ on the right-hand side by $N_{2,q}$.
\end{proposition}
\begin{proof}
We have that
\begin{equation*}
\sum_{j_1+j_2=n}2\Omega^2 \mathcal{E}^{N_{p,q}}[L^{j_1}\uline{L}^{j_2}\phi]=\sum_{j_1+j_2=n}N_{p,q}(L^{j_1}\uline{L}^{j_2}\phi)2\Omega^2\square_{g}(L^{j_1}\uline{L}^{j_2}\phi).
\end{equation*}
From (\ref{est:hocommutation}) it follows that
\begin{equation*}
\begin{split}
\sum_{j_1+j_2=n}2\Omega^2 |\mathcal{E}^{N_{p,q}}[L^{j_1}\uline{L}^{j_2}\phi]|\leq &\:C\left(|u|^p|\uline{L}L^{j_1}\uline{L}^{j_2}\phi|+v^q|LL^{j_1}\uline{L}^{j_2}\phi|\right)\\
&\cdot \Bigg[(v+|u|)^{-2} \sum_{j_1+j_2\leq n} |L^{j_1}\uline{L}^{j_2} \Phi \phi|\\
&+ (v+|u|)^{-3}\sum_{j_1+j_2\leq n-1} |\snabla L^{j_1}\uline{L}^{j_2}\phi|+|\snabla^2L^{j_1}\uline{L}^{j_2}\phi|\\
&+(v+|u|)^{-3}\sum_{j_1+j_2+j_3\leq n-1} |L^{j_1}\uline{L}^{j_2}\Phi^{j_3+1}\phi|\Bigg].
\end{split}
\end{equation*}
We apply Cauchy--Schwarz to further estimate
\begin{equation*}
\begin{split}
\sum_{j_1+j_2=n}2\Omega^2 |\mathcal{E}^{N_{p,q}}[L^{j_1}\uline{L}^{j_2}\phi]|\leq &\:C\sum_{j_1+j_2\leq n}v^{-1-\eta}|u|^p(\uline{L}L^{j_1}\uline{L}^{j_2}\phi)^2+|u|^{-1-\eta}v^q(LL^{j_1}\uline{L}^{j_2}\phi)^2\\
&+C\left(v^{1+\eta}|u|^p+|u|^{1+\eta}v^q\right)\Bigg[(v+|u|)^{-4}\sum_{j_1+j_2\leq n} |L^{j_1}\uline{L}^{j_2} \Phi \phi|^2\\
&+ (v+|u|)^{-6}\sum_{j_1+j_2\leq n-1} |\snabla L^{j_1}\uline{L}^{j_2}\phi|^2+|\snabla^2L^{j_1}\uline{L}^{j_2}\phi|^2\\
&+(v+|u|)^{-6}\sum_{j_1+j_2+j_3\leq n-1} |L^{j_1}\uline{L}^{j_2}\Phi^{j_3+1}\phi|^2\Bigg].
\end{split}
\end{equation*}
We can apply (\ref{est:ellipticspherev2}) to obtain
\begin{equation*}
\begin{split}
\sum_{j_1+j_2\leq n-1} \int_{S^2_{u,v}}|\snabla^2L^{j_1}\uline{L}^{j_2}\phi|^2\,d\mu_{\slashed{g}}\leq &\:C \int_{S^2_{u,v}}\sum_{j_1+j_2\leq n-1}(Q L^{j_1}\uline{L}^{j_2}\phi)^2+(T^2 L^{j_1}\uline{L}^{j_2}\phi)^2+( L^{j_1}\uline{L}^{j_2}\Phi^2 \phi)^2\\
&+\sum_{j_1+j_2\leq n+1}  (L^{j_1}\uline{L}^{j_2}\phi)^2+ \sum_{j_1+j_2\leq n}|\snabla  L^{j_1}\uline{L}^{j_2}\phi|^2\,d\mu_{\slashed{g}}.
\end{split}
\end{equation*}
We can further estimate
\begin{equation*}
\sum_{j_1+j_2\leq n-1} (Q L^{j_1}\uline{L}^{j_2}\phi)^2= \sum_{j_1+j_2\leq n-1} (L^{j_1}\uline{L}^{j_2} Q \phi)^2+([Q,L^{j_1}\uline{L}^{j_2}]\phi)^2.
\end{equation*}
We have that
\begin{equation*}
\begin{split}
[Q,L^{j_1}\uline{L}^{j_2}](f)=&\:\sum_{k=1}^{j_2}L^{j_1}\uline{L}^{j_2-k}[Q,\uline{L}](\uline{L}^{k-1}f)+\sum_{k=1}^{j_1}L^{j_1-k}[Q,L](L^{k-1}\uline{L}^{j_2}f).
\end{split}
\end{equation*}
Furthermore,
\begin{align*}
[Q,\underline{L}](f)=&\:[\Delta_{\s^2},\underline{L}](f)+a^2\underline{L}(\sin^2\theta)T^2f\\
=&\:[\sin{\theta}^{-1}\partial_{\theta}(\sin\theta\partial_{\theta}),\underline{L}]f-\underline{L}(\sin^{-2}\theta)\Phi^2f+a^2\underline{L}(\sin^2\theta)T^2f,\\
[Q,L](f)=&\:[\Delta_{\s^2},L](f)+a^2L(\sin^2\theta)T^2f\\
=&\:[\sin{\theta}^{-1}\partial_{\theta}(\sin\theta\partial_{\theta}),L]f-L(\sin^{-2}\theta)\Phi^2f+a^2L(\sin^2\theta)T^2f,
\end{align*}
with
\begin{align*}
[\underline{L},\partial_{\theta}]f=&\:\underline{L}(\partial_{\theta}\theta_*)\partial_{\theta_*}f+\frac{1}{2}\underline{L}(\partial_{\theta}r_*)(Lf-b^{{\varphi_*}}\Phi f-\underline{L}f),\\
[L,\partial_{\theta}]f=&\:L(\partial_{\theta}\theta_*)\partial_{\theta_*}f+\frac{1}{2}L(\partial_{\theta}r_*)(Lf-b^{{\varphi_*}}\Phi f-\underline{L}f).
\end{align*}
By using the estimates from Section~ \ref{sec:estimatesmetricconn}, we can estimate
\begin{equation*}
\begin{split}
|[Q,L^{j_1}\uline{L}^{j_2}](f)|\leq&\: C(v+|u|)^{-2}\Bigg(\sum_{k+l\leq j_1+j_2}|\snabla^2L^l\uline{L}^k\phi|+\sum_{k+l\leq j_1+j_2+1}|\snabla L^l\uline{L}^k\phi|\\
&+\sum_{k+l\leq j_1+j_2+2}|L^l\uline{L}^k\phi|\Bigg).
\end{split}
\end{equation*}

We conclude that, for either $v\geq v_1$ or $|u|\geq |u_1|$, with $v_1$ and $|u_1|$ suitably large
\begin{equation*}
\begin{split}
\sum_{j_1+j_2\leq n-1} \int_{S^2_{u,v}}|\snabla^2L^{j_1}\uline{L}^{j_2}\phi|^2\,d\mu_{\slashed{g}}\leq &\:C \int_{S^2_{u,v}}\sum_{j_1+j_2\leq n-1} (L^{j_1}\uline{L}^{j_2}Q\phi)^2+( L^{j_1}\uline{L}^{j_2}\Phi^2 \phi)^2\\
&+\sum_{j_1+j_2\leq n+1}  (L^{j_1}\uline{L}^{j_2}\phi)^2+ \sum_{j_1+j_2\leq n}|\snabla  L^{j_1}\uline{L}^{j_2}\phi|^2\,d\mu_{\slashed{g}}.
\end{split}
\end{equation*}
We can therefore use the above estimate to obtain
\begin{equation*}
\begin{split}
\sum_{j_1+j_2=n}\int_{S^2_{u,v}}&2\Omega^2 |\mathcal{E}^{N_{p,q}}[L^{j_1}\uline{L}^{j_2}\phi]|\,d\mu_{\slashed{g}}\\
\leq &\:C\sum_{j_1+j_2\leq n}\int_{S^2_{u,v}}2v^{-1-\eta}|u|^p(\uline{L}L^{j_1}\uline{L}^{j_2}\phi)^2+|u|^{-1-\eta}v^q(LL^{j_1}\uline{L}^{j_2}\phi)^2\,d\mu_{\slashed{g}}\\
&+C\int_{S^2_{u,v}}2\left(v^{1+\eta}|u|^p+|u|^{1+\eta}v^q\right)\Bigg[(v+|u|)^{-4}\sum_{j_1+j_2\leq n} |L^{j_1}\uline{L}^{j_2} \Phi \phi|^2\\
&+ (v+|u|)^{-6}\sum_{j_1+j_2\leq n} |\snabla L^{j_1}\uline{L}^{j_2}\phi|^2\\
&+(v+|u|)^{-6}\sum_{j_1+j_2\leq n-1} (L^{j_1}\uline{L}^{j_2}Q\phi)^2+( L^{j_1}\uline{L}^{j_2}\Phi^2 \phi)^2\\
&+(v+|u|)^{-6}\sum_{j_1+j_2+j_3\leq n-1} |L^{j_1}\uline{L}^{j_2}\Phi^{j_3+1}\phi|^2\Bigg]\,d\mu_{\slashed{g}},
\end{split}
\end{equation*}
for $0\leq p,q\leq 2$. We obtain similar estimates for $\sum_{j_1+j_2=n}2\Omega^2|\mathcal{E}^{Y_p}[L^{j_1}\uline{L}^{j_2}\phi]|$, by replacing the weight $v^q$ above by $|u|^p$.

We apply the divergence theorem, as in the previous propositions, to obtain the following energy estimate:
\begin{equation*}
\begin{split}
\sum_{j_1+j_2=n}& \int_{\uline{H}_v}J^{N_{p,q}}[L^{j_1}\uline{L}^{j_2}\phi]\cdot \uline{L}+\int_{H_u\cap\{v\geq v_1\}}J^{N_{p,q}}[L^{j_1}\uline{L}^{j_2}\phi]\cdot L\\
\leq &\: C\sum_{j_1+j_2=n} \int_{\uline{H}_{v_1}}J^{N_{2,q\beta \alpha}}[L^{j_1}\uline{L}^{j_2}\phi]\cdot \uline{L}+\int_{\mathcal{H}^+\cap\{v\geq v_1\}}J^{N_{2,q\beta \alpha}}[L^{j_1}\uline{L}^{j_2}\phi]\cdot L\\
&+C\sup_{v_1\leq v<\infty}\int_{\uline{H}_v} \sum_{j_1+j_2\leq n-2}\sum_{\Gamma\in\{Q,\Phi^2\}}J^{N_{p,q}}[L^{j_1}\uline{L}^{j_2}\Gamma\phi]\cdot \uline{L}\\
&+\sum_{j_1+j_2\leq n-1}J^{N_{p,q}}[L^{j_1}\uline{L}^{j_2}\Phi \phi]\cdot \uline{L}\\
&+ \sum_{j_1+j_2+j_3\leq n-2}J^{N_{p,q}}[L^{j_1}\uline{L}^{j_2}\Phi^{j_3+1}\phi]\cdot \uline{L}\\
&+C\sup_{-\infty \leq u<u_0}\int_{H_u\cap\{v\geq v_1\}} \sum_{j_1+j_2\leq n-2}\sum_{\Gamma\in\{Q,\Phi^2\}}J^{N_{p,q}}[L^{j_1}\uline{L}^{j_2}\Gamma\phi]\cdot L\\
&+\sum_{j_1+j_2\leq n-1}J^{N_{p,q}}[L^{j_1}\uline{L}^{j_2}\Phi \phi]\cdot L\\
&+ \sum_{j_1+j_2+j_3\leq n-2}J^{N_{p,q}}[L^{j_1}\uline{L}^{j_2}\Phi^{j_3+1}\phi]\cdot L.
\end{split}
\end{equation*}
By induction, using the estimates Proposition \ref{commenergyestimateskn1}, it follows that for odd $n$:
\begin{equation*}
\begin{split}
\sum_{j_1+j_2=n}& \int_{\uline{H}_v}J^{N_{p,q}}[L^{j_1}\uline{L}^{j_2}\phi]\cdot \uline{L}+\int_{H_u\cap\{v\geq v_1\}}J^{N_{p,q}}[L^{j_1}\uline{L}^{j_2}\phi]\cdot L\\
\leq &\: C\sum_{j_1+j_2+2j_3+j_4\leq n}\int_{\uline{H}_{v_1}}J^{N_{2,\beta \alpha}}[L^{j_1}\uline{L}^{j_2}Q^{j_3}\Phi^{j_4}\phi]\cdot \uline{L}+\int_{\mathcal{H}^+\cap\{v\geq v_1\}}J^{N_{2,\beta \alpha}}[L^{j_1}\uline{L}^{j_2}Q^{j_3}\Phi^{j_4}\phi]\cdot L\\
&+C\sum_{2j_1+j_2+j_3\leq n}\int_{\uline{H}_{v_1}}J^{N_{2,\beta \alpha}}[Q^{j_1}\Phi^{j_2+1}T^{j_3}\phi]\cdot \uline{L}+\int_{\mathcal{H}^+\cap\{v\geq v_1\}}J^{N_{2,\beta \alpha}}[Q^{j_1}\Phi^{j_2+1}T^{j_3}\phi]\cdot L\\
&+C\sum_{2j_1+j_2+j_3\leq n-1}\int_{\uline{H}_{v_1}}J^{N_{2,\beta \alpha}}[Q^{j_1+1}\Phi^{j_2}T^{j_3}\phi]\cdot \uline{L}+\int_{\mathcal{H}^+\cap\{v\geq v_1\}}J^{N_{2,\beta \alpha}}[Q^{j_1}\Phi^{j_2}T^{j_3}\phi]\cdot L\\
&+C\sum_{j_1+2j_2+j_3\leq n+1}\int_{\mathcal{H}^+\cap\{v\geq v_1\}}v^{-2+\epsilon}(\Phi^{j_1}Q^{j_2}T^{j_3}\phi)^2,
\end{split}
\end{equation*}
whereas for even $n$ we can estimate
\begin{equation*}
\begin{split}
\sum_{j_1+j_2=n}& \int_{\uline{H}_v}J^{N_{p,q}}[L^{j_1}\uline{L}^{j_2}\phi]\cdot \uline{L}+\int_{H_u\cap\{v\geq v_1\}}J^{N_{p,q}}[L^{j_1}\uline{L}^{j_2}\phi]\cdot L\\
\leq &\: C\sum_{j_1+j_2+2j_3+j_4\leq n}\int_{\uline{H}_{v_1}}J^{N_{2,\beta \alpha}}[L^{j_1}\uline{L}^{j_2}Q^{j_3}\Phi^{j_4}\phi]\cdot \uline{L}+\int_{\mathcal{H}^+\cap\{v\geq v_1\}}J^{N_{2,\beta \alpha}}[L^{j_1}\uline{L}^{j_2}Q^{j_3}\Phi^{j_4}\phi]\cdot L\\
&+C\sum_{2j_1+j_2+j_3\leq n}\int_{\uline{H}_{v_1}}J^{N_{2,\beta \alpha}}[Q^{j_1}\Phi^{j_2+1}T^{j_3}\phi]\cdot \uline{L}+\int_{\mathcal{H}^+\cap\{v\geq v_1\}}J^{N_{2,\beta \alpha}}[Q^{j_1}\Phi^{j_2+1}T^{j_3}\phi]\cdot L\\
&+C\sum_{j_1+2j_2+j_3\leq n+1}\int_{\mathcal{H}^+\cap\{v\geq v_1\}}v^{-2+\epsilon}(\Phi^{j_1}Q^{j_2}T^{j_3}\phi)^2,
\end{split}
\end{equation*}
The estimates in the region $\{|u|\geq |u_1|\}$ proceed similarly.
\end{proof}

\begin{remark}
In the estimates of Proposition \ref{commutedenergyestkn}, we estimate integrals of $n$ derivatives of $\phi$ by initial integrals of $n+1$ derivatives. The loss of derivatives in the even $n$ case arises only because of $\Phi$. Therefore, if $\phi$ is axisymmetric and we can neglect the $\Phi$ derivatives, we \underline{do not lose any derivatives} in Proposition \ref{commutedenergyestkn} for even $n$. This fact is important when proving energy estimates for nonlinear wave equations.
\end{remark}

We combine the results of Propositions \ref{commutedenergyestkn} and \ref{prop:commfinitevolregion} to obtain an energy estimate in the entire region ${\mathcal{D}_{u_0,v_0}}$.
\begin{corollary}
\label{cor:commeestimate}
Let $k\in \N_0$. Restrict to $|a|<a_c$ and let $0\leq p<2$ and $0\leq q<2$, or restrict to axisymmetric $\phi$, with $p=2$ and $0<q\leq 2$. Then there exist $\alpha=\alpha(p,q)>1$ and $\beta=\beta(p,q)>1$ and a constant $C=C(n,a,M,u_0,v_0,p,q,\alpha,\beta)>0$, such that
\begin{equation*}
\begin{split}
\sum_{j_1+j_2=2k+1}& \int_{\uline{H}_v}J^{N_{p,q}}[L^{j_1}\uline{L}^{j_2}\phi]\cdot \uline{L}+\int_{H_u}J^{N_{p,q}}[L^{j_1}\uline{L}^{j_2}\phi]\cdot L\\
\leq &\: C\sum_{j_1+j_2+2j_3+j_4\leq 2k+1}\int_{\uline{H}_{v_0}}J^{N_{2,q\beta \alpha}}[L^{j_1}\uline{L}^{j_2}Q^{j_3}\Phi^{j_4}\phi]\cdot \uline{L}+\int_{\mathcal{H}^+}J^{N_{2,q\beta \alpha}}[L^{j_1}\uline{L}^{j_2}Q^{j_3}\Phi^{j_4}\phi]\cdot L\\
&+C\sum_{2j_1+j_2\leq 2k+1}\int_{\uline{H}_{v_0}}J^{N_{2,q\beta \alpha}}[Q^{j_1}\Phi^{j_2+1}\phi]\cdot \uline{L}+\int_{\mathcal{H}^+}J^{N_{2,q\beta \alpha}}[Q^{j_1}\Phi^{j_2+1}\phi]\cdot L\\
&+C\sum_{2j_1+j_2\leq 2k}\int_{\uline{H}_{v_0}}J^{N_{2,q\beta \alpha}}[Q^{j_1+1}\Phi^{j_2}\phi]\cdot \uline{L}+\int_{\mathcal{H}^+}J^{N_{2,q\beta \alpha}}[Q^{j_1}\Phi^{j_2}\phi]\cdot L\\
&+C\sum_{j_1+2j_2\leq 2k+2}\int_{\mathcal{H}^+}v^{-2+\epsilon}(\Phi^{j_1}Q^{j_2}\phi)^2=:CE_{q\beta \alpha;2k+1,\epsilon},
\end{split}
\end{equation*}
and
\begin{equation*}
\begin{split}
\sum_{j_1+j_2=2k}& \int_{\uline{H}_v}J^{N_{p,q}}[L^{j_1}\uline{L}^{j_2}\phi]\cdot \uline{L}+\int_{H_u}J^{N_{p,q}}[L^{j_1}\uline{L}^{j_2}\phi]\cdot L\\
\leq &\: C\sum_{j_1+j_2+2j_3+j_4\leq 2k}\int_{\uline{H}_{v_0}}J^{N_{2,q\beta \alpha}}[L^{j_1}\uline{L}^{j_2}Q^{j_3}\Phi^{j_4}\phi]\cdot \uline{L}+\int_{\mathcal{H}^+}J^{N_{2,q\beta \alpha}}[L^{j_1}\uline{L}^{j_2}Q^{j_3}\Phi^{j_4}\phi]\cdot L\\
&+C\sum_{2j_1+j_2\leq 2k}\int_{\uline{H}_{v_0}}J^{N_{2,q\beta \alpha}}[Q^{j_1}\Phi^{j_2+1}\phi]\cdot \uline{L}+\int_{\mathcal{H}^+}J^{N_{2,q\beta \alpha}}[Q^{j_1}\Phi^{j_2+1}\phi]\cdot L\\
&+C\sum_{j_1+2j_2\leq 2k+1}\int_{\mathcal{H}^+}v^{-2+\epsilon}(\Phi^{j_1}Q^{j_2}\phi)^2=:CE_{q\beta \alpha;2k,\epsilon},
\end{split}
\end{equation*}
\end{corollary}

\section{Pointwise estimates}
\label{sec:poinwiseestimates}

\subsection{Uniform boundedness of $\phi$}
\label{sec:uniformboundednessphi}
We can use the higher-order energy estimates in the previous section to obtain a uniform pointwise bound on $\phi$. As in the previous section, $\phi$ always denotes a solution to (\ref{eq:waveqkerr}) arising from initial data prescribed in Proposition \ref{prop:wellposedness}. We will always indicate whether we are assuming axisymmetry of $\phi$ or the restriction $0\leq |a|<a_c$ for the rotation parameter $a$.
\begin{proposition}
\label{pointwiseboundkn}
Let $n\in \N_0$. Restrict to $0\leq |a|<a_c$ and take $0\leq p<2$, or restrict to axisymmetric $\phi$ and take $0\leq p\leq 2$. Let $\epsilon>0$ arbitrarily small and take $0<q<2$. There exists a constant $C=C(a,M,v_0,u_0,q,\epsilon)>0$ such that,
\begin{equation*}
\begin{split}
\sum_{j_1+j_2\leq n}(L^{j_1}\uline{L}^{j_2}\phi)^2(u,v,\theta_*,{\varphi_*})\leq \:& \sum_{|k|\leq 2}\sum_{j_1+j_2\leq n}\int_{S^2_{-\infty,v}}|\snabla^kL^{j_1}\uline{L}^{j_2}\phi|^2\,d\mu_{\slashed{g}}\\
&+C|u|^{1-p}\Bigg[E_{q;n+2,\epsilon}[\phi]+\sum_{\Gamma\in \{\Phi^2,T^2,Q\}}E_{q;n+1,\epsilon}[\Gamma \phi]\\
&+\sum_{\Gamma \in \{\Phi^3,T^2\Phi,Q\Phi\}}E_{q;n,\epsilon}[\Gamma \phi]\Bigg].
\end{split}
\end{equation*}
\end{proposition}
\begin{proof}
By the fundamental theorem of calculus applied to integrating along ingoing null geodesics, together with Cauchy--Schwarz, we can estimate
\begin{equation*}
\begin{split}
\phi^2(u,v,\theta_*,{\varphi_*})\leq \:& \phi^2(-\infty,v,\theta_*,{\varphi_*})+\left(\int_{-\infty}^u|\underline{L}\phi|(u',v,\theta_*,{\varphi_*})\,du'\right)^2,\\
\leq \:& \phi^2(-\infty,v,\theta_*,{\varphi_*})+|u|^{-1+p}\int_{-\infty}^u (\underline{L}\phi)^2(u',v,\theta_*,{\varphi_*})\,du'.
\end{split}
\end{equation*}
We can integrate over the spheres and apply Proposition \ref{energyestsmallakn} to arrive at
\begin{equation*}
\begin{split}
\int_{S^2_{u,v}}\phi^2\,d\mu_{\slashed{g}}\leq \:& \int_{S^2_{-\infty,v}}\phi^2\,d\mu_{\slashed{g}}+ C|u|^{-1+p}\int_{H_v\cap\{|u'|\geq |u|\}}|u|^p(\underline{L}\phi)^2\\
\leq \:& \int_{S^2_{-\infty,v}}\phi^2\,d\mu_{\slashed{g}}+C|u|^{1-p}E_q[\phi],
\end{split}
\end{equation*}
for $q>0$.

To arrive at a pointwise estimate, we apply the standard Sobolev inequality on the spheres $S^2_{u,v}$, together with Proposition \ref{prop:ellipticsphere1}:
\begin{equation*}
\begin{split}
||L^{j_1}\uline{L}^{j_2}\phi||^2_{L^{\infty}(S^2_{u,v})}\leq &\: C \int_{S^2_{u,v}}(L^{j_1}\uline{L}^{j_2}\phi)^2+|\snabla L^{j_1}\uline{L}^{j_2}\phi|^2+|\snabla^2 L^{j_1}\uline{L}^{j_2}\phi|^2\,d\mu_{\slashed{g}}\\
\leq &\: C\sum_{k=0}^1\sum_{\Gamma\in\{\textnormal{id},{\Phi},{\Phi}^2,T^2,Q\}}\int_{S^2_{u,v}}(\Gamma \underline{L}^k L^{j_1}\uline{L}^{j_2}\phi)^2+(\Gamma L^kL^{j_1}\uline{L}^{j_2} \phi)^2+(\underline{L}^2L^{j_1}\uline{L}^{j_2} \phi)^2\\
&+({L}^2 L^{j_1}\uline{L}^{j_2} \phi)^2+(L\underline{L} L^{j_1}\uline{L}^{j_2} \phi)^2+(b^{{\varphi_*}})^2[(Q\Phi L^{j_1}\uline{L}^{j_2} \phi)^2+(T^2\Phi L^{j_1}\uline{L}^{j_2} \phi)^2\\
&+(\Phi^3L^{j_1}\uline{L}^{j_2} \phi)^2]\,d\mu_{\slashed{g}}.
\end{split}
\end{equation*}
We now combine the results of Propositions \ref{prop:ellipticsphere1}, \ref{commenergyestimateskn1}, \ref{prop:commfinitevolregion} and Corollary \ref{cor:commeestimate}; in particular, we commute $\Gamma$ in the terms above to act directly on $\phi$, in order to arrive at the estimate in the proposition.
\end{proof}

We have now proved Theorem \ref{thm:pointwiseboundkn}.

\subsection{Extendibility of $\phi$ in $C^0$}
\label{sec:extendibilityphiC0}
We can use Proposition \ref{pointwiseboundkn} to show that $\phi$ can be extended as a continuous function beyond the Cauchy horizon $\mathcal{CH}^+$. As this extension is independent of the characteristic data, it is non-unique.

\begin{proposition}
\label{prop:C0extension}
Let the initial data for $\phi$ satisfy
\begin{equation*}
E_{q;2,\epsilon}[\phi]+\sum_{\Gamma\in \{\Phi^2,T^2,Q\}}E_{q;1,\epsilon}[\Gamma \phi]+\sum_{\Gamma \in \{\Phi^3,T^2\Phi,Q\Phi\}}E_{q}[\Gamma \phi]<\infty,
\end{equation*}
for some $q>1$ and $\epsilon>0$.

Let $x_{\mathcal{CH}^+}$ be a point on $\mathcal{CH}^+$. Then, for any $x\in {{\mathcal{D}_{u_0,v_0}}}$,
\begin{equation*}
\lim_{x\to x_{\mathcal{CH}^+}}\phi(x)
\end{equation*}
is well-defined, so $\phi$ can be extended as a $C^0$ function to the region beyond $\mathcal{CH}^+$.
\end{proposition}
\begin{proof}
Consider a sequence of points $x_k$ in ${{\mathcal{D}_{u_0,v_0}}}\setminus \mathcal{H}^+$, such that $\lim_{k\to \infty}x_k=x_{\mathcal{CH}^+}$. The sequence $\{x_k\}$ is in particular a Cauchy sequence. We will show that the sequence of points $(X_{m,n,l}\phi)(x_k)$ must also be a Cauchy sequence, from which it follows immediately that the sequence converges to a finite number as $k\to \infty$.

For simplicity, we will take $m-n=l=0$, but the steps of the proof can be repeated for the general case. Denote $x_k=(u_k,V_k,\theta_k,\varphi_k)$. Let $l>k$, then
\begin{equation*}
\begin{split}
|\phi(x_l)-\phi(x_k)|^2\leq \:& |\phi(u_l,V_k,(\theta_*)_k,(\varphi_*)_k)-\phi(u_k,V_k,(\theta_*)_k,(\varphi_*)_k)|\\
&+|\phi(u_k,V_l,(\theta_*)_k,(\varphi_*)_k)-\phi(u_k,V_k,(\theta_*)_k,(\varphi_*)_k)|^2\\
&+|\phi(u_k,V_k,(\theta_*)_l,(\varphi_*)_k)-\phi(u_k,V_k,(\theta_*)_k,(\varphi_*)_k)|^2\\
&+|\phi(u_k,V_k,(\theta_*)_k,(\varphi_*)_l)-\phi(u_k,V_k,(\theta_*)_k,(\varphi_*)_k)|^2.
\end{split}
\end{equation*}
By the fundamental theorem of calculus, a Sobolev inequality on $\s^2$ and Cauchy--Schwarz, we can estimate for $q>0$
\begin{equation*}
\begin{split}
&\left|\phi(u_l,V_k,(\theta_*)_k,(\varphi_*)_k)-\phi(u_k,V_k,(\theta_*)_k,(\varphi_*)_k)\right|^2\\
\leq \:& C\sum_{|s|\leq 2}\left|\int_{u_k}^{u_l}\int_{\s^2}u^{1+\epsilon}|\partial_u\snabla^s\phi|^2(u,V_k,(\theta_*)_k,(\varphi_*)_k)\,d\mu_{\s^2}du\right|\\
\leq \:& C \sum_{0\leq s_1+s_2\leq 2}\int_{\uline{H}_{v(V_k)}}J^{N_{2,q}}[\partial_{\theta_*}^{s_1}\Phi^{s_2}\phi]\cdot \underline{L}.
\end{split}
\end{equation*}

Similarly, we find that for $q>1$, in $(\widetilde{u},\widetilde{V},\widetilde{\theta}_*,\widetilde{\varphi}_*)$ coordinates:
\begin{equation*}
\begin{split}
|\phi&(\widetilde{u}_k,\widetilde{V}_l,(\widetilde{\theta}_*)_k,(\widetilde{\varphi}_*)_k)-\phi(\widetilde{u}_k,\widetilde{V}_k,(\widetilde{\theta}_*)_k,(\widetilde{\varphi}_*)_k)|^2\\
\leq \:& C\left|(-\widetilde{V}_l)^{q-1}-(-\widetilde{V}_k)^{q-1}\right|\sum_{|s|\leq 2}\left|\int_{\widetilde{V}_k}^{\widetilde{V}_l}\int_{\s^2}(-\widetilde{V})^{2-q}|\partial_{\widetilde{V}}\snabla^s\phi|^2(\widetilde{u},\widetilde{V},(\widetilde{\theta}_*)_k,(\widetilde{\varphi}_*)_k)\,d\mu_{\s^2}d\widetilde{V}\right|\\
\leq \:& C\left|(-\widetilde{V}_l)^{q-1}-(-\widetilde{V}_k)^{q-1}\right|\sum_{|s|\leq 2}\left|\int_{\widetilde{V}(\widetilde{V}_k)}^{\widetilde{V}(\widetilde{V}_l)}\int_{\s^2}\widetilde{V}^q|\partial_{\widetilde{V}}\snabla^s\phi|^2(\widetilde{u},\widetilde{V},(\widetilde{\theta}_*)_k,(\widetilde{\varphi}_*)_k)\,d\mu_{\s^2}d\widetilde{V}\right|\\
\leq \:& C \left|(-\widetilde{V}_l)^{q-1}-(-\widetilde{V}_k)^{q-1}\right|\sum_{0\leq s_1+s_2\leq 2}\int_{H_{u_k}}J^{N_{2,q}}[\partial_{\theta_*}^{s_1}\Phi^{s_2}\phi]\cdot L,
\end{split}
\end{equation*}
where we used that $(-\widetilde{V})^{2-q}\sim \widetilde{V}^{q-2}$ and $|\partial_{\widetilde{V}}\snabla^s\phi|^2d\widetilde{V}\sim  \widetilde{V}^2|\partial_{\widetilde{V}}\snabla^s\phi|^2d\widetilde{V}$.

Finally, we can estimate by Cauchy--Schwarz on $\s^2$,
\begin{equation*}
\begin{split}
|\phi&(u_k,v_l,(\theta_*)_l,(\varphi_*)_k)-\phi(u_k,v_k,(\theta_*)_k,(\varphi_*)_k)|^2+|\phi(u_k,v_l,(\theta_*)_k,(\varphi_*)_l)-\phi(u_k,v_k,(\theta_*)_k,(\varphi_*)_k)|^2\\
\leq \:& C\int_{\s^2}|\snabla\phi|^2(u_k,v_k,\theta_*,\varphi_*)\leq C \sum_{s_1+s_2=1}\int_{\uline{H}_{v(V_k)}}J^{N_{2,q}}[\partial_{\theta_*}^{s_1}\Phi^{s_2}\phi]\cdot \underline{L},
\end{split}
\end{equation*}
where we need $q>0$.

By the above estimates it follows that $\phi(x_k)$ must also be a Cauchy sequence if the energies on the right-hand sides are finite.

Finally, as in Proposition \ref{pointwiseboundkn}, we can estimate the energies on the right-hand sides of the above estimates by the initial energy $E_{\Gamma;q}[\phi]$.
\end{proof}

We have now proved Theorem \ref{thm:C0extension}.

\subsection{Decay of $L\phi$}
\label{sec:decayLphi}
Consider the function $\mathcal{\phi_{\mathcal{H}^+}}:\mathcal{M}\cap {{\mathcal{D}_{u_0,v_0}}}\to \R$ defined by
\begin{equation*}
\phi_{\mathcal{H}^+}(u,v,\theta_*,\varphi_*):=\phi(-\infty,\theta_*,\varphi_*).
\end{equation*}
In particular, $\underline{L}\phi_{\mathcal{H}^+}=0$.

We consider $\psi:=\phi-\phi_{\mathcal{H}^+}$. We can improve the pointwise decay in $\psi$ with respect to $|u|$ and use the wave equation (\ref{eq:waveqkerr}) to obtain decay of $|L\phi|$ in $v$. Moreover, we can obtain boundedness of
\begin{equation*}
\int_{H_u}v^2(L\phi)^2+\int_{\underline{H}_v}v^2\Omega^2|\snabla \phi|^2.
\end{equation*}
\begin{proposition}
\label{improvedpointwisedecaykn}
Denote
\begin{equation*}
D=||\partial_U\phi||^2_{L^{\infty}\left(\uline{H}_{v_0}\right)}+||\snabla \phi||^2_{L^{\infty}\left(\uline{H}_{v_0}\right)}.
\end{equation*}
Let $0<p<2$, $0<q<2$ and $0\leq s \leq 1$ if $|a|<a_c$. Let $p=2$, $0<q\leq 2$ and $0 \leq s \leq 1$ if $\phi$ is axisymmetric. Then, for every $\epsilon>0$, there exists a constant $C=C(M,u_0,v_0,p,q,s,\epsilon)>0$, such that for all $H_u$ and $\uline{H}_v$ in ${{\mathcal{D}_{u_0,v_0}}}$,
\begin{equation*}
\begin{split}
\int_{H_u}& J^{N_{p,q}}[\psi]\cdot L+\int_{\uline{H}_v\cap\{u'\leq u\}} J^{N_{p,q}}[\psi]\cdot \underline{L}\\
\leq \:& C|u|^{-s} \left[\int_{\mathcal{H}^+\cap\{v\geq v_1\}} v^{s+\epsilon}\left[(L\phi)^2+|\snabla\phi|^2+|\snabla^2\phi|^2\right]+ D\right]\\
:=\: &C|u|^{-s}\widetilde{E}_{\epsilon,s}[\phi].
\end{split}
\end{equation*}
\end{proposition}
\begin{proof}
We have that
\begin{equation*}
2\Omega^2\square_g\psi=-2\Omega^2\square_g\phi|_{\mathcal{H}^+}=-\Omega \tr \underline{\chi} L\phi_{\mathcal{H}^+}+2\snabla\Omega^2\cdot \snabla\phi_{\mathcal{H}^+}+2\Omega^2\slashed{\Delta}\phi_{\mathcal{H}^+}.
\end{equation*}
Consequently, we can estimate,
\begin{equation*}
2\Omega^2|\square_g\psi|\leq C(v+|u|)^{-2}\left(|L\phi_{\mathcal{H}^+}|+|\snabla\phi_{\mathcal{H}^+}|+|\snabla^2\phi_{\mathcal{H}^+}|\right).
\end{equation*}
By applying Stokes' theorem in ${{\mathcal{D}_{u_0,v_0}}}$ we obtain the following error term:
\begin{equation*}
\left|\int_{{{\mathcal{D}_{u_0,v_0}}}}\mathcal{E}^{N_{p,q}}[\psi]\right|\leq \int \int \int_{S^2_{u,v}}(v+|u|)^{-2}(|u|^p|\underline{L}\psi|+v^q|L\psi|)\left(|L\phi_{\mathcal{H}^+}|+|\snabla\phi_{\mathcal{H}^+}|+|\snabla^2\phi_{\mathcal{H}^+}|\right).
\end{equation*}
By Cauchy--Schwarz, we can estimate for $\eta>0$
\begin{equation*}
\begin{split}
&(v+|u|)^{-2}(|u|^p|\underline{L}\psi|+v^q|L\psi|)\left(|L\phi_{\mathcal{H}^+}|+|\snabla\phi_{\mathcal{H}^+}|+|\snabla^2\phi_{\mathcal{H}^+}|\right)\lesssim v^{-1-\eta}|u|^p(\underline{L}\psi)^2\\
&+|u|^{-1-\eta'}v^q(L\psi)^2+(v+|u|)^{-4}(|u|^p{v}^{1+\eta}+v^q|u|^{1+\eta})\left[(L\phi_{\mathcal{H}^+})^2+|\snabla\phi_{\mathcal{H}^+}|^2+|\snabla^2\phi_{\mathcal{H}^+}|^2\right].
\end{split}
\end{equation*}
We further estimate for $0\leq s\leq 1$,
\begin{equation*}
\begin{split}
&(v+|u|)^{-4}(|u|^p{v}^{1+\eta}+v^q|u|^{1+\eta})\left[(L\phi_{\mathcal{H}^+})^2+|\snabla\phi_{\mathcal{H}^+}|^2+|\snabla^2\phi_{\mathcal{H}^+}|^2\right]\\
&\lesssim |u|^{-1-s} (v^{p-2+s+\eta}+v^{q-2+s+\eta})\left[(L\phi_{\mathcal{H}^+})^2+|\snabla\phi_{\mathcal{H}^+}|^2+|\snabla^2\phi_{\mathcal{H}^+}|^2\right].
\end{split}
\end{equation*}
Hence,
\begin{equation*}
\begin{split}
\int_{{{\mathcal{D}_{u_0,v_0}}}}&|u|^{-1-s} (v^{p-2+s+\eta}+v^{q-2+s+\eta})\left[(L\phi_{\mathcal{H}^+})^2+|\snabla\phi_{\mathcal{H}^+}|^2+|\snabla^2\phi_{\mathcal{H}^+}|^2\right]\\
\leq \:& |u|^{-s}\int_{\mathcal{H}^+\cap \{v'\geq v_0\}}(v^{p-2+s+\eta}+v^{q-2+s+\eta})\left((L\phi)^2+|\snabla\phi|^2+|\snabla^2\phi|^2\right),
\end{split}
\end{equation*}
where we used that
\begin{equation*}
(L\phi_{\mathcal{H}^+})^2+|\snabla\phi_{\mathcal{H}^+}|^2+|\snabla^2\phi_{\mathcal{H}^+}|^2 \sim (L\phi)|_{\mathcal{H}^+}^2+|\snabla\phi|^2|_{\mathcal{H}^+}+|\snabla^2\phi|^2|_{\mathcal{H}^+}.
\end{equation*}
The remaining terms in $\mathcal{E}^{N_{p,q}}[\psi]$ and the terms in $K^{N_{p,q}}$ can be estimated as in Proposition \ref{prop:mainenergyestimate} and the propositions in Section~ \ref{sec:energyestimatesslowlyrot}.
\end{proof}

\begin{proposition}
\label{improvedpointwisedecayknc2}
Let $k\in \N_0$. Denote
\begin{align*}
D_{2k}:=&\:\sum_{j_1+j_2+2j_3+j_4\leq 2k}||\partial_UL^{j_1}\uline{L}^{j_2}Q^{j_3}\Phi^{j_4}\phi||^2_{L^{\infty}\left(\uline{H}_{v_0}\right)}+||\snabla L^{j_1}\uline{L}^{j_2}Q^{j_3}\Phi^{j_4}\phi||^2_{L^{\infty}\left(\uline{H}_{v_0}\right)}\\
&+\sum_{j_1+2j_2\leq 2k}||\partial_U\Phi^{j_1+1}Q^{j_2}\phi||^2_{L^{\infty}\left(\uline{H}_{v_0}\right)}+||\snabla \Phi^{j_1+1}Q^{j_2}\phi||^2_{L^{\infty}\left(\uline{H}_{v_0}\right)},\\
D_{2k+1}:=&\:\sum_{j_1+j_2+2j_3+j_4\leq 2k+1}||\partial_UL^{j_1}\uline{L}^{j_2}Q^{j_3}\Phi^{j_4}\phi||^2_{L^{\infty}\left(\uline{H}_{v_0}\right)}+||\snabla L^{j_1}\uline{L}^{j_2}Q^{j_3}\Phi^{j_4}\phi||^2_{L^{\infty}\left(\uline{H}_{v_0}\right)}\\
&+\sum_{j_1+2j_2\leq 2k+1}||\partial_U\Phi^{j_1+1}Q^{j_2}\phi||^2_{L^{\infty}\left(\uline{H}_{v_0}\right)}+||\snabla \Phi^{j_1+1}Q^{j_2}\phi||^2_{L^{\infty}\left(\uline{H}_{v_0}\right)}\\
&+\sum_{j_1+2j_2\leq 2k}||\partial_U\Phi^{j_1}Q^{j_2+1}\phi||^2_{L^{\infty}\left(\uline{H}_{v_0}\right)}+||\snabla \Phi^{j_1}Q^{j_2+1}\phi||^2_{L^{\infty}\left(\uline{H}_{v_0}\right)}.
\end{align*}
Let $0<p<2$, $0<q\leq 2$ and $0\leq s\leq 1$ if $|a|<a_c$. Let $p=2$, $0<q\leq 2$ and $0\leq s\leq 1$ if $\phi$ is axisymmetric.

Let $n\in \N_0$. Then, for every $\epsilon>0$, there exists a constant $C=C(M,a,n,u_0,v_0,p,q,s,\epsilon)>0$, such that
\begin{equation*}
\begin{split}
\sum_{j_1+j_2=n}\int_{H_u}& J^{N_{p,q}}[L^{j_1}\uline{L}^{j_2}\psi]\cdot L+\int_{\uline{H}_v\cap\{u'\leq u\}} J^{N_{p,q}}[L^{j_1}\uline{L}^{j_2}\psi]\cdot \underline{L}\\
\leq \:& C|u|^{-s} \sum_{j_1+2j_2+j_3\leq n}\Bigg[\int_{\mathcal{H}^+\cap\{v\geq v_1\}} v^{s+\epsilon}\Big[(L^{j_1+1}Q^{j_2}\Phi^{j_3}\phi)^2+|\snabla L^{j_1}Q^{j_2}\Phi^{j_3}\phi|^2\\
&+|\snabla^2 L^{j_1}Q^{j_2}\Phi^{j_3}\phi|^2\Big]+ {D}_n\Bigg]\\
&=:C|u|^{-s}\widetilde{E}_{s,\epsilon;n}[\phi].
\end{split}
\end{equation*}
\end{proposition}
\begin{proof}
We have that
\begin{equation*}
\begin{split}
2\Omega^2\square_g(L^{j_1}\uline{L}^{j_2}\psi)=&\:[L^{j_1}\uline{L}^{j_2},\underline{L}]L\psi+\underline{L}([L^{j_1}\uline{L}^{j_2},L]\psi)+[L^{j_1}\uline{L}^{j_2},L]\underline{L}\psi+L([L^{j_1}\uline{L}^{j_2},\underline{L}]\psi)\\
&\sum_{l=1}^{j_1}\sum_{k=1}^{j_2} \binom{j_1}{l}\binom{j_2}{k} L^l\uline{L}^k(\Omega\tr \underline{\chi})L^{j_1-l}\uline{L}^{j_2-k}L\psi+\Omega \tr {\underline{\chi}}[L^{j_1}\uline{L}^{j_2},L]\psi\\
&+ \sum_{l=1}^{j_1}\sum_{k=1}^{j_2} \binom{j_1}{l}\binom{j_2}{k} L^l\uline{L}^k(\Omega\tr {\chi})L^{j_1-l}\uline{L}^{n-k}\underline{L}\psi+\Omega \tr {\chi}[L^{j_1}\uline{L}^{j_2},\underline{L}]\psi\\
&-2\sum_{l=1}^{j_1}\sum_{k=1}^{j_2}  \binom{j_1}{l}\binom{j_2}{k} L^l\uline{L}^k(\slashed{g}^{AB}\partial_A\Omega^2)L^{j_1-l}\uline{L}^{j_2-k}\partial_B\psi-2\slashed{g}^{AB}\partial_A\Omega^2[L^{j_1}\uline{L}^{j_2},\partial_B]\psi\\
&-2\sum_{l=1}^{j_1}\sum_{k=1}^{j_2}  \binom{j_1}{l}\binom{j_2}{k} L^l\uline{L}^k(\Omega^2)L^{j_1-l}\uline{L}^{j_2-k}\slashed{\Delta}\psi-2\Omega^2[L^{j_1}\uline{L}^{j_2},\slashed{\Delta}]\psi\\
&-L^{j_1}\uline{L}^{j_2}(2\Omega^2\square_g\psi).
\end{split}
\end{equation*}
We can repeat the proof of Proposition \ref{commutedenergyestkn}, but we have to additionally estimate the contribution of the final term in the above expression for $2\Omega^2\square_g(L^{j_1}\uline{L}^{j_2}\psi)$. We can estimate
\begin{equation*}
|L^{j_1}\uline{L}^{j_2}(2\Omega^2\square_g\psi)|\leq \sum_{j_1\leq n}|L^{j_1+1}\phi_{\mathcal{H}^+}|+|\snabla L^{j_1}\phi_{\mathcal{H}^+}|+|\snabla^2 L^{j_1}\phi_{\mathcal{H}^+}|.
\end{equation*}
We can therefore deal with the corresponding term in $\mathcal{E}^{N_{p,q}}[L^{j_1}\uline{L}^{j_2}\psi]$ in the same way as in the proof of Proposition \ref{improvedpointwisedecaykn}. 
\end{proof}

\begin{proposition}
\label{cor:improvedpointbound}
Let $s\leq 1$ and $0\leq p<2$ for $0\leq |a|<a_c$. For axisymmetric $\phi$ we let $0\leq |a|\leq M$ and we can also take $p=2$. Then there exists a constant $C=C(M,v_0,u_0,p,s)>0$ such that,
\begin{equation*}
\begin{split}
\sum_{j_1+j_2=n}&\int_{S^2_{u,v}}(L^{j_1}\uline{L}^{j_2}\psi)^2+|\snabla L^{j_1}\uline{L}^{j_2}\psi|^2+|\snabla^2L^{j_1}\uline{L}^{j_2}\psi|^2\,d\mu_{\slashed{g}}\\
\leq \:& C|u|^{1-s-p}\Bigg[\widetilde{E}_{s,\epsilon; n+2}[\phi]+\sum_{\Gamma\in \{\Phi^2,T^2,Q\}}\widetilde{E}_{s,\epsilon;n+1}[\Gamma \phi]+\sum_{\Gamma \in \{\Phi^3,T^2\Phi,Q\Phi\}}\widetilde{E}_{s,\epsilon;n}[\Gamma \phi]\Bigg]
\end{split}
\end{equation*}
for $s\leq 1$ and $\epsilon>0$ arbitrarily small.
\end{proposition}
\begin{proof}
We obtain estimates for angular derivatives along $S^2_{u,v}$ from the higher-order energy estimates in Proposition \ref{improvedpointwisedecayknc2} in the same way as in Proposition \ref{pointwiseboundkn}.
\end{proof}

We can now obtain decay for $L\phi$.
\begin{proposition}
\label{prop:decayLphi}
Let $0\leq |a|<a_c$, $1<p<2$ and $0\leq s\leq 1$. For $\delta,\epsilon,q >0$ arbitrarily small, there exists a constant $C=C(a,M,v_0,u_0,p,q,\epsilon,\delta,s)>0$ such that,
\begin{equation}
\label{eq:improvedestLphi}
\begin{split}
\int_{S^2_{u,v}}(L\phi)^2(u,v,\theta_*,{\varphi_*})\,d\mu_{\slashed{g}}\leq \:& \int_{S^2_{-\infty,v}}(L\phi)^2\,d\mu_{\slashed{g}}\\
&+C(v+|u|)^{-4}|u|^{1-p}E_{q}[\phi]\\
&+Cv^{-4+(1-s)+(2-p)+\delta}\Bigg[\widetilde{E}_{s,\epsilon;2}[\phi]+\sum_{\Gamma\in \{\Phi^2,T^2,Q\}}\widetilde{E}_{s,\epsilon;1}[\Gamma \phi]\\
&+\sum_{\Gamma \in \{\Phi^3,T^2\Phi,Q\Phi\}}\widetilde{E}_{s,\epsilon}[\Gamma \phi]\Bigg]\\
&+C(v+|u|)^{-2}\int_{S^2_{-\infty,v}}|\snabla\phi_{\mathcal{H}^+}|^2+|\snabla^2\phi_{\mathcal{H}^+}|^2d\mu_{\slashed{g}}.
\end{split}
\end{equation}

Moreover, for axisymmetric $\phi$, we have a stronger estimate for all $0\leq |a|\leq M$,
\begin{equation*}
\begin{split}
\int_{S^2_{u,v}}(L\phi)^2(u,v,\theta_*,{\varphi_*})\,d\mu_{\slashed{g}}\leq \:& \sum_{|k|\leq 2}\int_{S^2_{-\infty,v}}(L\phi)^2\,d\mu_{\slashed{g}}+C(v+|u|)^{-4}|u|^{1-p}E_{q}[\phi]\\
&+Cv^{-4}\log\left(\frac{v+|u|}{|u|}\right)\Bigg[\widetilde{E}_{1,\epsilon;2}[\phi]+\sum_{\Gamma\in \{\Phi^2,T^2,Q\}}\widetilde{E}_{1,\epsilon;1}[\Gamma \phi]\\
&+\sum_{\Gamma \in \{\Phi^3,T^2\Phi,Q\Phi\}}\widetilde{E}_{1,\epsilon}[\Gamma \phi]\Bigg]\\
&+C(v+|u|)^{-2}\int_{S^2_{-\infty,v}}|\snabla\phi_{\mathcal{H}^+}|^2+|\snabla^2\phi_{\mathcal{H}^+}|^2d\mu_{\slashed{g}}.
\end{split}
\end{equation*}
\end{proposition}
\begin{proof}
We can write the wave equation as a transport equation for $(\det\slashed{g})^{\frac{1}{4}}L\phi$,
\begin{equation*}
\underline{L}((\det\slashed{g})^{\frac{1}{4}}L\phi)=(\det\slashed{g})^{\frac{1}{4}}\left[-\Omega\tr {\chi}\underline{L}\phi+2\Omega^2\zeta^{\varphi_*}\partial_{\varphi_*}\phi+\snabla\Omega^2\cdot\snabla \phi+\Omega^2\slashed{\Delta}\phi\right],
\end{equation*}
see Appendix \ref{app:treq}. In particular, we can estimate
\begin{equation*}
|\underline{L}((\det\slashed{g})^{\frac{1}{4}}L\phi)|(v,u,\theta_*,\varphi_*)\leq C(v+|u|)^{-2}\left(|\underline{L}\phi|+|\snabla\phi|+|\snabla^2\phi|\right).
\end{equation*}
We can split
\begin{equation*}
\int_{S_{u,v}^2}(L\phi)^2\,d\mu_{\slashed{g}}\leq \int_{S_{-\infty,v}^2}(L\phi)^2\,d\mu_{\slashed{g}}+\int_{S_{u,v}^2}(L\psi)^2\,d\mu_{\slashed{g}}.
\end{equation*}
We now integrate along ingoing null geodesics for fixed $\theta_*$ and $\varphi_*$, and subsequently integrate in $\theta_*$ and $\varphi_*$ to obtain
\begin{equation*}
\begin{split}
\int_{S_{u,v}^2}(L\phi)^2\,d\mu_{\slashed{g}}\leq \:& \int_{S^2_{u,v}}\left(\int_{-\infty}^u|\underline{L}((\det\slashed{g})^{\frac{1}{4}}L\phi)|\,du'\right)^2\,d\mu_{\slashed{g}}\\
\leq \:& C\int_{S^2_{u,v}}\left(\int_{-\infty}^u (v+|u'|)^{-2}\left(|\underline{L}\phi|+|\snabla\phi|+|\snabla^2\phi|\right)\,du'\right)^2\,d\mu_{\slashed{g}}\\
\leq \:& C(v+|u|)^{-4}|u|^{-\eta}\int_{H_v\cap\{|u'|\leq |u|\}}|u'|^{1+\eta}(\underline{L}\phi)^2\\
&+C\int_{-\infty}^u |u|^{-s}(v+|u'|)^{-2}\,du'\\
&\quad\cdot \int_{-\infty}^u|u'|^s(v+|u'|)^{-2}\left(\int_{S^2_{u',v}}|\snabla\psi|^2+|\snabla^2\psi|^2d\mu_{\slashed{g}}\right)\,du'\\
&+C(v+|u|)^{-2}\int_{S^2_{-\infty,v}}|\snabla\phi_{\mathcal{H}^+}|^2+|\snabla^2\phi_{\mathcal{H}^+}|^2d\mu_{\slashed{g}}.
\end{split}
\end{equation*}
We now apply the results of Propositions \ref{improvedpointwisedecaykn} and \ref{cor:improvedpointbound} to arrive at
\begin{equation*}
\begin{split}
\int_{S_{u,v}^2}(L\psi)^2\,d\mu_{\slashed{g}}\leq \:& C(v+|u|)^{-4}|u|^{-\eta}E_{q}[\phi]\\
&+C\Bigg(\widetilde{E}_{s,\epsilon;2}[\phi]+\sum_{\Gamma\in \{\Phi^2,T^2,Q\}}\widetilde{E}_{s,\epsilon;1}[\Gamma \phi]+\sum_{\Gamma \in \{\Phi^3,T^2\Phi,Q\Phi\}}\widetilde{E}_{s,\epsilon}[\Gamma \phi]\Bigg)\\
&\cdot \int_{-\infty}^u|u|^{-s}(v+|u'|)^{-2}\,du'\int_{-\infty}^u|u|^{1-p}(v+|u'|)^{-2}\,du'\\
&+C(v+|u|)^{-2}\int_{S^2_{-\infty,v}}|\snabla\phi_{\mathcal{H}^+}|^2+|\snabla^2\phi_{\mathcal{H}^+}|^2d\mu_{\slashed{g}}.
\end{split}
\end{equation*}
where $0<s<1$ and $\epsilon'>0$ suitably small. Moreover, if $\psi$ is axisymmetric we can take $s=1$.

Now we use that
\begin{align*}
\int_{-\infty}^u|u|^{-1}(v+|u'|)^{-2}\,du'\leq \:& C(v+|u|)^{-2}\log\left(\frac{v+|u|}{|u|}\right),\\
\int_{-\infty}^u|u'|^{-s}(v+|u'|)^{-2}\,du'\leq \:& C(v+|u|)^{-2+(1-s+\eta)}|u|^{-\eta},
\end{align*}
for $0<s<1$ and $\eta>0$ arbitrarily small, to arrive at the statement in the proposition.
\end{proof}

\begin{corollary}
\label{cor:improvedestoutgoingenergy}
Let $|a|<a_c$ and $0<s\leq 1$. Then there exist $1<p<2$, $\epsilon>0$ and a constant $C=C(a,M,v_0,u_0,p,\epsilon,\delta,s)>0$ such that,
\begin{equation*}
\begin{split}
\int_{H_u}&v^2(L\phi)^2+u^2\Omega^2 |\snabla\phi|^2+\int_{\underline{H}_v}v^2\Omega^2 |\snabla\phi|^2\\
\leq&\: C|u|^{-p}E_{q}[\phi]+C\Bigg[\widetilde{E}_{s,\epsilon;2}[\phi]+\sum_{\Gamma\in \{\Phi^2,T^2,Q\}}\widetilde{E}_{s,\epsilon;1}[\Gamma \phi]+\sum_{\Gamma \in \{\Phi^3,T^2\Phi,Q\Phi\}}\widetilde{E}_{s,\epsilon}[\Gamma \phi]\Bigg]\\
&+C\int_{\mathcal{H}^+\cap \{v\geq v_0\}}v^2(L\phi)^2+|\snabla\phi|^2+|\snabla^2\phi|^2,
\end{split}
\end{equation*}
with $q>0$ arbitrarily small. 
\end{corollary}
\begin{proof}
To estimate the terms involving $\snabla \phi$, we use that
\begin{equation*}
\int_{S^2_{u,v}}|\snabla \phi|^2\,d\mu_{\slashed{g}}\leq C\int_{S^2_{-\infty,v}}|\snabla \phi|^2\,d\mu_{\slashed{g}}+C\int_{S^2_{u,v}}|\snabla \psi|^2\,d\mu_{\slashed{g}}
\end{equation*}
and apply Proposition \ref{cor:improvedpointbound}.

To estimate the term involving $L\phi$, we multiply (\ref{eq:improvedestLphi}) by $v^2$ and integrate from $v=v_0$ to $v=\infty$.
\end{proof}
\begin{remark}
Recall that Corollary \ref{energyestsmallakn} gives a bound on
\begin{equation*}
\int_{H_u}v^q(L\phi)^2,
\end{equation*}
with the restriction $q<2$. Corollary \ref{cor:improvedestoutgoingenergy} provides moreover an estimate for $q=2$, at the expense of losing derivatives on the right-hand side.
\end{remark}

We have now proved Theorem \ref{thm:improvedestoutgoingenergy}.
\begin{proposition}
\label{prop:c0alphaextendibility}
Let $0\leq |a|<a_c$ or assume $\phi$ is axisymmetric with $0\leq |a|\leq M$. For $\delta,\epsilon,q>0$ arbitrarily small and $0\leq s\leq 1$, there exists a constant \\$C=C(a,M,v_0,u_0,\epsilon,\delta,s)>0$ such that,
\begin{equation*}
\begin{split}
v^{4+(s-1)-\delta}||L\phi||^2_{L^{\infty}(S^2_{u,v})}\leq &\:C\int_{S^2_{-\infty,v}}\sum_{j_1+j_2\leq 2}(LL^{j_1}\uline{L}^{j_2}\phi)^2+\sum_{j_1+j_2\leq 1} \sum_{\Gamma \in \{\Phi,\Phi^2,T^2,Q\}}(LL^{j_1}\uline{L}^{j_2} \Gamma \phi)^2\,d\mu_{\slashed{g}}\\
&+Cv^{2+(s-1)-\delta}\int_{S^2_{-\infty,v}}\sum_{j_1+j_2\leq 2}\left(|\snabla L^{j_1}\uline{L}^{j_2}\phi|^2+|\snabla^2 L^{j_1}\uline{L}^{j_2}\phi|^2\right)\,d\mu_{\slashed{g}}\\
&+{E}_{q;6,\epsilon}[\phi]+{E}_{q;6,\epsilon}[\Phi \phi]+\widetilde{E}_{s,\epsilon;6}[\phi]+\widetilde{E}_{s,\epsilon;6}[\Phi \phi].
\end{split}
\end{equation*}
\end{proposition}
\begin{proof}
From Appendix A it follows that
\begin{equation*}
\underline{L}((\det\slashed{g})^{\frac{1}{4}}L f)=(\det\slashed{g})^{\frac{1}{4}}\left[-\Omega\tr {\chi}\underline{L}\phi+2\Omega^2\zeta^{\varphi_*}\partial_{\varphi_*}\phi+\snabla\Omega^2\cdot\snabla \phi+\Omega^2\slashed{\Delta}\phi-\Omega^2 \square_g(f)\right],
\end{equation*}
for any suitably regular function $f: \mathcal{M}\cap {\mathcal{D}_{u_0,v_0}}\to \R$. In particular, we can estimate
\begin{equation*}
\begin{split}
\sum_{j_1+j_2=n}&|\underline{L}((\det\slashed{g})^{\frac{1}{4}}LL^{j_1}L^{j_2}\phi)|(v,u,\theta_*,\varphi_*)\\
\leq &\: C\sum_{j_1+j_2=n}(v+|u|)^{-2}\left(|\underline{L}L^{j_1}L^{j_2}\phi|+|\snabla L^{j_1}L^{j_2}\phi|+|\snabla^2 L^{j_1}L^{j_2}\phi|\right)+C|\Omega^2\square_g(L^{j_1}L^{j_2}\phi)|.
\end{split}
\end{equation*}
Using Lemma \ref{lm:waveoperatorcomm}, we therefore obtain
\begin{equation*}
\begin{split}
\sum_{j_1+j_2=n}&|\underline{L}((\det\slashed{g})^{\frac{1}{4}}LL^{j_1}\uline{L}^{j_2}\phi)|\\
\leq &\: C\sum_{j_1+j_2\leq n}(v+|u|)^{-2}\left(|\underline{L}L^{j_1}\uline{L}^{j_2}\phi|+|\snabla L^{j_1}\uline{L}^{j_2}\phi|+|\snabla^2 L^{j_1}\uline{L}^{j_2}\phi|\right)\\
&+C\sum_{j_1+j_2+j_3\leq n-2}|L^{j_1}\uline{L}^{j_2}\Phi^{j_3+1}\phi|.
\end{split}
\end{equation*}
We can therefore repeat the proof of Proposition \ref{prop:decayLphi}, using appropriate higher-order energy estimates, to obtain
\begin{equation*}
\begin{split}
\sum_{j_1+j_2= n}&\int_{S^2_{u,v}}(LL^{j_1}\uline{L}^{j_2}\phi)^2(u,v,\theta_*,{\varphi_*})\,d\mu_{\slashed{g}}\\
\leq &\: \sum_{j_1+j_2=n}\int_{S^2_{-\infty,v}}(LL^{j_1}\uline{L}^{j_2}\phi)^2\,d\mu_{\slashed{g}}+C(v+|u|)^{-4}|u|^{1-p}E_{q;n,\epsilon}[\phi]\\
&+Cv^{-4+(1-s)+(2-p)+\delta}\Bigg[\widetilde{E}_{s,\epsilon;n+2}[\phi]+\sum_{\Gamma\in \{\Phi^2,T^2,Q\}}\widetilde{E}_{s,\epsilon;n+1}[\Gamma \phi]\\
&+\sum_{\Gamma \in \{\Phi^3,T^2\Phi,Q\Phi\}}\widetilde{E}_{s,\epsilon;n}[\Gamma \phi]\Bigg]+\sum_{j_1+j_2+j_3\leq n-2}\widetilde{E}_{s,\epsilon;j_1+j_2+2}[\Phi^{j_3+1}\phi]\\
&+C(v+|u|)^{-2}\sum_{j_1+j_2= n}\int_{S^2_{-\infty,v}}|\snabla L^{j_1}\uline{L}^{j_2}\phi_{\mathcal{H}^+}|^2+|\snabla^2 L^{j_1}\uline{L}^{j_2}\phi_{\mathcal{H}^+}|^2d\mu_{\slashed{g}}
\end{split}
\end{equation*}
with $q>0$ arbitrarily small.

We now apply (\ref{est:ellipticspherev1}) together with a standard Sobolev inequality on $S^2_{u,v}$ to obtain the following $L^{\infty}$ estimate:
\begin{equation*}
\begin{split}
||L\phi||^2_{L^{\infty}(S^2_{u,v})}\leq &\:C \sum_{j_1+j_2\leq 2}\int_{S^2_{u,v}} (L^{j_1}\uline{L}^{j_2} L\phi)^2\,d\mu_{\slashed{g}}\\
& +C \sum_{\Gamma \in \{\textnormal{id},\Phi,\Phi^2,T^2,Q\}}\int_{S^2_{u,v}}(\Gamma L\phi)^2+(\Gamma \uline{L} L\phi)^2+(\Gamma L^2 \phi)^2\,d\mu_{\slashed{g}}\\
&+C(v+|u|)^{-2} \sum_{\Gamma \in\{Q,T^2,\Phi^2\}} \int_{S^2_{u,v}}(\Gamma \Phi L  \phi)^2\,d\mu_{\slashed{g}}.
\end{split}
\end{equation*}
We can bring the operator $L$ in front of $\phi$ in the above inequality to the front at the expense of including commutators with $L$:
\begin{equation*}
\begin{split}
||L\phi||^2_{L^{\infty}(S^2_{u,v})}\leq &\:C \sum_{j_1+j_2\leq 2}\int_{S^2_{u,v}} (L L^{j_1}\uline{L}^{j_2}\phi)^2\,d\mu_{\slashed{g}}\\
& +C \sum_{\Gamma \in \{\textnormal{id},\Phi,\Phi^2,T^2,Q\}}\int_{S^2_{u,v}}(L \Gamma \phi)^2+(L \Gamma \uline{L} \phi)^2+(L L \Gamma\phi)^2\,d\mu_{\slashed{g}}\\
&+C(v+|u|)^{-2} \sum_{\Gamma \in\{Q,T^2,\Phi^2\}} \int_{S^2_{u,v}}(L\Gamma \Phi  \phi)^2\,d\mu_{\slashed{g}}+J,
\end{split}
\end{equation*}
where
\begin{equation*}
J:=\int_{S^2_{u,v}}J_1+J_2+J_3\,d\mu_{\slashed{g}},
\end{equation*}
with
\begin{align*}
J_1:=&\:\sum_{\Gamma \in \{T^2,Q\}} ([\Gamma,L]\phi)^2+([\Gamma,L]\uline{L}\phi)^2+([\Gamma,L]L\phi)^2+(L[\Gamma,\uline{L}]\phi)^2+(\uline{L}[\Gamma,L]\phi)^2+\Gamma([L,\uline{L}]\phi)^2,\\
J_2:=&\:\sum_{j_1+j_2\leq 2} ([L^{j_1}\uline{L}^{j_2},L]\phi)^2,\\
J_3:=&\:(v+|u|)^{-2}\sum_{\Gamma \in \{T^2,Q\}}([\Gamma,L]\Phi\phi)^2.
\end{align*}
We apply the estimates for the above commutators that are derived in the proofs of Lemma \ref{lm:waveoperatorcomm} and Proposition \ref{commutedenergyestkn} to estimate
\begin{align*}
J_1\leq &\:C(v+|u|)^{-4}\sum_{j_1+j_2\leq 1} |L^{j_1}\uline{L}^{j_2} \Phi \phi|^2,\\
J_2\leq &\:C(v+|u|)^{-4}\left[\sum_{j_1+j_2 \leq 2} |\snabla^2 L^{j_1}\uline{L}^{j_2}\phi|^2+\sum_{j_1+j_2\leq 3} |\snabla L^{j_1}L^{j_2}\phi|^2+\sum_{j_1+j_2\leq 4} |L^{j_1}\uline{L}^{j_2}\phi|^2\right],\\
J_3\leq &\:C(v+|u|)^{-6}\left[\sum_{j_1+j_2 \leq 1} |\snabla^2 L^{j_1}\uline{L}^{j_2} \Phi \phi|^2+\sum_{j_1+j_2\leq 2} |\snabla L^{j_1}L^{j_2}\Phi \phi|^2+\sum_{j_1+j_2\leq 3} |L^{j_1}\uline{L}^{j_2} \Phi\phi|^2\right].
\end{align*}

We use the estimates above to obtain
\begin{equation*}
||L\phi||^2_{L^{\infty}(S^2_{u,v})}\leq C(I_1+I_2),
\end{equation*}
where
\begin{align*}
I_1:=&\:C \sum_{j_1+j_2\leq 2}\int_{S^2_{u,v}} (L L^{j_1}\uline{L}^{j_2}\phi)^2\,d\mu_{\slashed{g}}\\
& +C \sum_{\Gamma \in \{\textnormal{id},\Phi,\Phi^2,T^2,Q\}}\int_{S^2_{u,v}}(L \Gamma \phi)^2+(L  \uline{L} \Gamma\phi)^2+(L L \Gamma\phi)^2\,d\mu_{\slashed{g}}\\
&+C(v+|u|)^{-2} \sum_{\Gamma \in\{Q,T^2,\Phi^2\}} \int_{S^2_{u,v}}(L\Gamma \Phi  \phi)^2\,d\mu_{\slashed{g}},\\
I_2:=&\:C(v+|u|)^{-4}\int_{S^2_{u,v}}\sum_{j_1+j_2\leq 1}|\snabla^2 L^{j_1}\uline{L}^{j_2}\Phi\phi|^2+\sum_{j_1+j_2\leq 2}|\snabla^2 L^{j_1}\uline{L}^{j_2}\phi|^2+\sum_{j_1+j_2\leq 3}|\snabla L^{j_1}\uline{L}^{j_2}\phi|^2\\
&+\sum_{j_1+j_2\leq 4}(L^{j_1}\uline{L}^{j_2}\phi)^2\,d\mu_{\slashed{g}}.
\end{align*}
By applying Corollary \ref{cor:commeestimate}, as in Proposition \ref{pointwiseboundkn}, and moreover the estimates for the $L^2(S^2_{u,v})$ norms of the angular derivatives from Proposition \ref{pointwiseboundkn}, we can further estimate
\begin{equation*}
\begin{split}
I_2 \leq C(v+|u|)^{-4}\int_{S^2_{u,v}}E_{q;5,\epsilon}[\phi]+\sum_{\Gamma \in \{\Phi^2,T^2,Q\}}(E_{q;4,\epsilon}[\Gamma \phi]+E_{q;4,\epsilon}[\Gamma \Phi \phi]).
\end{split}
\end{equation*}
Furthermore,
\begin{equation*}
\begin{split}
I_1 \leq &\:C\int_{S^2_{-\infty,v}}\sum_{j_1+j_2\leq 2}(LL^{j_1}\uline{L}^{j_2}\phi)^2+\sum_{j_1+j_2\leq 1} \sum_{\Gamma \in \{\Phi,\Phi^2,T^2,Q\}}(LL^{j_1}\uline{L}^{j_2} \Gamma \phi)^2\,d\mu_{\slashed{g}}\\
&+C(v+|u|)^{-2}\int_{S^2_{-\infty,v}}\sum_{j_1+j_2\leq 2}\left(|\snabla L^{j_1}\uline{L}^{j_2}\phi|^2+|\snabla^2 L^{j_1}\uline{L}^{j_2}\phi|^2\right)\\
&+\sum_{j_1+j_2\leq 1} \sum_{\Gamma \in \{\Phi,\Phi^2,T^2,Q\}}\left(|\snabla L^{j_1}\uline{L}^{j_2}\Gamma \phi|^2+|\snabla^2 L^{j_1}\uline{L}^{j_2}\Gamma \phi|^2 \right)\,d\mu_{\slashed{g}}\\
&+C(v+|u|)^{-4}\left(E_{q;2,\epsilon}[\phi]+\sum_{\Gamma \in \{\Phi,\Phi^2,T^2,Q\}}E_{q;1,\epsilon}[\phi]\right)\\
&+Cv^{-4+(1-s)+\delta}\Bigg(\widetilde{E}_{s,\epsilon;4}[\phi]+\sum_{\Gamma \in \{\Phi,\Phi^2,T^2,Q\}}\widetilde{E}_{s,\epsilon;3}[\Gamma\phi]+\widetilde{E}_{s,\epsilon;3}[\Gamma \Phi \phi]\\
&+C\sum_{\Gamma' \in \{\Phi,\Phi^2,T^2,Q\}}\sum_{\Gamma \in \{\Phi,\Phi^2,T^2,Q\}}\widetilde{E}_{s,\epsilon;2}[\Gamma \Gamma'\phi]\Bigg).
\end{split}
\end{equation*}

We conclude that
\begin{equation*}
\begin{split}
v^{4+(s-1)-\delta}||L\phi||^2_{L^{\infty}(S^2_{u,v})}\leq &\:C\int_{S^2_{-\infty,v}}\sum_{j_1+j_2\leq 2}(LL^{j_1}\uline{L}^{j_2}\phi)^2+\sum_{j_1+j_2\leq 1} \sum_{\Gamma \in \{\Phi,\Phi^2,T^2,Q\}}(LL^{j_1}\uline{L}^{j_2} \Gamma \phi)^2\,d\mu_{\slashed{g}}\\
&+Cv^{2+(s-1)-\delta}\int_{S^2_{-\infty,v}}\sum_{j_1+j_2\leq 2}\left(|\snabla L^{j_1}\uline{L}^{j_2}\phi|^2+|\snabla^2 L^{j_1}\uline{L}^{j_2}\phi|^2\right)\\
&+{E}_{q;6,\epsilon}[\phi]+{E}_{q;6,\epsilon}[\Phi \phi]+\widetilde{E}_{s,\epsilon;6}[\phi]+\widetilde{E}_{s,\epsilon;6}[\Phi \phi].\qedhere
\end{split}
\end{equation*}
\end{proof}
We have now proved Theorem \ref{thm:alphaholdercont}.
\appendix
\section{Energy currents in Kerr--Newman}
\label{sec:bulkest}
We consider a spacetimes $(\mathcal{N},g)$ equipped with a double null foliation $(u,v,\vartheta^1,\vartheta^2)$, such that the metric is given by
\begin{equation}
\label{eq:standmetricdoublenull}
g=-2\Omega^2(u,v)(du\otimes dv+dv\otimes du)+ \sg_{AB}(d\vartheta^A-b^Adv)\otimes (d\vartheta^B-b^Bdv).
\end{equation}
Here, $u,v$ solve the Eikonal equation and the (topological) spheres $(S^2_{u,v},\sg)$ are covered by coordinates $\vartheta^A$, with $A=1,2$, foliate the null hypersurfaces $\{u=\textnormal{const.}\}$ and $\{v=\textnormal{const.}\}$. 

Let
\begin{align*}
L=\:&\partial_v+b^A\partial_A,\\
\underline{L}=\:&\partial_u.
\end{align*}
In the $(L,\underline{L},\partial_{\vartheta_A})$ basis, the metric components are given by
\begin{align*}
g(L,L)=\:&0,\\
g(\underline{L},\underline{L})=\:&0,\\
g(L,\underline{L})=\:&-2\Omega^2.\\
g(L,\partial_{\vartheta_A})=\:&0,\\
g(\underline{L},\partial_{\vartheta_A})=\:&0.
\end{align*}

\begin{lemma}
\label{lm:commvfields}
\begin{align}
\label{eq:covderL}
\nabla_L L=\:&\Omega^{-2}L(\Omega^2)L,\\
\label{eq:covderLbar}
\nabla_{\underline{L}} \underline{L}=\:&\Omega^{-2}\underline{L}(\Omega^2)\underline{L},\\
\label{eq:gcovderL}
g(\nabla_L \underline{L},L)=\:&0,\\
\label{eq:gcovderLbar}
g(\nabla_{\underline{L}}L,\underline{L})=\:&0,\\
g(\nabla_{\underline{L}}L,\partial_A)=\:&-g(L,\nabla_{A}\underline{L}),\\
g(\nabla_{L}\underline{L},\partial_A)=\:&-g(\underline{L},\nabla_{A}L),\\
\label{eq:commLLbar1}
[L,\underline{L}]_A=\:&2g(\underline{L},\nabla_A L)+2\partial_A\Omega^2,\\
\label{eq:commLLbar2}
=\:&-2g(L,\nabla_A \underline{L})-2\partial_A\Omega^2.
\end{align}
\end{lemma}
\begin{proof}
We have that
\begin{align*}
[L,\underline{L}]=\:&-\partial_ub^C\partial_C,\\
[L,\partial_A]=\:&-\partial_Ab^C\partial_C,\\
[\underline{L},\partial_A]=\:&0.
\end{align*}
and
\begin{align*}
[L,\underline{L}]=\:&\nabla_L\underline{L}-\nabla_{\underline{L}}L,\\
[L,\partial_A]=\:&\nabla_L \partial_A-\nabla_AL,
\end{align*}
so
\begin{align*}
g(\nabla_LL,L)=g(\nabla_{\underline{L}}\underline{L},\underline{L})=\:&0,\\
g(\nabla_L L,\partial_A)=-g(L,\nabla_L\partial_A)=-g(L,\nabla_A L)=\:&0,\\
g(\nabla_{\underline{L}}\underline{L},\partial_A)=-g(\underline{L},\nabla_{\underline{L}}\partial_A)=\:&0.
\end{align*}
We obtain
\begin{align*}
g(\nabla_LL,\underline{L})=L(g(L,\underline{L})-g(L,\nabla_L\underline{L})=\:&-2L(\Omega^2),\\
g(\nabla_{\underline{L}}\underline{L},L)=\:&-2\underline{L}(\Omega^2).
\end{align*}
The equations (\ref{eq:covderL}) and (\ref{eq:covderLbar}) now immediately follow.

Furthermore, by the above identities we can rewrite
\begin{equation*}
\begin{split}
[L,\underline{L}]_A=\:&g(\nabla_L\underline{L}-\nabla_{\underline{L}}L,\partial_A)=g(\underline{L},\nabla_{L}\partial_A)-g(L,\nabla_{\underline{L}}\partial_A)\\
=\:&g(\underline{L},\nabla_A L)-g(L,\nabla_A \underline{L})=\partial_A(g(L,\underline{L}))-2g(L,\nabla_A \underline{L})\\
=\:&-2g(L,\nabla_A \underline{L})-2\partial_A\Omega^2.
\end{split}
\end{equation*}
From (\ref{eq:covderL}) and (\ref{eq:covderLbar}) it follows moreover that
\begin{align*}
g(\nabla_L \underline{L},L)=\nabla_L(g(L,\underline{L})-g(\underline{L},\nabla_L L)=-2L(\Omega^2)+2L(\Omega^2)=\:&0,\\
g(\nabla_{\underline{L}}L,\underline{L})=\:&0. 
\end{align*}
Finally, we use that $g(\underline{L},[L,\partial_A])=g(L,[\underline{L},\partial_A])=0$ to obtain,
\begin{align*}
g(\nabla_{\underline{L}}L,\partial_A)=-g(L,\nabla_{\underline{L}}\partial_A)=-g(L,\nabla_{A}\underline{L}),\\
g(\nabla_{L}\underline{L},\partial_A)=-g(\underline{L},\nabla_{L}\partial_A)=-g(\underline{L},\nabla_{A}L).  
\end{align*}
\end{proof}

We can write $L=\Omega^2L'$, $\underline{L}=\Omega^2\underline{L}'$, where $L'$ and $\underline{L}'$ are geodesic vector fields, i.e.\
\begin{align*}
\nabla_{L'}L'=\:&0,\\
\nabla_{\underline{L}'}\underline{L'}=\:&0,
\end{align*}
which follows from (\ref{eq:covderL}) and (\ref{eq:covderLbar}).

We can define a renormalised ingoing null vector $e_3$ and outgoing null vector $e_4$, satisfying $g(e_3,e_4)=-2$ by
\begin{align*}
e_3=\:&\Omega^{-1}\partial_u,\\
e_4=\:&\Omega^{-1}(\partial_v+b^A\partial_A).
\end{align*}
The inverse metric in the basis $(e_3,e_4,\partial_A)$
\begin{equation*}
g^{-1}=-\frac{1}{2}(e_3\otimes e_4+e_4\otimes e_3) +(\sg^{-1})^{AB}\partial_A\otimes \partial_B
\end{equation*}
can therefore be expressed in the double-null coordinate basis as
\begin{equation*}
g^{-1}=-\frac{1}{2}\Omega^{-2}(u,v)(\partial_u\otimes \partial_v+\partial_v\otimes \partial_u)-\frac{1}{2}\Omega^{-2} b^A (\partial_u\otimes \partial_A+\partial_A\otimes \partial_u)+(\sg^{-1})^{AB}\partial_A\otimes \partial_B.
\end{equation*}
With respect to the basis $(L,\underline{L},\partial_{\vartheta_A})$ the inverse metric is given by
\begin{equation*}
g^{-1}=-\frac{1}{2}\Omega^{-2}(u,v)(L\otimes\underline{L}+\underline{L}\otimes L)+(\sg^{-1})^{AB}\partial_A\otimes \partial_B.
\end{equation*}

Define the \emph{second fundamental forms} $\chi_{AB}$ and $\underline{\chi}$ by
\begin{align*}
\chi_{AB}&:=g(\nabla_{\partial_A} e_4,\partial_B)=\Omega^{-1}g(\nabla_{\partial_A} L,\partial_B),\\
\underline{\chi}_{AB}&:=g(\nabla_{\partial_A} e_3,\partial_B)=\Omega^{-1}g(\nabla_{\partial_A} \underline{L},\partial_B).
\end{align*}

\begin{lemma}
\label{lm:exprchi}
We can express,
\begin{align}
\label{eq:chiid1}
L(\slashed{g}_{AB})=\:&2\Omega\chi_{AB}-\partial_Ab^C\slashed{g}_{CB}-\partial_Bb^C\slashed{g}_{CA},\\
\label{eq:chiid2}
\underline{L}(\slashed{g}_{AB})=\:&2\Omega\underline{\chi}_{AB}.
\end{align}
Moreover,
\begin{align}
\label{eq:trchiid1}
\frac{L(\sqrt{\det \slashed{g}})}{\sqrt{\det \slashed{g}}}=\:&\Omega \tr \chi-\partial_Cb^C,\\
\label{eq:trchiid2}
\frac{\underline{L}(\sqrt{\det \slashed{g}})}{\sqrt{\det \slashed{g}}}=\:&\Omega \tr \underline{\chi}.
\end{align}
\end{lemma}
\begin{proof}
We use the expression for $[L,\partial_A]$ in Lemma \ref{lm:commvfields} to obtain
\begin{equation*}
\begin{split}
L(\slashed{g}_{AB})=\:&g(\nabla_L \partial_A,\partial_B)+g(\partial_A,\nabla_L\partial_B)\\
=\:&g(\nabla_AL,\partial_B)+g(\partial_A,\nabla_B L)--\partial_Ab^C\slashed{g}_{CB}-\partial_Bb^C\slashed{g}_{CA}\\
=\:&2\Omega\chi_{AB}-\partial_Ab^C\slashed{g}_{CB}-\partial_Bb^C\slashed{g}_{CA},
\end{split}
\end{equation*}
where we used that $\chi_{AB}=\chi_{BA}$, which can easily be shown. Equation (\ref{eq:chiid2}) can be proved similarly.

We can apply the chain rule to obtain
\begin{equation*}
L(\det \slashed{g})=\frac{\partial\det \slashed{g}}{\partial \slashed{g}_{AB}}L(\slashed{g}_{AB}).
\end{equation*}
By Laplace's formula for the determinant of a matrix, we can express
\begin{equation*}
(\det \slashed{g}) \delta^B_C=\slashed{g}_{AB}\textnormal{Adj}(\slashed{g})^{BC},
\end{equation*}
where $\textnormal{Adj}(\slashed{g})^{BC}$ are the components of the adjugate matrix of $\slashed{g}_{AB}$. Consequently,
\begin{equation*}
\frac{\partial\det \slashed{g}}{\partial \slashed{g}_{AB}}=\textnormal{Adj}(\slashed{g})^{AB}=\det \slashed{g}\slashed{g}^{AB},
\end{equation*}
so
\begin{equation*}
L(\det \slashed{g})=(\det \slashed{g})\slashed{g}^{AB}L(\slashed{g}_{AB})
\end{equation*}
We can therefore conclude that
\begin{equation*}
\frac{L(\sqrt{\det \slashed{g}})}{\sqrt{\det \slashed{g}}}=\frac{1}{2}\frac{L(\det \slashed{g})}{\det \slashed{g}}=\Omega\tr \chi-\partial_Cb^C.
\end{equation*}
Equation (\ref{eq:trchiid2}) can be proved similarly.
\end{proof}

We introduce additional metric derivatives,
\begin{align*}
\omega&:=-\frac{1}{4}g(\nabla_{e_4}e_3,e_4),\\
\underline{\omega}&:=-\frac{1}{4}g(\nabla_{e_3}e_4,e_3),\\
\zeta_A&:=\frac{1}{2}g(\nabla_{A}e_4,e_3).
\end{align*}
$\zeta_A$ are the components of the \emph{torsion tensor}.
\begin{lemma}
\label{lm:addriccicoeff}
We can express $\omega$, $\underline{\omega}$ and $\zeta_A$ as follows:
\begin{align*}
4\Omega \omega=\:&\Omega^{-2}L(\Omega^2),\\
4\Omega \underline{\omega}=\:&\Omega^{-2}\underline{L}(\Omega^2),\\
\zeta_A=\:&\frac{1}{4}\Omega^{-2}[L,\underline{L}]_A=-\frac{1}{4}\Omega^{-2}\slashed{g}_{AB}\partial_ub^B.
\end{align*}
\end{lemma}
\begin{proof}
We have that
\begin{equation*}
\begin{split}
4\omega=\:&-\Omega^{-2}g(\nabla_L(\Omega^{-1}\underline{L}),L)=-\Omega^{-2}L(\Omega^{-1})g(L,\underline{L})-\Omega^{-3}g(\nabla_L\underline{L},L)\\
=\:&\frac{1}{2}\Omega^{-5}L(\Omega^2)g(L,\underline{L})-\Omega^{-3}g(\underline{L},\nabla_LL)\\
=\:&\Omega^{-3}L(\Omega^2),
\end{split}
\end{equation*}
where we used (\ref{eq:covderL}) in the last equality. The expression for $\underline{\omega}$ follows similarly.

Moreover, by (\ref{eq:commLLbar1}),
\begin{equation*}
\begin{split}
\zeta_A=\:&\frac{1}{2}\Omega^{-1}g(\nabla_A(\Omega^{-1}L),\underline{L})\\
=\:&\frac{1}{2}\Omega^{-2}g(\nabla_AL,\underline{L})-\frac{1}{4}\Omega^{-4}\partial_A\Omega^2 g(L,\underline{L})\\
=\:&\frac{1}{4}\Omega^{-2}[L,\underline{L}]_A\\
=\:&-\frac{1}{4}\Omega^{-2}\slashed{g}_{AB}\partial_ub^B\qedhere
\end{split}
\end{equation*}
\end{proof}
Consider the weighted vector field
\begin{equation*}
N=N^LL+N^{\underline{L}} \underline{L}=N^{u}\partial_u+N^{v}\partial_v+N^{A}\partial_A,
\end{equation*}
where $N^L=N^L(u,v)=N^v(u,v)$, $N^{\underline{L}}=N^{\underline{L}}(u,v)=N^u(u,v)$ and $N^A=b^AN^{L}=b^A N^v$. The corresponding compatible current $K^N$ is given by
\begin{equation*}
K^N[\phi]:=T_{\alpha\beta}[\phi]\prescript{N}{}\pi^{\alpha\beta},
\end{equation*}
with the components of the deformation tensor $\prescript{N}{}\pi_{\alpha\beta}=\frac{1}{2}[\nabla_{\alpha}N_{\beta}+\nabla_{\beta}N_{\alpha}]$ given by
\begin{align*}
\prescript{N}{}\pi^{LL}=g^{L\underline{L}}g^{L\underline{L}}\prescript{N}{}\pi_{\underline{L}\underline{L}}=\:&\frac{1}{4}\Omega^{-4}g(\nabla_{\underline{L}} N,\underline{L}),\\
\prescript{N}{}\pi^{\underline{L}\underline{L}}=g^{L\underline{L}}g^{L\underline{L}}\prescript{N}{}\pi_{LL}=\:&\frac{1}{4}\Omega^{-4}g(\nabla_L N,L),\\
\prescript{N}{}\pi^{LA}=\:&g^{L\underline{L}}\slashed{g}^{AB}\prescript{N}{}\pi_{\underline{L}B},\\
\prescript{N}{}\pi^{\underline{L}A}=\:&g^{L\underline{L}}\slashed{g}^{AB}\prescript{N}{}\pi_{LB},\\
\prescript{N}{}\pi^{AB}=\:&\frac{1}{2}(\slashed{g}^{AC}\partial_CN^B+\slashed{g}^{BC}\partial_CN^A+\slashed{g}^{AC}\slashed{g}^{BD}(N^{u}\partial_u+N^v\partial_v+N^E\partial_E)\slashed{g}_{CD}).
\end{align*}
We use equations (\ref{eq:covderL})-(\ref{eq:commLLbar2}) to obtain
\begin{align*}
\prescript{N}{}\pi^{LL}=\:&\frac{1}{4}\Omega^{-4}g(\nabla_{\underline{L}} N,\underline{L})=-\frac{1}{2}\Omega^{-2}\underline{L}(N^{L}),\\
\prescript{N}{}\pi^{\underline{L}\underline{L}}=\:&\frac{1}{4}\Omega^{-4}g(\nabla_{L} N,L)=-\frac{1}{2}\Omega^{-2}L(N^{\underline{L}}),\\
\prescript{N}{}\pi^{L\underline{L}}=\:&\frac{1}{8}\Omega^{-4}\left[g(\nabla_L N,\underline{L})+g(\nabla_{\underline{L}} N,L)\right]\\
=\:&-\frac{1}{4}\Omega^{-2}\left[L(N^L)+\underline{L}(N^{\underline{L}})+\Omega^{-2}N^LL(\Omega^2)+\Omega^{-2}N^{\underline{L}}\underline{L}(\Omega^2)\right],\\
\prescript{N}{}\pi^{LA}=\:&-\frac{1}{4}\Omega^{-2}\slashed{g}^{AB}\left[g(\nabla_B N,\underline{L})+g(\nabla_{\underline{L}}N,\partial_B)\right],\\
=\:&-\frac{1}{4}N^L\Omega^{-2}\slashed{g}^{AB}[g(\nabla_B L,\underline{L})-g(L,\nabla_B\underline{L})]=-\frac{1}{4}N^L\Omega^{-2}[L,\underline{L}]^A\\
=\:&-\frac{1}{4}N^L\Omega^{-2}\partial_ub^A,\\
\prescript{N}{}\pi^{\underline{L}A}=\:&-\frac{1}{4}N^{\underline{L}}\Omega^{-2}[\underline{L},L]^A,\\
=\:&\frac{1}{4}N^{\underline{L}}\Omega^{-2}\partial_ub^A,\\
\prescript{N}{}\pi^{AB}=\:&\frac{1}{2}\slashed{g}^{AC}\slashed{g}^{BD}(N^{u}\partial_{u}+N^v\partial_v)\slashed{g}_{CD}+\frac{1}{2}\slashed{g}^{AC}\slashed{g}^{BD}N^Lb^E\partial_E\slashed{g}_{CD}+\frac{1}{2}\left(\slashed{g}^{AC}\partial_CN^B+\slashed{g}^{BC}\partial_CN^A\right)\\
=\:&\frac{1}{2}\slashed{g}^{AC}\slashed{g}^{BD}(N^{\underline{L}}\underline{L}+N^LL)\slashed{g}_{CD}+\frac{1}{2}\left(\slashed{g}^{AC}\partial_CN^B+\slashed{g}^{BC}\partial_CN^A\right).
\end{align*}

Now consider the wave equation (\ref{eq:waveqkerr}) on an extremal Kerr--Newman background in a double-null foliation, with a corresponding energy momentum tensor
\begin{equation*}
\mathbf{T}_{\alpha\beta}[\phi]=\partial_{\alpha}\phi\partial_{\beta}\phi-\frac{1}{2}g_{\alpha\beta}(g^{\lambda \kappa}\partial_{\lambda}\phi\partial_{\kappa}\phi).
\end{equation*}
We have that
\begin{equation*}
\begin{split}
g^{\lambda \kappa}\partial_{\lambda}\phi\partial_{\kappa}\phi=\:&g^{L\underline{L}}L\phi\underline{L}\phi+\slashed{g}^{AB}\partial_A\phi\partial_B\phi\\
=\:&-\frac{1}{2}\Omega^{-2}L\phi\underline{L}\phi+|\snabla\phi|^2.
\end{split}
\end{equation*}
Now we obtain the components
\begin{align*}
\mathbf{T}(L,L)=\:&(L\phi)^2,\\
\mathbf{T}(\underline{L},\underline{L})=\:&(\underline{L}\phi)^2\\
\mathbf{T}(L,\underline{L})=\:&\Omega^2|\snabla\phi|^2,\\
\mathbf{T}(L,\partial_A)=\:&L\phi\partial_A\phi,\\
\mathbf{T}(\underline{L},\partial_A)=\:&\underline{L}\phi\partial_A\phi,\\
\mathbf{T}(\partial_A,\partial_B)=\:&(\partial_A\phi)(\partial_B\phi)+\frac{1}{2}\slashed{g}_{AB}(\Omega^{-2}L\phi\underline{L}\phi-|\snabla\phi|^2).
\end{align*}

Therefore, by using the expressions in Lemma \ref{lm:exprchi} and \ref{lm:addriccicoeff}, we obtain
\setlength{\abovedisplayskip}{0.5cm}
\setlength{\belowdisplayskip}{0.5cm}
\begin{equation*}
\begin{split}
K^N=\:&\prescript{N}{}\pi^{LL}\mathbf{T}(L,L)+\prescript{N}{}\pi^{\underline{L}\underline{L}}\mathbf{T}(\underline{L},\underline{L})+2\prescript{N}{}\pi^{L\underline{L}}\mathbf{T}(L,\underline{L})+\prescript{N}{}\pi^{AB}\mathbf{T}(\partial_A,\partial_B)\\
&+2\prescript{N}{}\pi^{LA}\mathbf{T}(L,\partial_A)+2\prescript{N}{}\pi^{\underline{L}A}\mathbf{T}(\underline{L},\partial_A)\\
=\:&-\frac{1}{2}\Omega^{-2}\underline{L}(N^L)(L\phi)^2-\frac{1}{2}\Omega^{-2}L(N^{\underline{L}})(\underline{L}\phi)^2\\
&+\frac{1}{4}\Omega^{-2}[\slashed{g}^{AB}N^{L}L(\slashed{g}_{AB})+N^{\underline{L}}\underline{L}(\slashed{g}_{AB})+2N^L\partial_Eb^E]L\phi\underline{L}\phi\\
&-\frac{1}{2}\left[\uline{L}(N^{\uline{L}})+L(N^L)+\Omega^{-2}(N^{\uline{L}}\uline{L}(\Omega^2)+N^L L(\Omega^2))\right]|\snabla\phi|^2\\
&+\left[\frac{1}{2}\slashed{g}^{AC}\slashed{g}^{BD}(N^LL-N^{\underline{L}}\uline{L})\slashed{g}_{CD}+N^L\slashed{g}^{AC}\partial_Cb^B\right](\partial_A\phi)(\partial_B\phi)\\
&-\frac{1}{4}\left[\slashed{g}^{BD}(N^LL-N^{\underline{L}}\uline{L})\slashed{g}_{BD}+2N^L\partial_Eb^E\right]|\snabla\phi|^2\\
&-\frac{1}{2}N^L\Omega^{-2}\partial_ub^AL\phi\partial_A\phi\\
&+\frac{1}{2}N^{\underline{L}}\Omega^{-2}\partial_ub^A\underline{L}\phi\partial_A\phi\\
=\:&-\frac{1}{2}\Omega^{-2}\uline{L}(N^L)(L\phi)^2-\frac{1}{2}\Omega^{-2}L(N^{\underline{L}})(\underline{L}\phi)^2+\frac{1}{2}\Omega^{-2}(N^L\Omega\tr \chi+N^{\underline{L}}\Omega \tr \underline{\chi})L\phi\underline{L}\phi\\
&-\frac{1}{2}\left[\uline{L}(N^{\underline{L}})+L(N^L)+4\Omega(N^L\omega+N^{\underline{L}}\underline{\omega})\right]|\snabla\phi|^2+\left[N^L\Omega \hat{\chi}^{AB}+N^{\underline{L}}\Omega\hat{\underline{\chi}}^{AB}\right](\partial_A\phi)(\partial_B\phi)\\
&+2[N^L(L\phi)-N^{\underline{L}}(\underline{L}\phi)]\zeta^A\partial_A\phi,
\end{split}
\end{equation*}
where we used the notation

\begin{align*}
\hat{\chi}_{AB}&:=\chi_{AB}-\frac{1}{2}\slashed{g}_{AB}\tr\chi,\\
\hat{\underline{\chi}}_{AB}&:=\underline{\chi}_{AB}-\frac{1}{2}\slashed{g}_{AB}\tr\underline{\chi}.
\end{align*}

\section{The wave equation in double-null coordinates}
\label{app:treq}
Consider the extremal Kerr--Newman metric in Eddington--Finkelstein-type double-null coordinates $(u,v,{\theta}_*,{\varphi}_*)$. Then the wave operator becomes
\begin{equation}
\label{eq:waveequationdn}
\begin{split}
\square_g\phi=\:&\frac{1}{2\sqrt{\det \slashed{ g}}}\Omega^{-2}\partial_{\alpha}\left(g^{\alpha\beta}2 \Omega^2\sqrt{\det \slashed{g}}\partial_{\beta}\phi\right)\\
=\:&\frac{1}{2\sqrt{\det \slashed{ g}}}\Omega^{-2} \Bigg[-\partial_v\left(\sqrt{\det \slashed{g}} \partial_u\phi\right)-\partial_u \left(\sqrt{\det \slashed{g}}(\partial_v\phi+b^A\partial_A\phi)\right)-\partial_A\left(\sqrt{\det \slashed{g}}b^A\partial_u\phi\right)\\
&+\partial_A\left(\sqrt{\det \slashed{g}}2\Omega^2 \slashed{g}^{AB} \partial_B\phi\right)\Bigg]\\
=\:&\frac{1}{2\sqrt{\det \slashed{ g}}}\Omega^{-2}\left[-L\left(\sqrt{\det \slashed{g}} \uline{L}\phi\right)-\sqrt{\det \slashed{g}}\partial_Ab^A \uline{L}\phi-\uline{L}\left(\sqrt{\det \slashed{g}} L\phi\right)\right] \\
&+ \Omega^{-2}\slashed{g}^{AB}\partial_A\Omega^2\partial_B\phi +\slashed{\Delta}\phi\\
=\:&-\frac{1}{2}\Omega^{-2} \left[L\uline{L}\phi+\uline{L}L\phi+\Omega \tr \chi \uline{L}\phi+\Omega \tr \uline{\chi} L\phi\right]+\Omega^{-2} \snabla\Omega^2 \cdot \snabla \phi +\slashed{\Delta}\phi.
\end{split}
\end{equation}

We can write
\begin{equation}
\label{eq:finalexpressionwaveeq}
\begin{split}
2\Omega^2\square_g\phi=\:&-\underline{L}L\phi-L\underline{L}\phi-\Omega \tr \underline{\chi}L\phi-\Omega \tr \chi \underline{L}\phi+2 \slashed{\nabla}\Omega^2\cdot \slashed{\nabla}\phi+2\Omega^2\slashed{\Delta}\phi.
\end{split}
\end{equation}

Moreover, we use that $[L,\underline{L}]=L(b^A)\partial_A=4\Omega^2 \zeta^A$ and $\Omega \tr \uline{\chi}=(\det \slashed{g})^{-\frac{1}{2}}\uline{L}((\det \slashed{g})^{\frac{1}{2}})$, to write
\begin{equation*}
\begin{split}
2\Omega^2\square_g\phi=\:&-2\underline{L}L\phi-4\Omega^2 \zeta(\phi)-\frac{\uline{L}(\sqrt{\det \slashed{g}})}{\sqrt{\det \slashed{g}}}-\Omega \tr \chi \underline{L}\phi+2 \slashed{\nabla}\Omega^2\cdot \slashed{\nabla}\phi+2\Omega^2\slashed{\Delta}\phi.
\end{split}
\end{equation*}
The wave equation in the above form can easily be rewritten as a transport equation for $(\det \slashed{g})^{\frac{1}{4}}L\phi$:
\begin{equation}
\label{eq:transpeq1}
2\uline{L}((\det \slashed{g})^{\frac{1}{4}}L\phi))=(\det \slashed{g})^{\frac{1}{4}} \Omega \tr \chi \uline{L}\phi+2 (\det \slashed{g})^{\frac{1}{4}} \snabla \Omega^2\cdot \snabla \phi-4(\det \slashed{g})^{\frac{1}{4}}\Omega^2\zeta(\phi)+2(\det \slashed{g})^{\frac{1}{4}}\Omega^2\slashed{\Delta}\phi.
\end{equation}
\bibliography{mybib2}
\bibliographystyle{amsplaininitials}

\end{document}